\RequirePackage[log]{snapshot}		
\RequirePackage[l2tabu,orthodox]{nag}		
\documentclass[reqno]{amsart}
\usepackage[margin=1.5in,bottom=1.25in]{geometry}		

\usepackage{amsmath}		
\usepackage{amssymb}		
\usepackage{amsfonts}		
\usepackage{amsthm}		
\usepackage[foot]{amsaddr}		

\usepackage{mathtools}		
\mathtoolsset{%
}
\usepackage[utf8]{inputenc}		
\usepackage[T1]{fontenc}		

\usepackage[proportional,tabular,lining,sf,mono=false]{libertine}

\usepackage{dsfont}		

\usepackage[
cal=cm,
]
{mathalfa}

\usepackage{acronym}		
\renewcommand*{\aclabelfont}[1]{\acsfont{#1}}		
\newcommand{\acli}[1]{\emph{\acl{#1}}}		
\newcommand{\acdef}[1]{\emph{\acl{#1}} \textup{(\acs{#1})}\acused{#1}}		

\usepackage[labelfont={bf,small},labelsep=colon,font=small]{caption}	
\usepackage[svgnames]{xcolor}		
\colorlet{MyRed}{Crimson!90!Black}
\colorlet{MyGreen}{DarkGreen!80!Black}
\colorlet{MyBlue}{MediumBlue}

\newcommand{\afterhead}{.\;}		
\newcommand{\ackperiod}{}		

\newcommand{\para}[1]{\medskip\paragraph{\textbf{#1\afterhead}}}

\usepackage{latexsym}		

\usepackage{pifont}
\newcommand{\cmark}{\checkmark}		
\newcommand{\xmark}{\ding{53}}

\usepackage{subcaption}		
\usepackage{tikz}
\usetikzlibrary{calc,patterns,positioning}		

\usepackage{array}		
\usepackage{booktabs}		
\usepackage[inline,shortlabels]{enumitem}		
\setenumerate{itemsep=\smallskipamount,topsep=\smallskipamount,leftmargin=2em,label={\itshape \roman*\hspace*{.5pt}\upshape)}}		
\setitemize{itemsep=\smallskipamount,topsep=\smallskipamount,leftmargin=2em}		

\usepackage[kerning=true]{microtype}		

\usepackage{xspace}		

\usepackage[linesnumbered,ruled,vlined]{algorithm2e}

\SetCommentSty{mycommfont} 
\SetAlgoSkip{bigskip} 

\theoremstyle{plain}
\newtheorem{theorem}{Theorem}		
\newtheorem{proposition}{Proposition}		

\newtheorem*{corollary*}{Corollary}		

\theoremstyle{definition}
\newtheorem{definition}{Definition}		
\newtheorem{assumption}{Assumption}		
\newtheorem{example}{Example}		

\newtheorem*{definition*}{Definition}		
\newtheorem*{assumption*}{Assumptions}		
\newtheorem*{example*}{Example}		

\theoremstyle{remark}
\newtheorem{remark}{Remark}		

\newtheorem*{remark*}{Remark}		

\def\endenv{\hfill\S}

\newenvironment{Proof}[1][Proof]{\begin{proof}[#1]}{\end{proof}}

\newcounter{proofpart}

\theoremstyle{TH}
\newtheorem{informal}{Informal Version of Theorem}

\numberwithin{example}{section}		

\usepackage[authoryear,sort&compress]{natbib}		
\renewcommand{\bibfont}{\footnotesize}%

\bibpunct[, ]{[}{]}{,}{n}{}{,}%

\usepackage{hyperref}
\hypersetup{
colorlinks=true,
linktocpage=true,
pdfstartview=FitH,
breaklinks=true,
pdfpagemode=UseNone,
pageanchor=true,
pdfpagemode=UseOutlines,
plainpages=false,
bookmarksnumbered,
bookmarksopen=false,
bookmarksopenlevel=1,
hypertexnames=true,
pdfhighlight=/O,
urlcolor=MyBlue,linkcolor=MyBlue,citecolor=MyBlue,	
pdftitle={},
pdfauthor={},
pdfsubject={},
pdfkeywords={},
pdfcreator={pdfLaTeX},
pdfproducer={LaTeX with hyperref}
}

\newcommand{\EMAIL}[1]{\email{\href{mailto:#1}{#1}}}

\usepackage[sort&compress,capitalize,nameinlink]{cleveref}		
\crefname{algorithm}{Algorithm}{Algorithms}
\crefname{equation}{Equation}{Equations}
\crefname{assumption}{Assumption}{Assumptions}
\crefname{app}{Appendix}{Appendices}
\crefformat{footnote}{#2\footnotemark[#1]#3}		

\usepackage[showdeletions,suppress]{color-edits}		
\usepackage[normalem]{ulem}		
\newcommand{\debug}[1]{#1}		

\newcommand{\revise}[1]{#1}		

\def\endedit{\color{black}}		

\newcommand{\newmacro}[2]{\newcommand{#1}{\debug{#2}}}		
\newcommand{\newop}[2]{\DeclareMathOperator{#1}{\debug{#2}}}		

\DeclarePairedDelimiter{\braces}{\{}{\}}		
\DeclarePairedDelimiter{\bracks}{[}{]}		
\DeclarePairedDelimiter{\parens}{(}{)}		

\DeclarePairedDelimiter{\abs}{\lvert}{\rvert}		

\DeclarePairedDelimiterX{\setdef}[2]{\{}{\}}{#1:#2}		
\DeclarePairedDelimiterXPP{\exclude}[1]{\mathopen{}\setminus}{\{}{\}}{}{#1}

\newcommand{\R}{\mathbb{R}}		

\DeclareMathOperator*{\argmin}{arg\,min}		
\DeclareMathOperator*{\union}{\bigcup}		

\DeclareMathOperator{\bigoh}{\mathcal{O}}		
\DeclarePairedDelimiterXPP{\bigof}[1]{\bigoh}{(}{)}{}{#1}		

\newmacro{\diam}{D}		
\DeclareMathOperator{\dist}{dist}		
\DeclareMathOperator{\ess}{ess}		
\DeclareMathOperator{\grad}{\nabla}		
\DeclareMathOperator{\one}{\mathds{1}}		
\DeclareMathOperator{\relint}{ri}		

\newcommand{\cf}{cf.\xspace}		
\newcommand{\eg}{e.g.,\xspace}		
\newcommand{\ie}{i.e.,\xspace}		
\newcommand{\wrt}{w.r.t\xspace}		

\newcommand{\textpar}[1]{\textup(#1\textup)}		

\newcommand{\txs}{\textstyle}		

\newcommand{\alt}[1]{#1'}		
\newcommand{\altalt}[1]{#1''}		

\newmacro{\dd}{\:d}		
\newcommand{\pd}{\partial}		

\newcommand{\insum}{\sum\nolimits}		

\newmacro{\const}{{A}}		
\newmacro{\constalt}{{B}}		
\newmacro{\constdual}{\const_{\texttt{dual}}}		
\newmacro{\conststoch}{\const_{\texttt{stoch}}}		

\newmacro{\multi}{\texttt{H}}		

\newmacro{\coef}{\lambda}		
\newmacro{\param}{\theta}		
\newmacro{\params}{\Theta}		

\newmacro{\pexp}{r}		
\newmacro{\qexp}{q}		
\newmacro{\rexp}{r}		

\newmacro{\beforestart}{0}		
\newmacro{\start}{1}		
\newmacro{\afterstart}{2}		
\newmacro{\running}{\start,\afterstart,\dotsc\,}		

\newmacro{\run}{t}		
\newmacro{\runprev}{\run-1}		
\newmacro{\runalt}{s}		
\newmacro{\runaltalt}{\alt \run}		
\newmacro{\nRuns}{T}		
\newmacro{\runs}{\mathcal{\nRuns}}		

\newmacro{\state}{f}		
\newmacro{\statealt}{\alt\state}		
\newmacro{\disstate}{x}		

\newcommand{\beforeinit}[1][\state]{\debug{#1}^{\beforestart}}		
\newcommand{\init}[1][\state]{\debug{#1}^{\start}}		
\newcommand{\afterinit}[1][\state]{\debug{#1}^{\afterstart}}		

\newcommand{\iter}[1][\state]{\debug{#1}^{\runalt}}		

\newcommand{\prev}[1][\state]{\debug{#1}^{\run-1}}		
\newcommand{\curr}[1][\state]{\debug{#1}^{\run}}		
\renewcommand{\next}[1][\state]{\debug{#1}^{\run+1}}		

\newcommand{\beforelast}[1][\state]{\debug{#1}^{\nRuns-1}}		
\newcommand{\last}[1][\state]{\debug{#1}^{\nRuns}}		

\newop{\Eq}{Eq}		
\newop{\Nash}{NE}		

\newop{\gap}{Gap}
\newop{\brep}{br}		
\newop{\reg}{\overline{\regstoch}}		
\newop{\regstoch}{\mathcal{R}}		
\newop{\preg}{Reg}		

\newop{\val}{val}		
\newmacro{\play}{i}		
\newmacro{\playalt}{j}		
\newmacro{\playaltalt}{k}		
\newmacro{\nPlayers}{N}		
\newmacro{\players}{\mathcal{\nPlayers}}		

\newmacro{\pure}{\alpha}		
\newmacro{\purealt}{\beta}		
\newmacro{\purealtalt}{\gamma}		
\newmacro{\nPures}{A}		
\newmacro{\pures}{\mathcal{\nPures}}		

\newmacro{\strat}{x}		
\newmacro{\stratalt}{\alt\strat}		
\newmacro{\strataltalt}{\altalt\strat}		
\newmacro{\strats}{\mathcal{X}}		
\newmacro{\intstrats}{\strats^{\circle}}		

\newmacro{\pay}{u}		
\newmacro{\payv}{v}		
\newmacro{\pot}{\Phi}		
\newmacro{\meanpot}{\pot}		

\newmacro{\game}{\Gamma}		
\newmacro{\meangame}{\game}		
\newmacro{\gameall}{\game(\players,\points,\loss)}		

\newmacro{\fingame}{\Gamma}		
\newmacro{\fingameall}{\Gamma(\players,\pures,\pay)}		

\newmacro{\gmat}{g}		
\newmacro{\gdist}{\dist_{\gmat}}
\newmacro{\mfld}{M}		
\newmacro{\form}{\omega}		

\newmacro{\tvec}{z}		
\newmacro{\uvec}{u}		

\newmacro{\ball}{\mathbb{B}}		
\newmacro{\sphere}{\mathbb{S}}		

\newmacro{\vertex}{v}		
\newmacro{\vertexbase}{v^*}		
\newmacro{\vertexalt}{u}		
\newmacro{\vertexaltalt}{\alt\vertexalt}		
\newmacro{\nVertices}{V}		
\newmacro{\vertices}{\mathcal{\nVertices}}		

\newcommand{\head}[1]{\vertex^{-}_{#1} }
\newcommand{\tail}[1]{\vertex^{+}_{#1}}

\newcommand{\In}[2]{{\text{\ttfamily In}^{#2}_{#1}}} 
\newcommand{\Out}[2]{{\text{\ttfamily Out}^{#2}_{#1}} }

\newcommand{\child}[2]{ \text{\ttfamily Child} ^{#2}_{#1}}
\newcommand{\parent}[2]{\text{\ttfamily Parent} ^{#2}_{#1}}
\newmacro{\vchild}{\alt \vertex}
\newmacro{\vparent}{u}

\newmacro{\edge}{e}		
\newmacro{\edgebase}{e^*}		
\newmacro{\edgealt}{\alt\edge}		
\newmacro{\edgealtalt}{\altalt\edge}		
\newmacro{\nEdges}{E}		
\newmacro{\edges}{\mathcal{\nEdges}}		

\newmacro{\graph}{\mathcal{\graphsize}}		
\newmacro{\graphsize}{G}		
\newmacro{\graphall}{\graph(\vertices,\edges)}		
\newmacro{\length}{K}   

\newmacro{\vecspace}{\mathcal{V}}		
\newmacro{\subspace}{\mathcal{W}}		

\newmacro{\bvec}{e}		
\newmacro{\bvecs}{\mathcal{E}}		

\newmacro{\coord}{k}		
\newmacro{\coordalt}{\coord^{\prime}}		
\newmacro{\coordaltalt}{\coordalt^{\prime}}		
\newmacro{\nCoords}{d}		
\newmacro{\dims}{\nCoords}		
\newmacro{\vdim}{\nCoords}		

\newmacro{\pspace}{\mathcal{X}}		
\newmacro{\dspace}{\mathcal{Y}}		

\newmacro{\ppoint}{\pstate}		

\newmacro{\ppointalt}{\pstatealt}		
\newmacro{\pstate}{z}		
\newmacro{\pstatealt}{\bar\pstate}		
\newmacro{\ppointaltalt}{\altalt\ppoint}		
\newmacro{\pvec}{u}		
\newmacro{\pstateavg}{\bar\state}		
\newmacro{\signalavg}{\bar\signal}		

\newmacro{\dstate}{w}		
\newmacro{\dvec}{{v}}		
\newmacro{\disdstate}{\hat{\dstate}}		

\newmacro{\Anc}{\textsf{r}}
\newmacro{\anchor}{r}
\newmacro{\Anchor}{\Anc^{\pair}_{\vertex}}
\newmacro{\tildeAnchor}{\tilde{\Anc}^{\pair}_{\vertex}}

\newmacro{\ptest}{\tilde{\pstate}}		
\newmacro{\test}{\tilde{\state}}		
\newmacro{\disqtest}{\tilde{\disflow}}		
\newmacro{\distest}{\tilde{\disflow}}		
\newmacro{\dtest}{\tilde{\drecom}}		
\newmacro{\testsignal}{\tilde{\signal}}		
\newmacro{\precom}{\pstate}	
\newmacro{\drecom}{\dstate}		

\newmacro{\dispstate}{Z}		
\newmacro{\dispstatealt}{\bar\dispstate}		

\newmacro{\dpoint}{w}		

\newmacro{\dpointave}{\bar{\dpoint}}		
\newmacro{\dpointalt}{{W}}		
\newmacro{\dpointaltalt}{\altalt\dpoint}		
\newmacro{\dpoints}{\mathcal{Y}}		

\newmacro{\mat}{M}		
\newmacro{\hmat}{H}		

\newmacro{\ones}{\mathbf{1}}		
\newmacro{\eye}{I}		
\newmacro{\zer}{0}		

\DeclarePairedDelimiter{\norm}{\lVert}{\rVert}		
\DeclarePairedDelimiterXPP{\dnormdef}[1]{}{\lVert}{\rVert}{_{\ast}}{#1}
\DeclarePairedDelimiterXPP{\dnorm}[1]{}{\lVert}{\rVert}{_{\infty}}{#1}

\DeclarePairedDelimiterXPP{\onenorm}[1]{}{\lVert}{\rVert}{_{1}}{#1}		
\DeclarePairedDelimiterXPP{\twonorm}[1]{}{\lVert}{\rVert}{_{2}}{#1}		
\DeclarePairedDelimiterXPP{\supnorm}[1]{}{\lVert}{\rVert}{_{\infty}}{#1}		

\DeclarePairedDelimiterX{\braket}[2]{\langle}{\rangle}{#1\mathopen{},\mathopen{}#2}

\DeclarePairedDelimiterX{\inner}[2]{\langle}{\rangle}{#1,#2}		
\newmacro{\cartprod}{\bigtimes}

\newcommand{\defeq}{\coloneqq}		

\newcommand{\from}{\colon}		

\newop{\Opt}{Opt}		
\newop{\Sol}{Sol}		
\newop{\orcl}{Or}		

\newmacro{\obj}{g}		
\newmacro{\objalt}{\alt \obj}		
\newmacro{\sobj}{F}		
\newmacro{\func}{\textsl{g}}

\newmacro{\gvec}{g}		
\newmacro{\oper}{A}		
\newmacro{\vecfield}{v}		

\newcommand{\sol}[1][\flow]{#1^{\ast}}		
\newcommand{\sols}{\sol[\flows]}		

\newmacro{\vbound}{M}		
\newmacro{\lips}{L}
\newmacro{\lipsEW}{\lips_{\textrm{EW}}}		

\newmacro{\strong}{\kappa}		
\newmacro{\smooth}{\beta}		
\newmacro{\lagran}{\hat{\lambda}}

\newmacro{\cvx}{\mathcal{K}}		
\newmacro{\subd}{\partial}		

\newmacro{\minmax}{L}		

\newmacro{\minvar}{\theta}		
\newmacro{\minvaralt}{\alt\minvar}		
\newmacro{\minvars}{\Theta}		

\newmacro{\maxvar}{\phi}		
\newmacro{\maxvaralt}{\alt\maxvar}		
\newmacro{\maxvars}{\Phi}		

\newmacro{\hreg}{h}		
\newmacro{\mprox}{P}		
\newmacro{\fench}{F}		
\newmacro{\hstr}{K}		
\newmacro{\depth}{H}		
\newmacro{\proxdom}{\points^{\hreg}}		

\DeclarePairedDelimiterXPP{\proxof}[2]{\mprox_{#1}}{(}{)}{}{#2}		

\newmacro{\zone}{\mathbb{D}}		

\newop{\Eucl}{\Pi}		
\newop{\logit}{\Lambda}		
\DeclarePairedDelimiterXPP{\dkl}[2]{\mathrm{KL}}{(}{)}{}{#1 \nonscript\,\delimsize\|\nonscript\,\mathopen{} #2}
\newmacro{\subgrad}{W}

\newmacro{\point}{z}		
\newmacro{\pointalt}{\alt\point}		
\newmacro{\pointaltalt}{\altalt\point}		
\newmacro{\points}{\mathcal{K}}		
\newmacro{\intpoints}{\relint\points}		

\newmacro{\base}{\point^{\ast}}		
\newmacro{\basealt}{u^{\ast}}		

\newmacro{\real}{x}
\newmacro{\realalt}{\alt\real}
\newmacro{\realaltalt}{\altalt \real}

\newmacro{\open}{\mathcal{U}}		
\newmacro{\closed}{\mathcal{C}}		
\newmacro{\cpt}{\mathcal{K}}		
\newmacro{\nhd}{\mathcal{U}}		

\newop{\ex}{\mathbb{E}}		
\newop{\prob}{\mathbb{P}}		
\newop{\Var}{Var}		
\newop{\simplex}{\Delta}		

\DeclarePairedDelimiterXPP{\exof}[1]{\ex}{[}{]}{}{
	 #1}

\DeclarePairedDelimiterXPP{\probof}[1]{\prob}{(}{)}{}{
	 #1}

\newcommand{\oneof}[1]{\one_{\braces*{#1}}}		

\newmacro{\sample}{\omega}		
\newmacro{\samples}{\Omega}		

\newmacro{\filter}{\mathcal{F}}		
\newmacro{\probspace}{(\samples,\filter,\prob)}		

\newmacro{\event}{E}       
\newmacro{\eventalt}{H}       
\newmacro{\mean}{\mu}		
\newmacro{\sdev}{\sigma}		
\newmacro{\variance}{\sdev^{2}}		

\newmacro{\step}{\alpha}		
\newmacro{\stepalt}{\gamma}		
\newmacro{\stepaltalt}{\theta}		
\newmacro{\stepada}{1}		

\newmacro{\learn}{\eta}		
\newmacro{\dstep}{\psi}		

\newmacro{\proper}{\tau}		

\newmacro{\signal}{C}		
\newmacro{\altsignal}{\bar\signal}		
\newmacro{\error}{Z}		
\newmacro{\bias}{b}		
\newmacro{\brown}{W}		

\newmacro{\serror}{\theta}		
\newmacro{\snoise}{\xi}		
\newmacro{\sbias}{\psi}		

\newmacro{\sbound}{M}		
\newmacro{\bbound}{B}		
\newmacro{\noisepar}{\sdev}		
\newmacro{\noisevar}{\texttt{var}}		

\newcommand{\accelegrad}{\debug{\textsc{AcceleGrad}}\xspace}

\newcommand{\unixgrad}{\debug{\textsc{UnixGrad}}\xspace}

\newcommand{\expweight}{\debug{\textsc{Exp\-Weight}}\xspace}
\newcommand{\acceleweight}{\debug{\textsc{Acce\-le\-Weight}}\xspace}
\newcommand{\adaweight}{\debug{\textsc{Ada\-Weight}}\xspace}
\newcommand{\adapush}{\debug{\textsc{Ada\-Light}}\xspace}
\newcommand{\ppm}{\debug{\textsc{Push\-Pull\-Match}}\xspace}
\newcommand{\scorepush}{\debug{\textsc{Push\-Pull\-Match}}\xspace}

\newmacro{\cost}{c}
\newmacro{\meancost}{\bar\cost}

\newcommand{\source}[1]{O^{#1}}		
\newcommand{\sink}[1]{D^{#1}}		

\newmacro{\pair}{i}		
\newmacro{\pairalt}{j}		
\newmacro{\nPairs}{N}		
\newmacro{\pairs}{\mathcal{\nPairs}}		

\newmacro{\route}{p}		
\newmacro{\routealt}{q}		
\newmacro{\nRoutes}{P}		
\newmacro{\routes}{\mathcal{\nRoutes}}		

\newmacro{\flow}{f}		
\newmacro{\aveplay}{\bar{\state}}		
\newmacro{\flowbase}{\flow^{\ast}}
\newmacro{\flowalt}{\alt\flow}		
\newmacro{\flowaltalt}{\altalt\flow}		
\newmacro{\flows}{\mathcal{F}}		
\newmacro{\flowave}{\bar{\pstate}}		
\newmacro{\flowavehalf}{\state}		
\newmacro{\mirrorvec}{y}		

\newmacro{\load}{\ell}		
\newmacro{\loadalt}{\alt\load}		
\newmacro{\loadaltalt}{\altalt\load}		
\newmacro{\loads}{\mathcal{L}}		

\newmacro{\obs}{\signal}	
\newmacro{\obsave}{\testsignal} 		

\newmacro{\ar}{\mu}		

\newmacro{\late}{c}
\newmacro{\Late}{\barlate}
\newmacro{\latevec}{\vv{\late}}
\newmacro{\dislate}{\hat{\late}}
\newmacro{\barlate}{\bar{\late}}     

\newmacro{\derlate}{\late^{\prime}}
\newmacro{\latency}{\late_{\time, \edge}}		
\newmacro{\latencyTimeAlt}{\late_{\timealt, \edge}}	 
\newmacro{\latencyEdgeAlt}{\late_{\time, \alt\edge}}	
\newmacro{\latencyTimeEdgeAlt}{\late_{\timealt, \alt\edge}}	
\newmacro{\latencyTimeAltalt}{\late_{\altalt\time, \edge}}	 
\newmacro{\latencyEdgeAltalt}{\late_{\time, \altalt\edge}}	

\newmacro{\bderiv}{L}  
\newmacro{\ub}{H}  
\newmacro{\bound}{\ub}  
\newmacro{\noiseB}{\theta}  
\newmacro{\m}{M}
\newcommand{\mass}[1]{\m^{#1}}
\newmacro{\massmax}{\m_{\max}} 
\newmacro{\massmin}{\m_{\min}} 
\newmacro{\masssum}{\m_{\mathrm{tot}}}

\newmacro{\noise}{U}  
\newmacro{\noiseave}{\tilde{\noise}}  
\newmacro{\noisetest}{\tilde{\noise}}  

\newmacro{\mindiff}{\kappa}
\newmacro{\diff}{\xi}

\newmacro{\rat}{\mass}
\newmacro{\rate}{\rat}
\newmacro{\rateTimeAlt}{\rat_{\timealt, \pair}}	 
\newmacro{\ratePairAlt}{\rat_{\time,  \pairalt}}	
\newmacro{\rateTimePairAlt}{\rat_{\timealt, \pairalt}}	 
\newmacro{\rateTimeAltalt}{\rat_{\altalt\time, \pair}}	 
\newmacro{\ratePairAltalt}{\rat_{\time,  \pairaltalt}}	
\newmacro{\ratmax}{\rat_{\max}} 

\newmacro{\SP}{\scorepush}
\newmacro{\SPA}{\textrm{SPAvg}}

\newmacro{\Dpoint}{Y}		
\newmacro{\graphalt}{\tilde{\graph}}		
\newmacro{\edgesalt}{\tilde{\edges}}
\newmacro{\nEdgesalt}{\tilde{\nEdges}}
\newmacro{\verticesalt}{\tilde{\vertices}}
\newmacro{\nVerticesalt}{\tilde{\nVertices}}
\newmacro{\routesalt}{\tilde{\routes}}
\newmacro{\nRoutesalt}{\tilde{\nRoutes}}

\newmacro{\score}{S}
\newmacro{\Weight}{\textsf{w}}
\newmacro{\Weightanchor}{W^{\textrm{anchor}}}

\newcommand{\scorebw}[3]{\texttt{BWscore} ^{#2}_{#1} }
\newcommand{\scorefw}[3]{\texttt{FWscore}^{#2}_{#1} }
\newcommand{\logscore}[1]{\texttt{logscore}_{#1} }
\newcommand{\edgescore}[1]{\texttt{edgescore}_{#1} }
\newmacro{\loadavg}{\bar{\load}}

\newmacro{\grade}{\mathcal{S}}
\newmacro{\gradeave}{\bar{\grade}}
\newmacro{\gradealt}{\grade^{\prime}}

\newmacro{\dpointaltada}{\mathcal{Z}}

\newmacro{\disflow}{x}		
\newmacro{\disflowbase}{\widehat{\disflow}}		
\newmacro{\disflowalt}{\alt\disflow}		
\newmacro{\disflowaltalt}{\altalt\disflow}		
\newmacro{\disflows}{{\mathcal{X}}}

\newmacro{\dppointave}{\prescript{\flowave}{}{\disflow}}		

\newmacro{\disload}{{\load}}		
\newmacro{\disloadalt}{\alt\disload}		
\newmacro{\drecomdis}{\Weight}		
\newmacro{\dtestdis}{\tilde{\drecomdis}}		

\newmacro{\dpointA}{Y}
\newmacro{\dpointAalt}{Z}

\newmacro{\weightDSP}{\bar{w}}
 
\colorlet{DQVcolor}{black}
\addauthor[\textbf{DQV}]{DQV}{DQVcolor}
\newcommand{\DQV}{\DQVmargincomment}

\colorlet{KAcolor}{black}
\addauthor[\textbf{KA}]{KA}{KAcolor}

\colorlet{PMcolor}{black}
\addauthor[\textbf{PM}]{PM}{PMcolor}
\newcommand{\PM}{\PMmargincomment}
\def\beginPM{\color{PMcolor}}		

\begin{document}


\title
[Routing in an Uncertain World]
{Routing in an Uncertain World:\\
Adaptivity, Efficiency, and Equilibrium}

\author
[D. Q. Vu]
{Dong Quan Vu$^{\ast}$}
\address{$^{\ast}$\,%
Univ. Grenoble Alpes, CNRS, Inria, LIG, 38000, Grenoble, France.}
\EMAIL{dong-quan.vu@inria.fr}

\author
[K. Antonakopoulos]
{Kimon Antonakopoulos$^{\ast}$}
\EMAIL{kimon.antonakopoulos@inria.fr}

\author
[P.~Mertikopoulos]
{Panayotis Mertikopoulos$^{\ast,\sharp}$}
\address{$^{\sharp}$\,%
Criteo AI Lab.}
\EMAIL{panayotis.mertikopoulos@imag.fr}



\thanks{
%
%
The authors are grateful for financial support by the French National Research Agency (ANR) in the framework of
the ``Investissements d'avenir'' program (ANR-15-IDEX-02),
the LabEx PERSYVAL (ANR-11-LABX-0025-01),
MIAI@Grenoble Alpes (ANR-19-P3IA-0003),
and the grant ALIAS (ANR-19-CE48-0018-01).
The authors' research was also supported by the COST Action CA16228 ``European Network for Game Theory'' (GAMENET)\ackperiod}

\newacro{dis}{local flow}

\newacro{EW}[\expweight]{exponential weights}
\newacro{AEW}[\adaweight]{adaptive exponential weights}
\newacro{ALW}[\adapush]{adaptive local weights}
\newacro{XLEW}[\acceleweight]{accelerated exponential weights}
\newacro{ACSA}[AC-SA]{accelerated stochastic approximation}
\newacro{UPGD}{universal primal gradient descent}
\newacro{DE}{dual extrapolation}

\newacro{BMW}{Beckmann\textendash McGuire\textendash Winsten}

\newacro{LHS}{left-hand side}
\newacro{RHS}{right-hand side}
\newacro{iid}[i.i.d.]{independent and identically distributed}
\newacro{lsc}[l.s.c.]{lower semi-continuous}
\newacro{NE}{Nash equilibrium}
\newacroplural{NE}[NE]{Nash equilibria}
\newacro{WE}{Wardrop equilibrium}
\newacroplural{WE}[WE]{Wardrop equilibria}
\newacro{OD}[O/D]{origin-destination}
\newacro{DAG}{directed acyclic graph}

\begin{abstract}

\beginPM
We consider the traffic assignment problem in nonatomic routing games where the players' cost functions may be subject to random fluctuations (\eg weather disturbances, perturbations in the underlying network, etc.).
We tackle this problem from the viewpoint of a control interface that makes routing recommendations based solely on observed costs and without any further knowledge of the system's governing dynamics \textendash\ such as the network's cost functions, the distribution of any random events affecting the network, etc.
In this online setting, learning methods based on the popular \acl{EW} algorithm converge to equilibrium at an $\bigof{1/\sqrt{\nRuns}}$ rate:
this rate is known to be order-optimal in stochastic networks, but it is otherwise suboptimal in static networks.
In the latter case, it is possible to achieve an $\bigof{1/\nRuns^{2}}$ equilibrium convergence rate via the use of finely tuned accelerated algorithms;
on the other hand, these accelerated algorithms fail to converge altogether in the presence of persistent randomness, so it is not clear how to achieve the ``best of both worlds'' in terms of convergence speed.
Our paper seeks to fill this gap by proposing an adaptive routing algortihm with the following desirable properties:
\begin{enumerate*}
[(\itshape i\hspace*{.5pt}\upshape)]
\item
it seamlessly interpolates between the $\bigoh(1/\nRuns^{2})$ and $\bigoh(1/\sqrt{\nRuns})$ rates for static and stochastic environments respectively;
\item
its convergence speed is polylogarithmic in the number of paths in the network;
\item
the method's per-iteration complexity and memory requirements are both linear in the number of nodes and edges in the network;
and
\item
it does not require any prior knowledge of the problem's parameters.
\end{enumerate*}
\endedit

\end{abstract}
\maketitle
\acresetall		

\maketitle
\thispagestyle{empty}

\acresetall		
\allowdisplaybreaks		

\section{Introduction}
\label{sec:intro}

\beginPM
Transportation networks in major metropolitan areas carry several million car trips and commutes per day, giving rise to a chaotic and highly volatile environment for the average commuter.
As a result, navigation apps like Google Maps, Waze and MapQuest have seen an explosive growth in their user base, routinely receiving upwards of $10^{4}$ routing requests per second during rush hour \textendash\ and going up to $10^{5}$ queries/second in the largest cities in the US and China \citep{cabannes2019regrets}.
\endedit
This vast number of users must be routed efficiently, in real-time, and without causing any ``ex-post'' regret at the user end;
otherwise, if a user could have experienced better travel times along a non-recommended route, they would have no incentive to follow the app recommendation in the first place.
In the language of congestion games \citep{NRTV07}, this requirement is known as a ``Wardrop equilibrium'', and it is typically represented as a high-dimensional vector describing the traffic flow along each path in the network \citep{War52}.

\beginPM
Ideally, this traffic equilibrium should be computed \emph{before} making a routing recommendation
in order to minimize the number of disgruntled users in the system.
In practice however, this is rarely possible:
the state of the network typically depends on random, unpredictable factors that may vary considerably from one epoch to the next \textendash\ \eg due to road incidents, rain, fog and/or other weather conditions \textendash\ 
so it is generally unrealistic to expect that such a recommendation can be made in advance.
Instead,
it is more apt to consider an online recommendation paradigm that unfolds as follows:
\begin{enumerate}
[1.]
\item
At each time slot $\run=\running,\nRuns$, a control interface \textendash\ such as Google Maps \textendash\ determines a traffic assignment profile and routes all demands received within this slot according to this profile.
\item
The interface observes and records the travel times of the network's users within the time slot in question;
based on this feedback, it updates its candidate profile for the next epoch and the process repeats.
\end{enumerate}

Of course, the crux of this paradigm is the optimization algorithm used to update traffic flow profiles from one epoch to the next.
Any such algorithm would have to satisfy the following \emph{sine qua non} requirements:
\begin{enumerate}

\item
\emph{Universal convergence rate guarantees in terms of the number of epochs:}
The distance from equilibrium of the candidate traffic assignment after $\nRuns$ epochs should be as small as possible in terms of $\nRuns$.
However, learning methods that are well-suited for rapidly fluctuating environments may be too slow in static environments (\ie when cost functions do not vary over time);
conversely, methods that are optimized for static environments may fail to converge in the presence of randomness.
As a result, it is crucial to employ \emph{universal} algorithms that achieve the ``best of both worlds'' in terms of convergence speed.

\item
\emph{Fast convergence in terms of the size of the network:}
The algorithm's convergence speed must be at most polynomial in the size $\graphsize$ of the underlying graph (\ie its number of edges $\nEdges$ plus the number of vertices $\nVertices$).
In turn, since the number of paths $\nRoutes$ in the network is typically exponential in $\graphsize$, the algorithm's convergence speed must be at most polylogarithmic in $\nRoutes$.

\item
\emph{Scalable per-iteration complexity:}
An algorithm can be implemented efficiently only if the number of arithmetic operations and total memory required at each iteration remains scalable as the network grows in size.
In practice, this means that that the algorithm's per-itration complexity must not exceed $\bigof{\graphsize}$.

\item
\emph{Parameter-freeness:}
In most practical situations, the parameters of the model \textendash\ \eg the distribution of random events or the smoothness modulus of the network's cost functions \textendash\ cannot be assumed known, so any routing recommendation algorithm must likewise not require such parameters as input.
\end{enumerate}
\indent
To put these desiderata in context, we begin below by reviewing the most relevant works in this direction.
\endedit

\subsection*{Related work\afterhead}

The \emph{static} regime of our model matches the framework of \citet{BEL06} who showed that a variant of the \acdef{EW} algorithm \citep{Vov90,LW94,ACBFS95,BecTeb03} converges to equilibrium at an $\bigof{\log\nRoutes \big/\sqrt{\nRuns}}$ rate (in the sense of time averages).
This result was subsequently extended to routing games with stochastic cost functions by \citet{KBB15,KKDB15}, who showed that \ac{EW} also enjoys an $\bigof{\log\nRoutes \big/\sqrt{\nRuns}}$ convergence rate to \emph{mean} \aclp{WE} (again, in the Cesàro sense).
However, if the learning rate of the \ac{EW} algorithm is not chosen appropriately in terms of $\nRuns$, the method may lead to non-convergent, chaotic behavior, even in symmetric congestion games over a $2$-link Pigou network \citep{PPP17}.

\beginPM
From an optimization viewpoint, nonatomic congestion games with stochastic cost functions correspond to convex minimization problems with stochastic first-order oracle feedback.
In this setting, the $\bigof{1/\sqrt{\nRuns}}$ rate is, in general, unimprovable \cite{Nes04,Bub15}, so the guarantee of \citet{KBB15} is order-optimal.
On the other hand, in the static regime, the $\bigof{1/\sqrt{\nRuns}}$ convergence speed of \citet{BEL06} is \emph{not} optimal:
since nonatomic congestion games admit a smooth convex potential \textendash\ known as the \acdef{BMW} potential \citep{BMW56} \textendash\ this rate can be improved to $\bigof{1/\nRuns^{2}}$ via the seminal ``accelerated gradient'' algorithm of \citet{Nes83}.
On the downside, if applied directly to our problem, the algorithm of \citet{Nes83} has a catastrophic $\Theta(\nRoutes)$ dependence on the number of paths;
however, by coupling it with a ``mirror descent'' template in the spirit of \cite{Nes09,Xia10}, \citet{KBB15} proposed an accelerated method with an exponential projection step that is particularly well-suited for congestion problems.
In fact, as we show in the sequel, it is possible to design an \acl{XLEW} method \textendash\ which we call \acs{XLEW} \textendash\ that achieves an $\bigof{\log\nRoutes \big/\nRuns^{2}}$ rate in static environments.
\acused{XLEW}
\endedit

Importantly, despite its optimality in the static regime, \ac{XLEW} fails to converge altogether in stochastic problems;
moreover, the method's step-size must also be tuned with prior knowledge of the problem's smoothness parameters
(which are not readily available to the optimizer).
The \acdef{UPGD} algorithm of \citet{Nes15} provides a work-around to resolve the latter issue, but it relies on a line-search mechanism that cannot be applied to stochastic problems, so it does not resolve the former.
Instead, a partial solution for the stochastic case is achieved by the \ac{ACSA} algorithm of \citet{Lan12} which achieves order-optimal rates in both the static and stochastic regimes;
however:
\begin{enumerate*}
[\itshape a\upshape)]
\item
the running time $\nRuns$ of the \ac{ACSA} algorithm must be fixed in advance as a function of the accuracy threshold required;
and
\item
\ac{ACSA} further assumes full knowledge of the smoothness modulus of the game's cost functions (which cannot be computed ahead of time).
\end{enumerate*}

To the best of our knowledge, the first parameter-free algorithm with optimal rate interpolation guarantees is the \accelegrad algorithm of \citet{LYC18}, which was originally developed for \emph{unconstrained} problems (and assumes knowledge of a compact set containing a solution of the problem).
The \unixgrad proposal of \citet{KLBC19} subsequently achieved the desired adaptation in constrained problems, but under the requirement of a bounded Bregman diameter.
This requirement rules out the \ac{EW} template (the simplex has infinite diameter under the entropy regularizer that generates the \ac{EW} algorithm), so the convergence speed of \unixgrad ends up being \emph{polynomial} in the number of paths in the network \textendash\ and since the latter scales exponentially with the size of the network, \unixgrad is unsuitable for networks with more than $30$ or so nodes.

\beginPM
Finally, another major challenge in terms of scalability is the algorithm's \emph{per-iteration complexity}, \ie the memory and processing requirements for making a single update.
In regard to this point, the per-iteration complexity of the \ac{UPGD}, \ac{ACSA}, \accelegrad and \unixgrad algorithms is linear in the number of paths $\nRoutes$ in the network, and hence \emph{exponential} in the size of the network.
The only exception to this list is the (non-accelerated) \ac{EW} algorithm, which can be implemented efficiently with $\bigof{\graphsize}$ operations per iteration via a technique known as ``weight-pushing'' \cite{TW03,GLLO07}.
However, the weight-pushing technique is by no means universal:
it was specifically designed for the \ac{EW} algorithm, and it is not applicable to any of the universal methods discussed above (precisely because of the acceleration mechanism involved).

\subsection*{Our contributions\afterhead}

In summary, all existing traffic assingment algorithms and methods fail in at least two of the axes mentioned above:
either they are slow / non-convergent outside the specific regime for which they were designed,
or they cannot be implemented efficiently (see also \cref{tab:related} for an overview).
Consequently, our paper focuses on the following question:
\medskip
\begin{center}
\itshape
Is it possible to design a \textbf{scalable, paramater-free traffic assignment algorithm}\\
which is \textbf{simultaneously order-optimal} in both static and stochastic environments?
\end{center}
\medskip

To provide a positive answer to this question, we take a two-step approach.
\endedit
Our first contribution is to design an \acl{AEW} algorithm \textendash\ dubbed \acs{AEW} \textendash\ which is simultaneously order-optimal, in both $\nRuns$ and $\nRoutes$, and in both static and stochastic environments.
Informally, we have:
\acused{AEW}
\setcounter{informal}{\getrefnumber{thm:AEW}-1}

\begin{informal}
The \ac{AEW} algorithm enjoys the following equilibrium convergence guarantees after $\nRuns$ epochs:
\begin{enumerate}
\item
In static networks, \ac{AEW} converges to a \acl{WE} at a rate of $\bigof[\big]{(\log\nRoutes)^{3/2} \big/ \nRuns^{2}}$.
\item
In stochastic networks, it converges to a mean \acl{WE} at a rate of $\bigof[\big]{(\log\nRoutes)^{3/2} \big/ \sqrt{\nRuns}}$.
\end{enumerate}
\end{informal}

In the above,
the logarithmic dependency on $\nRoutes$ ensures that the convergence speed of \ac{AEW} is polynomial \textendash\ and in fact, \emph{subquadratic} \textendash\ in $\graphsize$.
In this way, \ac{AEW} successfully solves the challenge of achieving optimal convergence rates in both static and stochastic environments, while remaining parameter-free and ``anytime'' (\ie there is no need to tune the algorithm in terms of $\nRuns$).
To the best of our knowledge, this is the first method that simultaneously achieves these desiderata.

\beginPM
On the other hand, a crucial drawback of \ac{AEW} is that it operates at the path level, \ie it updates at each stage a state variable of dimension $\nRoutes$;
as a result, the algorithm's per-iteration complexity is typically \emph{exponential} in $\graphsize$.
To overcome this, we propose an improved algorithm, which we call \acs{ALW} (short for \acli{ALW} algorithm), and which emulates the \ac{AEW} template with two fundamental differences:
\acused{ALW}
First, instead of maintaining a state variable \emph{per path}, \ac{ALW} maintains a state vector \emph{per node}, with dimension equal to the node's out-degree.
Second, the recommended path for any particular user of the navigation interface is constructed ``bottom-up'' via a sequence of routing probabilities updated at each node, so there is no need to ever keep in memory (or update) a path variable.

This process looks similar to the weight-pushing technique of \cite{TW03,GLLO07} used to scale down the per-iteration complexity of the \ac{EW} algorithm.
\endedit
However, the acceleration mechanism of \ac{AEW}
involves several moving averages that \emph{cannot} be implemented via weight-pushing;
and since these averaging steps are downright essential for the universality of \ac{AEW}, they cannot be omitted from the algorithm.
Instead, our solution for executing these averaging steps is based on the following observation:
although the averaged recommendations profiles cannot be ``weight-pushed'', the induced edge-load profiles (\ie the mass of traffic induced on each edge) can.
This observation allows us to introduce a novel dynamic programming subroutine \textendash\ called \ppm \textendash\ which
computes efficiently these averaged loads and then uses them to reconstruct a set of routing recommendations that are consistent with these loads.

\beginPM
In this way, by combining the \ac{AEW} blueprint with the \ppm subroutine, \ac{ALW} only requires $\bigof{\graphsize}$ memory and processing power per update, all the while retaining the sharp convergence properties of the \ac{AEW} mother scheme.
In more detail, we have:

\setcounter{informal}{\getrefnumber{thm:adapush}-1}
\begin{informal}
The \ac{ALW} algorithm converges to equilibrium at a rate of $\bigof[\big]{(\log\nRoutes)^{3/2} \big/\nRuns^{2}}$ in static environments and $\bigof[\big]{(\log\nRoutes)^{3/2} \big/\sqrt{\nRuns}}$ in stochastic environments;
moreover, the total amount of arithmetic operations required per epoch is $\bigof{\graphsize}$ where $\graphsize$ is the size of the network.
\end{informal}
\endedit

This result shows that \ac{ALW} successfully meets the desiderata stated earlier,
so \ac{AEW} might appear redundant.
However, it is not possible to derive the equilibrium convergence properties of \ac{ALW} without first going through \ac{AEW},
so,
to better convey the ideas involved, we also describe the \ac{AEW} algorithm in detail (even though it is not scalable per se).


\begin{table}[t]
\renewcommand{\arraystretch}{1.3}
\footnotesize
\centering

\begin{tabular}{lcccccc}
\toprule
	&\acs{EW}
	&\acs{XLEW}
	&\unixgrad
	&\acs{UPGD}
	&\acs{ACSA}
	&\acs{ALW}
	\\
\midrule
Static
	&${\log\nRoutes \big/\sqrt{\nRuns}}$
	&${\log\nRoutes \big/\nRuns^{2}}$
	&${\nRoutes \big/\nRuns^{2}}$
	&${\log\nRoutes \big/\nRuns^{2}}$
	&${\log\nRoutes \big/\nRuns^{2}}$
	&${(\log\nRoutes)^{3/2}/\nRuns^{2}}$
	\\
Stochastic
	&${\log\nRoutes \big/\sqrt{\nRuns}}$
	&\xmark
	&${\nRoutes \big/\sqrt{\nRuns}}$
	&\xmark
	&${\log\nRoutes \big/\sqrt{\nRuns}}$
	&${(\log\nRoutes)^{3/2}/\sqrt{\nRuns}}$
	\\
Anytime
	&partially
	&\cmark
	&\cmark
	&\xmark
	&\xmark
	&\cmark
	\\
Param.-Agn.
	&\xmark
	&\xmark
	&\cmark
	&\cmark
	&\xmark
	&\cmark
	\\
Compl./Iter.
	&${\graphsize}$
	&${\nRoutes}$
	&${\nRoutes}$
	&${\nRoutes}$
	&${\nRoutes}$
	&${\graphsize}$
	\\
\bottomrule
\end{tabular}
\medskip
\caption{{\beginPM
Overview of related work in comparison to the \acs{ALW} algorithm (this paper).
For the purposes of this table,
$\graphsize$ refers to the size of the underlying graph while $\nRoutes$ refers to the number of relevant paths in the network (so $\nRoutes$ is typically exponential in $\graphsize$).
The ``anytime'' property refers to whether the number of iterations $\nRuns$ must be fixed at the outset and included as a parameter in the algorithm;
if not, the algorithm is labeled ``anytime''.
Finally, the ``parameter-agnostic'' property refers to whether any other parameters \textendash\ such as the Lipschitz modulus of the network's cost functions \textendash\ need to be known beforehand or not.
All estimates are reported in the $\bigof{\cdot}$ sense.
\endedit}}
\label{tab:related}
\vspace{-1ex}
\end{table}


\subsection*{Paper outline\afterhead}

Our paper is structured as follows.
In \cref{sec:setup}, we formally define our congestion game setup, the relevant equilibrium notions, and our learning model.
Subsequently, to set the stage for our main results, we present in \cref{sec:nonadapt} the classic \acl{EW} algorithm as well as the \acceleweight variant which achieves an $\bigof{\log\nRoutes/\nRuns^{2}}$ convergence rate in static environments;
both algorithms are non-adaptive, and they are used as a baseline for our adaptive results.
Our analysis proper begins in \cref{sec:adaweight}, where we present the \ac{AEW} algorithm and its convergence analysis.
Then, in \cref{sec:adapush}, we present the \acl{dis} setup used to construct the \ac{ALW} algorithm and prove its convergence guarantees.
Finally, in \cref{sec:numerics}, we report a series of numerical experiments validating our theoretical results in real-life transport networks.

\section{Problem setup}
\label{sec:setup}

We begin in this section by introducing the basic elements of our model.

\subsection{The game\afterhead}
\label{sec:game}

Building on the classic congestion framework of \citet{BMW56}, we consider a class of nonatomic routing games defined by the following three primitives:
\begin{enumerate*}[(\itshape i\hspace*{.5pt}\upshape)]
\item
an underlying \emph{network structure};
\item
the associated set of \emph{traffic demands};
and
\item
the network's \emph{cost functions}.
\end{enumerate*}
The formal definition of each of these primitives is as follows:

\begin{enumerate}
[label=\bfseries\arabic*.,itemsep=\medskipamount]

\item
\textbf{Network structure:}
Consider a multi-graph $\graph \equiv \graphall$ with vertex set $\vertices$ and edge set $\edges$.
The focal point of interest is a set of \PMedit{distinct} \acdef{OD} pairs $(\source{\pair},\sink{\pair}) \in \vertices\times \vertices$ indexed by $\pair\in\pairs = \{1,\dotsc,\nPairs\}$.
\beginPM
For each $\pair \in \pairs$, we assume given a directed acyclic subgraph $\graph^{\pair} \equiv (\vertices^{\pair},\edges^{\pair})$ of $\graph$ that determines the set of routing paths from $\source{\pair}$ to $\sink{\pair}$, and we write $\routes^{\pair}$ for the corresponding set of paths joining $\source{\pair}$ to $\sink{\pair}$ in $\graph^{\pair}$.
For posterity, we will also write $\routes = \union_{\pair\in\pairs} \routes^{\pair}$ for the set of all routing paths in the network, and $\nRoutes = \abs{\routes}$ and $\nRoutes^{\pair} = \abs{\routes^{\pair}}$ for the respective cardinalities;
likewise, we will write $\graphsize^{\pair} = \abs{\vertices^{\pair}} + \abs{\edges^{\pair}}$ for the size of $\graph^{\pair}$, and $\graphsize = \sum_{\pair\in\pairs} \graphsize^{\pair}$ for the total size of the network.
\endedit

\item
\textbf{Traffic demands and flows:}
Each pair $\pair\in\pairs$ is associated to a \emph{traffic demand} $\mass{\pair} > 0$ that is to be routed from $\source{\pair}$ to $\sink{\pair}$ via $\routes^{\pair}$;
we also write $\masssum = \onenorm{\mass{}} = \sum_{\pair \in \pairs} \mass{\pair}$ and $\massmax = \supnorm{\mass{}} = \max_{\pair \in \pairs} \mass{\pair}$ for the total and maximum traffic demand associated to the network's \ac{OD} pairs respectively.

Now, to route this traffic, the set of feasible \emph{traffic assignment profiles} \textendash\ or \emph{flows} \textendash\  is defined as
\begin{equation}
\label{eq:flows}
\txs
\flows
	\defeq \setdef[\big]{\flow \in\R_{+}^{\routes}}{\sum_{\route\in\routes^{\pair}} \flow_{\route} = \mass{\pair}, \pair=1,\dotsc,\nPairs}
\end{equation}
\ie as the product of scaled simplices $\flows = \prod_{\pair} \mass{\pair} \simplex(\routes^{\pair})$. 
In turn, each feasible flow profile $\flow\in\flows$ induces on each edge $\edge\in\edges$ the corresponding \emph{traffic load}
\begin{equation}
\label{eq:load}
\load_{\edge}(\flow)
	= \sum_{\route\in\routes} \oneof{\edge\in\route} \flow_{\route}
\end{equation}
\ie the accumulated mass of all traffic going through $\edge$.
It is also worth noting here that the path index $\route\in\routes$ completely characterizes the \ac{OD} pair $\pair\in\pairs$ to which it belongs;
when we want to make this relation explicit, we will write $\flow_{\route}^{\pair}$ instead of $\flow_{\route}$.

\item
\textbf{Congestion costs:}
The traffic routed through a given edge $\edge\in\edges$ incurs a \emph{congestion cost} depending on the total traffic on the edge and/or any other exogenous factors.
Formally, we will collectively encode all exogenous factors in a \emph{state variable} $\sample\in\samples$ \textendash\ the ``\emph{state of the world}'' \textendash\ that takes values in some ambient probability space $\probspace$.
In addition, we will assume that each edge $\edge\in\edges$ is endowed with a \emph{cost function} $\cost_{\edge} \from \R_{+}\times\samples \to \R_{+}$ which determines the cost $\cost_{\edge}(\load_{\edge}(\flow);\sample)$ of traversing $\edge\in\edges$ when the network is at state $\sample\in\samples$ and traffic is assigned according to the flow profile $\flow\in\flows$.
Analogously, the cost to traverse a path $\route\in\routes$ will be given by the induced \emph{path-cost function} $\cost_{\route}\from\flows\times\samples\to\R_{+}$ defined as
\begin{equation}
\label{eq:cost-path}
\cost_{\route}(\flow;\sample)
	= \sum_{\edge\in\route} \cost_{\edge}(\load_{\edge}(\flow);\sample).
\end{equation}
\end{enumerate}

In this general setting, the only assumption that we will make for the game's cost functions is as follows:
\begin{assumption}
\label{asm:cost}
Each cost function $\cost_{\edge}(\real;\sample)$, $\edge\in\edges$, is
measurable in $\sample$
and
non-decreasing, bounded and Lipschitz continuous in $\real$.
Specifically, there exist $\ub >0$ and $\lips>0$ such that
$\cost_\edge (\real;\sample) \leq \ub$
and
$\abs{\cost_{\edge}(\real;\sample) - \cost_{\edge}(\realalt;\sample)} \leq \lips \abs{\realalt - \real}$ for all $\real,\realalt \in [0, \masssum]$, all $\edge\in\edges$, and $\prob$-almost all $\sample\in\samples$.
\end{assumption}

\cref{asm:cost} represents a very mild regularity requirement that is satisfied by most congestion models that occur in practice \textendash\ including BPR, polynomial, or regularly-varying latency functions, \cf \cite{NRTV07,BEL06,LMP75,CBCMS20,patriksson2015traffic,OCW16} and references therein.
For this reason, we will treat \cref{asm:cost} as a standing, blanket assumption and we will not mention it explicitly in the sequel.

\subsection{Regimes of uncertainty and examples\afterhead}

The advent of uncertainty in the game's cost functions (as modeled by $\sample\in\samples$) is an important element of our congestion game framework.
To elaborate further on this, it will be convenient to define the \emph{mean cost function} of edge $\edge\in\edges$ as
\PM{I did not touch this, but why does $\meancost$ take flow profiles as arguments while $\cost$ takes loads as arguments?
I know we \emph{can} do this, I'm just curious why\dots basically, is this intentional or is it a typo?}\DQV{I did this intentionally, but I am open to changes. $\cost_{\edge}$ should be viewed as 1-D to lighten Assumption 1. Meanwhile, I'd like to highlight that $\meancost_{\edge}$ can be seen as a function of flows (\eg Definition of equi flow).}
\PM{Not sure I understand this:
the definition of equilibrium flows uses $\meancost_{\route}$ not $\meancost_{\edge}$, no? [\cf \eqref{eq:cost-mean}]
Is there any other point you have in mind?}
\begin{align}
\meancost_{\edge}(\flow)
	&= \ex_{\sample} \bracks*{ \cost_{\edge}(\load_{\edge}(\flow);\sample)}
\intertext{and consider the corresponding (random) fluctuation process}
\label{eq:noise}
\noise_{\edge} (\flow;\sample)
	&= \cost_{\edge}(\load_{\edge}(\flow);\sample)  - \meancost_{\edge}(\flow)
\end{align}
\ie the deviation of the network's cost functions at state $\sample\in\samples$ from their mean value.
The magnitude of these deviations may then be quantified by the randomness parameter
\begin{equation}
\label{eq:noisevar}
\sdev
	= \ess\sup\nolimits_{\sample\in\samples} \max\nolimits_{\flow\in\flows, \edge\in\edges}
		\abs{\noise_{\edge}(\flow;\sample)}
\end{equation}
\PMedit{where $\ess\sup$ denotes the essential supremum over $\samples$ with respect to $\prob$ (so $\sdev \leq 2\ub$ by \cref{asm:cost}).}
Informally, larger values of $\sdev$ indicate a higher degree of randomness, implying in turn that the traffic assignment problem becomes more difficult to solve;
on the other hand,
if $\sdev=0$, one would expect algorithmic methods to achieve better results.
This distinction plays a key role in the sequel so we formalize it as follows:

\begin{definition}[Static and Stochastic Environments]
\label{def:regime}
When $\sdev =0$, we say that the environment is \emph{static};
otherwise, if $\sdev>0$, we say that the environment is \emph{stochastic}.
\end{definition}

The dependence of the network's cost functions on exogenous random factors is the main difference of our setup with standard congestion models in the spirit of \citet{BMW56} and \citet{NRTV07}.
For concreteness, we mention below two existing models that can be seen as special cases of our framework:


\begin{example}
[Deterministic routing games]
\label{ex:stationary}
The \emph{static regime} described above matches the deterministic routing models of \citet{FV04,BEL06} and \citet{KDB14,KDB15}.
In these models, the network only has a single (deterministic) state
so, trivially, $\sdev=0$.
\endenv
\end{example}

\begin{example}
[Noisy cost measurements]
\label{ex:BPR1}
\beginPM
To model uncertainty in the cost measurement process, \citet{KDB15,KKDB15} considered the case where the network's cost functions are fixed, but cost measurements are only accurate up to a random, zero-mean error.
In our framework, this can be modeled by simply assuming a familly of cost functions of the form $\cost_{\edge}(\flow;\sample) = \meancost_{\edge}(\flow) + \sample_{\edge}$ with $\exof{\sample_{\edge}} = 0$ for all $\edge\in\edges$.
\hfill
\endenv
\end{example}

Besides these standard examples, our model is sufficiently flexible to capture other random factors such as weather conditions, traffic accidents, road incidents, etc.
For instance, to model the difference between dry and wet weather, one can take $\samples = \{\sample_{\textrm{dry}},\sample_{\textrm{wet}}\}$ and consider a set of cost functions with $\cost_{\edge}(\load_\edge(\flow); \sample_{\textrm{dry}}) < \cost_{\edge} (\load_\edge(\flow);\sample_{\textrm{wet}})$ to reflect the fact that the congestion costs are higher when the roads are wet.
\PMedit{However, to maintain the generality of our model, we will not focus on any particular application.}


\subsection{Notions of equilibrium\afterhead}
\label{sec:equilibrium}

In the above framework, each state variable $\sample\in\samples$ determines an instance of a routing game, defined formally as a tuple $\meangame_{\sample} \equiv \meangame_{\sample}(\graph,\pairs,\routes,\cost_{\sample})$ where $\cost_{\sample}$ is shorthand for the network's cost functions $\{\cost_{\edge}(\cdot;\sample)\}_{\edge\in\edges}$ instantiated at $\sample$.
Of course, in analyzing the game, each individual instance $\meangame_{\sample}$ is meaningless by itself unless $\prob$ assigns positive probability only to a \emph{single} $\sample$.
For this reason, we will instead focus on the \emph{mean game} $\meangame \equiv \meangame(\graph,\pairs,\routes,\meancost)$ which has the same network and routing structure as every $\meangame_{\sample}$, $\sample\in\samples$, but with congestion costs given by the \emph{mean cost functions}
\begin{equation}
\label{eq:cost-mean}
\meancost_{\route}(\flow)
	= \exof{\cost_{\route}(\flow;\sample)}
	= \insum_{\edge\in\route} \exof{\cost_{\edge}(\load_{\edge}(\flow);\sample)}.
\end{equation}
Motivated by the route recommendation problem described in the introduction, we will focus on traffic assignment profiles where Wardrop's unilateral optimality principle \citep{War52} holds on average, \ie \emph{all traffic is routed along a path with minimal mean cost.}
Formally, we have:

\begin{definition}
[Mean equilibrium flows]
\label{def:equilibrium}
We say that $\flowbase \in \flows$ is a \emph{mean equilibrium flow} if and only if $\meancost_{\route}(\flowbase) \leq \meancost_{\routealt}(\flowbase)$ for all $\pair\in\pairs$ and all $\route,\routealt\in\routes^{\pair}$ such that $(\flowbase)_{\route}^{\pair} > 0$.
\end{definition}

\begin{remark}
\cref{def:equilibrium} means that, on average, no user has an incentive to deviate from the recommended route;
obviously, when the support of $\prob$ is a singleton, we recover the usual definition of a \emph{Wardrop equilibrium} \citep{War52,BMW56,NRTV07}.
This special case will be particularly important and we discuss it in detail in the next section.
\end{remark}

Importantly, the problem of finding an equilibrium flow of a (fixed) routing game $\meangame_{\sample}$ admits a \emph{potential function} \textendash\ often referred to as the \acf{BMW} potential \citep{BMW56,DS69}.
Specifically, for a given instance $\sample\in\samples$, the \ac{BMW} potential is defined as
\begin{equation}
\label{eq:BMW-inst}
\tag{BMW}
\pot_{\sample}(\flow)
	= \insum_{\edge \in \edges} \int_{0}^{\load_\edge(\flow)} \cost_{\edge} \parens{\real;\sample} \dd\real
	\quad
	\text{for all $\flow\in\flows$},
\end{equation}
and it has the fundamental property that $\argmin\pot_{\sample}$ coincides with the set $\Eq(\meangame_{\sample})$ of equilibria of $\meangame_{\sample}$.
In our stochastic context, a natural question that arises is whether this property can be extended to the mean game $\meangame$ where the randomness is ``averaged over''.
Clearly, the most direct candidate for a potential function in this case is the \emph{mean \ac{BMW} potential}
\begin{equation}
\label{eq:BMW-mean}
\meanpot(\flow)
	= \exof{\pot_{\sample}(\flow)}
	= \exof*{ \sum_{\edge\in\edges} \int_{0}^{\load_\edge(\flow)} \cost_{\edge}(\real;\sample)\dd\real}.\end{equation}
Since $\cost_{\edge}$ is continuous and non-decreasing for all $\edge\in\edges$, the mean \ac{BMW} potential $\meanpot$ is, in turn, continuously differentiable and convex on $\flows$.
Of course, since the probability law $\prob$ is not known, we will \emph{not} assume that $\meanpot$ (and/or its gradients) can be explicitly computed in general.
Nevertheless, building on the deterministic characterization of \citet{BMW56}, we have the following equivalence between mean equilibria and minimizers of $\meanpot$:

\begin{proposition}
\label{prop:pot-mean}
$\meanpot$ is a potential function for the mean game $\meangame \equiv \meangame(\graph,\pairs,\routes,\meancost)$;
in particular, we have
\begin{equation}
\label{eq:pot-mean}
\frac{\pd\meanpot}{\pd\flow_{\route}}
	= \meancost_{\route}(\flow)
	\quad
	\text{for all $\flow\in\flows$ and all $\route\in\routes$}.
\end{equation}
Accordingly,
a flow profile $\flowbase \in \flows$ is an equilibrium of the mean game $\meangame$ if and only if it is a minimizer of $\meanpot$ over $\flows$;
more succinctly,
\PMedit{the equilibrium set $\Eq(\meangame)$ of $\meangame$ satisfies $\Eq(\meangame) = \argmin_{\flow\in\flows} \meanpot(\flow)$.}
\DQV{Is the $\Eq$ notation standard? Or you take this equation here as a definition?}
\PM{I think it's quite standard, but I added a phrase}
\end{proposition}

\begin{Proof}
First, apply the dominated convergence theorem (saying that gradients and expectations commute), we have $\frac{\pd\meanpot}{\pd\flow_{\route}} = \ex \bracks*{ \frac{\pd\pot_{\sample}}{\pd\flow_{\route}} } = \ex \bracks*{\cost_{\route}(\flow; \sample)} = \meancost_{\route}(\flow)$. The rest of the proof of \cref{prop:pot-mean} can be obtained by following \citep{BMW56,RT02} for the mean game $\game$.
\PM{We still need to show that gradients and expectations commute, and all these pesky little details, no?
Wouldn't it make sense to include a \emph{short} discussion of this and then say ``the rest of the proof follows as in blah blah''?}\DQV{done.}
\end{Proof}

In view of \cref{prop:pot-mean}, $\meanpot$ provides a natural merit function for examining how close a given flow profile $\flow\in\flows$ is to being an equilibrium.
\beginPM
Formally, given a sequence of candidate flow profiles $\curr\in\flows$, $\run=\running$, we define the associated \emph{equilibrium gap} after $\nRuns$ epochs as
\begin{equation}
\label{eq:gap}
\gap(\nRuns)
	= \meanpot(\last[\state]) - \min\meanpot.
\end{equation}
Clearly, the sequence in question converges to $\Eq(\meangame)$ if and only if $\gap(\nRuns)\to0$.
Moreover, the speed of this convergence is also captured by $\gap(\nRuns)$, so all our convergence certificates will be stated in terms of $\gap(\nRuns)$.
\endedit


\subsection{Sequence of events\afterhead}
\label{sec:learn}

With all this in hand, the sequence of events in our model unfolds as follows:
\begin{enumerate}
\item
\PMedit{At each time slot $\run=\running$, the navigation interface selects a traffic assignment profile $\curr \in \flows$ and it routes all demands received within this time slot according to $\curr$.}
\item
Concurrently, the state $\curr[\sample]$ of the network is drawn \acs{iid} from $\samples$ according to $\prob$.
\item
\PMedit{The interface observes the realized congestion costs $\curr[\signal]_{\edge} = \cost_\edge(\load_{\edge}(\curr);\curr[\sample])$ along each edge $\edge\in\edges$;
subsequently, it uses this information to update the routing profile $\curr$, and the process repeats.}
\end{enumerate}


\beginPM
In what follows, we will analyze different learning algorithms for updating $\curr$ in the above sequence of events, and we will examine their equilibrium convergence properties in terms of the equilibrium gap function \eqref{eq:gap}.
For posterity, we emphasize again that the optimizer has no prior knowledge of the enviroment, the distribution of random events, the network's cost functions, etc., so we will pay particular attention to the dependence of each algorithm's guarantees on these (otherwise unknown) parameters.
\endedit


\section{Non-adaptive methods}
\label{sec:nonadapt}

In the rest of our paper, we discuss a series of routing algorithms and their equilibrium convergence guarantees.
To set the stage for our main contributions, we begin by presenting two learning methods that are order-optimal in the two basic regimes described in the previous section:
\begin{enumerate}
\item
The ``vanilla'' \acf{EW} method of \citet{BEL06} for \emph{stochastic} environments.
\item
An \acl{XLEW} algorithm \textendash\ which we call \acs{XLEW} \textendash\ for \emph{static} environments.
\end{enumerate}
Both methods rely crucially on the (rescaled) \emph{logit choice map} $\logit\from\R^{\routes}\to\flows$, given in components as
\begin{equation}
\label{eq:logit}
\logit_{\route}(\dpoint)
	= \frac
		{\mass{\pair} \exp(\dpoint_{\route})}
		{\sum_{\routealt\in\routes^{\pair}} \exp(\dpoint_{\routealt})}
	\qquad
	\text{for all $\dpoint\in\R^{\routes}$, $\route\in\routes^{\pair}$ and $\pair\in\pairs$}.
\end{equation}
For concision, we will also write $\curr[\signal]_{\route} = \sum_{\edge\in\route} \curr[\signal]_{\edge}$ for the cost of path $\route\in\routes$ at stage $\run$, and $\curr[\signal] = (\curr[\signal]_{\route})_{\route\in\routes}$ for the profile thereof.

\subsection{\Acl{EW} in stochastic environments\afterhead}
\label{sec:EW}

We begin by presenting the \acl{EW} algorithm of \citet{BEL06} in standard pseudocode form:
\PM{I did not touch anything, but why do we write $\curr[\dstate] = \prev[\dstate] - \curr[\step] \curr[\signal]$ instead of $\next[\dstate] = \curr[\dstate] - \curr[\step] \curr[\signal]$?
Is this convention consistent with the proofs?
If so, where do we stop the history filtration?}\DQV{Thanks. For consistency with AdaWeight, I changed it to the latter.}


\begin{algorithm}[H]
\small
\DontPrintSemicolon

\textbf{Initialize} $\init[\dstate] \gets \zer$\;
\label{line:flow0}
\For{$\run = \running$}
{
	set $\curr[\state]  \gets \logit(\curr[\dstate])$
	and
	get $\curr[\signal] \gets \cost(\curr[\state],\curr[\sample])$
\tcp*[r]{route and measure costs}
\label{line:EW_flow}
	set $\next[\dstate]  = \curr[\dstate]  - \curr[\stepalt]\curr[\signal]$
\tcp*[r]{update path scores}
\label{line:EW_dpoint} 
}

\caption{\Acf{EW}}
\label{alg:EW}
\acused{EW}
\end{algorithm}


We then have the following equilibrium convergence result:

\begin{theorem}
\label{thm:EW}
Suppose that \ac{EW} \textpar{\cref{alg:EW}} is run for $\nRuns$ epochs with a fixed learning rate $\curr[\stepalt] = \sqrt{\log\parens*{{\massmax\nRoutes} \big/ {\masssum}}} / (\ub \sqrt{\nRuns})$ for any $\run$.
Then the time-averaged flow profile $\last[\aveplay] = (1/\nRuns) \sum_{\run=\start}^{\nRuns} \curr$ enjoys the equilibrium convergence rate:
\PM{Can we put the value of $\stepalt$ that optimizes the rate?
It would make the result sharper (and also highlight the fact that we need to tune the algorithm's parameters).}\DQV{Done}
\begin{equation}
\label{eq:rate-EW}
\exof*{\gap \parens*{\nRuns }}
	\leq \frac{2 \masssum \ub \sqrt{\log\parens*{{\massmax\nRoutes} \big/ {\masssum}}}} {\sqrt{\nRuns}}
	= \bigoh \parens*{\frac{(\log\nRoutes)^{1/2}}{\sqrt{\nRuns}}}.
\end{equation}
\end{theorem}%

\beginPM
The proof of \cref{thm:EW} relies on standard techniques, so we omit it;
for a series of closely related results, we refer the reader to \citet{KKDB15} and \citet[Corollary 2.14 and Theorem 4.1]{SS11};
What is more important for our purposes is that \cref{thm:EW} confirms that
\endedit
\ac{EW} achieves an $\bigof[\big]{1 \big/ \sqrt{\nRuns}}$ convergence speed in the number of epochs $\nRuns$, so it is order-optimal in stochastic environments \citep{Nes04,Bub15}.
Moreover, the algorithm's convergence speed is logarithmic in $\nRoutes$, and hence (at most) linear in the size $\graphsize$ of the underlying network.

On the downside, the convergence rate obtained for Algorithm \ref{alg:EW} concerns the time-averaged flow $\last[\aveplay]$, not the actual recommendation, a distinction which is important for practical applications. Second, \cref{thm:EW} is not an ``anytime'' convergence result in the sense that the algorithm's learning rate $\curr[\stepalt]$ must be tuned relative to $\nRuns$ (the rate is not valid otherwise).
\beginPM
This issue can be partially resolved either via a doubling trick \citep{SS11} (at the cost of restarting the algorithm ever so often) or by using a variable learning rate of the form $\curr[\stepalt] \propto 1/\sqrt{\run}$ (at the cost of introducing an additional $\log\nRuns$ factor in the algorithm's convergence speed).
However, in either case, the rate \eqref{eq:rate-EW} only applies to the sequence of time averages, not the actual recomendations.
\endedit

\subsection{\Acl{XLEW} in static environments\afterhead}
\label{sec:XLEW}

We now turn our attention to the static regime, \ie when there are no exogenous variations in the network's cost functions ($\sdev = 0$);
in this case, it is reasonable to expect that the $\bigof{1/\sqrt{\nRuns}}$ convergence speed of \ac{EW} can be improved.
\beginPM
Indeed, as we show below, the \ac{BMW} potential of the game is \emph{Lipschitz smooth}, so the optimal convergence speed in static environments is $\bigof{1/\nRuns^{2}}$ \citep{Nes04,Bub15}.
\PM{I reworked the proof below to make it a bit more transparent.}\DQV{Checked}

\begin{proposition}
\label{prop:smooth}
The \ac{BMW} potential $\meanpot$ is Lipschitz smooth relative to the $L^{1}$ norm on $\flows$.
In particular, its smoothness modulus is $\smooth = \length\bderiv$, where $\length = \max_{\route\in\routes} \sum_{\edge\in\route} \oneof{\edge\in\route}$ is the length of the longest path in $\routes$.
\end{proposition}    

\begin{Proof}[Proof of \cref{prop:smooth}]
It suffices to show that the gradient $\nabla\meanpot$ of $\meanpot$ is Lipschitz continuous with respect to the (primal) $L^{1}$ norm on $\flows$ and the (dual) $L^{\infty}$ norm on $\R^{\routes}$.
Indeed, for all $\flow,\flowalt\in\flows$, combining \cref{prop:pot-mean,asm:cost} yields
\begin{align}
\supnorm{\nabla\meanpot(\flow) - \nabla\meanpot(\flowalt)}
	&= \max_{\route\in\routes} 
		\abs{\meancost_{\route}(\flow) - \meancost_{\route}(\flowalt)}
	\tag*{[by \cref{prop:pot-mean}]}
	\\
	&\leq \max_{\route\in\routes} \sum_{\edge \in \route} \lips \abs*{  \load_{\edge}(\flow) - \load_{\edge}(\flowalt)}
	\tag*{[by \cref{asm:cost}]}
	\\
	&\leq \lips \cdot
		\max_{\route\in\routes}
		\sum_{\edge\in\edges} \oneof{\edge\in\route}
		\sum_{\routealt\in\routes}
			\oneof{\edge\in\routealt} \abs{\flow_{\routealt} - \flowalt_{\routealt}}
	\tag*{[definition of $\load_{\edge}$]}
	\\
	&\leq \lips \cdot
		\sum_{\edge\in\edges}
		\sum_{\routealt\in\routes}
			\oneof{\edge\in\routealt} \abs{\flow_{\routealt} - \flowalt_{\routealt}}
	\tag*{[since $\oneof{\edge\in\route} \leq 1$]}
	\\
	&\leq \lips \cdot
		\max_{\route\in\routes} \sum_{\edge\in\edges} \oneof{\edge\in\route} \cdot
		\sum_{\routealt\in\routes} \abs{\flow_{\routealt} - \flowalt_{\routealt}}
	= \length \lips \onenorm{\flow - \flowalt}
	\tag*{[by Cauchy-Schwarz]}
\end{align}
so our claim follows.
\end{Proof}

Given that $\meanpot$ is smooth, the iconic $\bigof{1/\nRuns^{2}}$ rate mentioned above can be achieved by the accelerated gradient algorithm of \citet{Nes83}.
However, if \citeauthor{Nes83}'s algorithm were applied ``off the shelf'', the constants involved would be linear in $\nRoutes$ (the problem's dimensionality), and hence exponential in the size of the network.
Instead, building on ideas by \citet{AZO17} and \citet{KBB15}, this problem can be overcome by combining the acceleration mechanism of \citet{Nes83} with the dimensionally-efficient \ac{EW} template.
The resulting \acdef{XLEW} method proceeds in pseudocode form as follows:
\endedit


\begin{algorithm}[H]
\DontPrintSemicolon
\small

\KwIn{Smoothness parameter $\smooth = \length\bderiv$}

\textbf{Initialize}
	$\init[\dstate] \gets \zer$;\,
	$\beforeinit[\state] \gets \zer$;\,
	$\beforeinit[\step] \gets 0$;\,
	$\beforeinit[\stepalt] \gets (\nPairs \massmax \smooth)^{-1}$
	\label{line:XLEW_initial}\;

\For{$\run = \running$}
{
		set $\curr[\pstate] \gets \logit (\curr[\dstate])$
	\tcp*[r]{exploratory flow obtained from path scores}
	\label{line:XLEW_ppoint}
		set $\curr[\state] \gets \prev[\step] \prev[\state] + (1 - \prev[\step]) \curr[\pstate]$
	\tcp*[r]{average with previous state}
	\label{line:XLEW_flow}%
		set $\curr[\stepalt]
			\gets \prev[\stepalt]
			+ \beforeinit[\stepalt]/2
			+ \sqrt{\prev[\stepalt] \beforeinit[\stepalt] + (\beforeinit[\stepalt]/2)^{2}}$
	\tcp*[r]{update step-size}
	\label{line:XLEW_stepalt}
		set $\curr[\step] \gets {\prev[\stepalt]} \big/ {\curr[\stepalt]} $
	\tcp*[r]{update averaging weight} 
	\label{line:XLEW_step}
		set $\curr[\pstateavg] \gets \curr[\step] \curr[\state] + (1 - \curr[\step] ) \curr[\pstate]$
		and
		get $\curr[\signalavg] \gets \cost \parens*{\curr[\pstateavg],\curr[\sample]}$
	\tcp*[r]{route and measure costs}
	\label{line:XLEW_dpoint}
		set $\next[\dstate]
			\gets \curr[\dstate]
			- (1 - \curr[\step]) \curr[\stepalt]\curr[\signalavg]$
	\tcp*[r]{update path scores}  
}

\caption{\Acf{XLEW}}
\label{alg:XLEW}
\acused{XLEW}
\end{algorithm}


At a high level, \ac{XLEW} follows the acceleration template of \citet{Nes83}, but replaces Euclidean projections with the logit map \eqref{eq:logit}.
With this in mind, we obtain the following convergence guarantee:

\begin{theorem}
\label{thm:XLEW}
The sequence of flow profiles $\init,\afterinit,\dots$ generated by \ac{XLEW} \textpar{\cref{alg:XLEW}}
\PMedit{in a static environment \textpar{$\sdev=0$}}
enjoys the equilibrium convergence rate
\begin{equation}
\label{eq:thm_acceleweight}
\gap \parens*{{\nRuns}} 
	\leq \frac{4\smooth \nPairs^{2} \massmax^{2} \log\parens*{\massmax\nRoutes/\masssum}}{(\nRuns-\start)^2}
	= \bigof*{\frac{\log\nRoutes}{\nRuns^2}}.
\end{equation}
\end{theorem}

The proof of \cref{thm:XLEW} is based on techniques that are widely used in the analysis of accelerated methods so, to streamline our presentation, we relegate it to \cref{app:non-adaptive}.

\cref{thm:XLEW} confirms that, in the static regime, \ac{XLEW} converges to equilibrium with an optimal $\bigof[\big]{1\big/\nRuns^2}$ convergence speed, as desired.
For our purposes, it is more important to note that \cref{thm:XLEW} confirms that, in static environments, \acl{XLEW} achieves a much faster convergence rate than \ac{EW}.
However, this improvement comes with several limitations:
\PMedit{
First and foremost, as we discuss in \cref{sec:numerics}, \ac{XLEW} fails to converge altogether in the stochastic regime, so it is not a universal method (\ie an algorithm that is simultaneously order-optimal in both static and stochastic environments).
}
Second, to obtain the accelerated rate \eqref{eq:thm_acceleweight}, the step-size of \ac{XLEW} must be tuned with prior knowledge of the problem's smoothness modulus $\smooth$ (which is not realistically available to the learner), so the method is not adaptive.
\PMedit{These are both crucial limitations that we seek to overcome in \cref{sec:adaweight,sec:adapush}.}

\subsection{Per-iteration complexity of each method\afterhead}

\PMedit{We close this section with a short discussion on the per-iteration complexity of \ac{EW} and \ac{XLEW}.}
The key point of note is that \cref{alg:EW,alg:XLEW} both require $\bigof{\nRoutes}$ time and space per iteration (\ie in terms of arithmetic operations and memory storage).
This is inefficient in large-scale networks where $\nRoutes$ is exponentially large in $\graphsize$ (the size of the network), so it is not clear how to implement either of these algorithms in practice.

As discussed in \cref{sec:intro}, 
there exist variants of the \acl{EW} algorithm that \emph{can} be implemented efficiently in $\bigoh(\graphsize)$ time and space per iteration.
This is accomplished via a dynamic programming procedure known as ``weight-pushing'' \cite{TW03,GLLO07};
however, this technique is limited to boosting the computation step involving the logit mapping (Line 3 of \cref{alg:EW} and \cref{line:XLEW_ppoint} of \cref{alg:XLEW}) and cannot efficiently execute other steps of \ac{XLEW} that require $\bigof{\nRoutes}$ operations or storage bits (like the averaging steps involved in the acceleration mechanism).
To maintain the flow of our discussion, we will revisit this issue in detail in \cref{sec:adapush}.

\section{\Ac{AEW}: Adaptive learning in the presence of uncertainty}
\label{sec:adaweight}

\subsection{Adaptivity of \adaweight in optimizing equilibrium-convergence rates}
\label{sec:adaweight_statement}

To summarize the situation so far, we have seen that \expweight attains an $\bigoh(\log(\nRoutes)/\sqrt{\nRuns})$ rate, which is order-optimal in the stochastic case but suboptimal in static environments;
by contrast, \ac{XLEW} attains an $\bigoh(\log(\nRoutes)/\nRuns^{2})$ rate in static environment, but has no convergence guarantees in the presence of randomness and uncertainty.
Consequently, neither of these algorithms meets our stated objective to concurrently achieve order-optimal guarantees in both the static and stochastic cases (and without requiring prior knowledge of the problem's smoothness modulus).

To resolve this gap, we propose below an \acli{AEW} method \textendash\ \acs{AEW} for short \textendash\ which achieves these objectives by mixing the acceleration template of \ac{XLEW} with the \acl{DE} method of \citet{Nes07}.
We present the pseudocode of \ac{AEW}~below:



\begin{algorithm}[htbp]
\DontPrintSemicolon
\ttfamily
\small
    \SetKw{Return}{Return} 
 		 \textbf{Initialize}  $\beforeinit[\step] \gets 0$, $\beforeinit[\precom] \gets \zer$, $\init[\learn] \gets 1$ and $\init[\drecom] \gets \zer$\;
		\For{$\run = \start, \afterstart, \ldots$}{%
		set $\curr[\anchor]  \gets \sum_{\runalt=\beforestart}^{\run-1} \iter[\step] \iter[\precom] $\tcp*[r]{set an anchor primal point} \label{line:AEW-anchor}
		\tcp*[l]{the test phase:}
		set $\curr[\ptest]  \gets \logit \parens*{\learn^{\run} \curr[\drecom]}$ \label{line:AEW-flow}\;
		%
		set $\test^{\run} \gets \parens*{\curr[\step] \curr[\ptest] + \curr[\anchor] } \big/  \sum_{\runalt =\beforestart}^{\run} \iter[\step]  $
		and get $\curr[\testsignal] \gets \late(\curr[\test],\curr[\sample])$
		\tcp*[r]{reweigh with anchor + query a test point}
		\label{line:AEW-flowave} 
		set	$\curr[\dtest] \gets \curr[\drecom] - \curr[\step] \curr[\testsignal]$
		\tcp*[r]{exploratory score update} \label{line:AEW-dpointave}
		\tcp*[l]{the recommendation phase:}
		set $\curr[\precom] \gets \logit \parens*{ \curr[\learn] \curr[\dtest]}$  \label{line:AEW-flow12}\;
		set $\curr[\state]  \gets \parens*{\curr[\step] \curr[\precom] + \curr[\anchor]   } \big/  \sum_{\runalt =\beforestart}^{\run} \iter[\step]$
		and get $\curr[\signal] \gets \late(\curr[\state],\curr[\sample])$
		\tcp*[r]{reweigh with anchor + route and measure costs}
		\label{line:AEW-recom-flow}
		set $\next[\drecom]  \gets \curr[\drecom]  - \curr[\step] \curr[\signal]$
		\tcp*[r]{update path scores}
		\label{line:AEW-dpoint}
		%
		%
		%
		set $\next[\learn] \gets 1 \big/ \sqrt{1 + \sum_{\runalt = \beforestart}^{\run} \supnorm*{ \iter[\step] (\iter[\signal] - \iter[\testsignal])}^{2}}$
	 \tcp*[r]{update learning rate}\label{line:AEW-stepalt}
}
\caption{\Acf{AEW}}
\label{alg:AEW}
\end{algorithm}
 

\revise{
The main novelty in the definition of the \ac{AEW} algorithm is the introduction of two ``extrapolation'' sequences, $\curr[\precom]$ and $\curr[\dtest]$, that venture outside the convex hull of the generated primal (flow) and dual (score) variables respectively. These leading states are subsequently averaged, and the method proceeds with an adaptive step-size rule.
In more details, \ac{AEW} relies on three key components: 
\begin{enumerate}
[\upshape \itshape a)]
\item
A \acl{DE} mechanism for generating the leading sequences in Lines \ref{line:AEW-flow} and \ref{line:AEW-flow12};
these sequences are central for anticipating the loss landscape of the problem.
\item
An acceleration mechanism obtained from the ($\step^{\run}$)-\emph{weighted average} steps in Lines \ref{line:AEW-flowave} and \ref{line:AEW-recom-flow};
in the analysis, $\step^{\run}$ will grow as $\run$, so almost all the weight will be attributed to the state closest to the current one.
\item
An \emph{adaptive sequence of learning rates} (\cf Line~\ref{line:AEW-stepalt}) in the spirit of \citet{RS13-NIPS}, \citet{KLBC19} and \citet{ABM21}.
This choice is based on the ansatz that, if the algorithm encounters coherent gradient updates (which can only occur in static environments), it will eventually stabilize to a strictly positive value;
otherwise, it will decrease to zero at a $\Theta(1/\sqrt{\run})$ rate.
This property is crucial to interpolate between the stochastic and static regimes.
\end{enumerate}

The combination of the weighted average iterates and adaptive learning rate in \ac{AEW} is shared by the \unixgrad algorithm proposed by \citet{KLBC19}, which also provides rate interpolation in constrained problems.
However, \unixgrad requires the problem's domain to have a finite Bregman diameter \textendash\ and, albeit compact, the set of feasible flows $\flows$ has an \emph{infinite} diameter under the entropic regularizer that generates the \ac{EW} template. Therefore, \unixgrad is not applicable to our routing games.
This is the reason for switching gears to the ``primal-dual'' approach offered by the \acl{DE} template;
this primal-dual interplay provides the missing link that allows \ac{AEW} to simultaneously enjoy order-optimal convergence guarantees in both settings while maintaining the desired polynomial dependency on the problem's dimension.
}

In light of the above, our main convergence result for \acs{AEW} is as follows:

\begin{theorem}
\label{thm:AEW}
Let $\flow^{\start}, \flow^{\afterstart}, \ldots,$ be the sequence of flows recommended by \adaweight (Algorithm \ref{alg:AEW}) running with $\step^{\run}=\run$ for any $\run=\running$, then it enjoys the following equilibrium convergence rate:%
\begin{align}
\label{eq:AEW-stoch}
\exof*{\gap \parens*{\nRuns}}
\leq \frac{ 2\sqrt{2}\const  \sdev }{\sqrt{\nRuns}}
+ \frac{16 \smooth \sqrt{\nPairs \massmax} \const^{{3}/{2}} + \constalt}{\nRuns^{2}}
	= \bigoh\parens*{ \parens*{\log\nRoutes}^{{3}/{2}} \parens*{\frac{\sdev }{\sqrt{\nRuns}} + \frac{1}{\nRuns^2}  }}.
%
%
\end{align}
Specifically, in the static case, \adaweight enjoys the sharper rate: $\gap \parens*{\nRuns} \le  \bigoh\parens*{ { \parens*{\log\nRoutes}^{{3}/{2}}} \big/  {\nRuns^2}}$. In the above expressions, $\sdev$ is as defined in \cref{eq:noisevar} while $\const$ and $\constalt$ are positive constants given by
$\const \defeq \nPairs \massmax \bracks*{2 \log \parens*{{\nRoutes\massmax}/{\masssum}} + 13} = \bigoh\parens{\log{\nRoutes}}$
and
$\constalt \defeq \masssum \log(\nRoutes \massmax/\masssum) = \bigoh\parens*{{\log\nRoutes}}$.
\end{theorem}

\cref{thm:AEW} characterizes the convergence speed of \adaweight according to the number of learning iterations $\nRuns$, the number of paths $\nRoutes$ and the level of randomness $\sdev$ of the network's states (as defined in \eqref{eq:noisevar}). In the stochastic environment, $\sdev \ge 0$ and hence, \adaweight enjoys a convergence rate of order $\bigoh \left((\log \nRoutes)^{3/2} \sdev \big/ \sqrt{\nRuns} \right)$. In the static case, $\sdev = 0$ and the convergence rate is accelerated to $\bigoh \left((\log \nRoutes)^{3/2} \big/ \nRuns^2 \right)$. In instances where the cost observations become more accurate over time, \cref{thm:AEW} also allows us to achieve a smooth trade-off in the convergence rate. For example, if $\sdev = \bigoh \parens*{1/ \nRuns^{\real}}$ for some suitable $\real$, \adaweight’s rate of convergence carries a dependence of the order of $\bigoh \parens*{\nRuns^{\max\{ -\real - 1/2, -2\}}} $.

In view of these result, \cref{thm:AEW} confirms that \ac{AEW} satisfies the following desiderata:
 \begin{enumerate*}
	[(\itshape i\upshape)]
	\item
	it achieves \emph{simultaneously optimal guarantees} in both stochastic and static environments in the number of learning iterations (\ie $\bigoh(1/\sqrt{\nRuns})$ and $\bigoh(1/\nRuns^{2})$ respectively);
	\item
	the derived rates maintain a \emph{polynomial dependency} in terms of the network' s combinatorial primitives; and
	\item
	it requires \emph{no prior tuning} by the learner.
\end{enumerate*}
Moreover, unlike \expweight, the convergence of \ac{AEW} corresponds to an actual traffic flow profile that is implemented in epoch $\run$ and not the average~flow.

\para{The per-iteration complexities of \adaweight} The \adaweight method, as presented in Algorithm~\ref{alg:AEW}, while having a simple presentation, takes $\mathcal{O}\parens{\nRoutes}$ time (and space) to complete each iteration. Therefore, implementing \adaweight is inefficient. Although \adaweight shares certain elements of the \ac{EW} template, the naive application of weight-pushing \emph{fails} to obtain an implementation of \adaweight having a polynomial per-iteration complexity in the network's size. This is due to the complicated averaging steps in \adaweight (Lines~\ref{line:AEW-flow12} and \ref{line:AEW-recom-flow} of Algorithm~\ref{alg:AEW}). In \cref{sec:adapush}, we re-discuss this in more details and we propose another algorithm, called \adapush, that maintains all desired features of \adaweight and achieves an efficient per-iteration complexity (\ie sub-quadratic in the size of the underlying graph).

%
%
\subsection{Proof of \cref{thm:AEW}}
\label{sec:proof_adaweight}

\adaweight is \emph{not} merely a “convex combination” of the two non-adaptive optimal algorithms (\acceleweight and \expweight). For this reason, the proof of \cref{thm:AEW} is technically involved. Particularly, the primal-dual averaging method that we use in \adaweight create sequences of filtration-dependent step-sizes; while this is the key element allowing \adaweight to achieve the best-of-both-worlds, previously known approaches in deriving the convergence analyses (as in non-adaptive algorithmic schemes \textendash\ \expweight and \acceleweight \textendash\ and/or in \unixgrad) are not applicable in this case. To handle this new challenge, we propose a completely new way to treat the learning rate in order to make the derived bounds summable. To ensure the comprehensibility, in this section, we first present a sketch of proof explaining the high-level idea before presenting the complete proof with all technical details.

\subsubsection{Sketch of proof (\cref{thm:AEW})}
Let \mbox{$	\reg_{\nRuns}(\flowbase) \defeq \sum_{\run=1}^{\nRuns}{\curr[\step]}\inner{\nabla \meanpot (\state^{\run})}{\curr[\precom ] - \flowbase}$}, the starting point of our proof is the following result: %
\begin{equation}
	\ex \bracks*{\gap \parens{{\nRuns}}} = \ex \bracks*{\meanpot( \state^{\nRuns})-\meanpot(\flowbase)} \leq {2 \ex \bracks*{\reg_{\nRuns} \parens*{\flowbase}}} \big/ {\nRuns^{2}} \textrm{ for any } \nRuns. \label{eq:lem:proof_adaEW}
\end{equation}
\eqref{eq:lem:proof_adaEW} is a standard result in working with the $\curr[\step]$-weighted averaging technique that also appears in several previous works \citep{cutkosky2019anytime,KLBC19}. For the sake of completeness, we provide the proof of \eqref{eq:lem:proof_adaEW} in \ref{sec:app_proof_lem_proof_adaEW}.

Recall that as \adaweight is run, $\grad \meanpot(\curr[\state])$ is \emph{not} observable and only the cost $\curr[\signal]$ is observed and used. Therefore, instead of $\reg_{\nRuns}(\flowbase)$, we now \emph{focus on the term} $\regstoch_{\nRuns}(\flowbase) \defeq  \sum_{\run=1}^{\nRuns} {\curr[\step]}\inner*{\signal^{\run}}{\curr[\precom]-\flowbase}$. In the stochastic case, we can prove that $\regstoch_{\nRuns}(\flowbase)$ is an unbiased estimation of \mbox{$\reg_{\nRuns}(\flowbase)$} (and in the static case, they coincide). Therefore, any upper-bound of $\regstoch_{\nRuns}(\flowbase)$ can be translated via \eqref{eq:lem:proof_adaEW} into an upper-bound of the left-hand-sides of \eqref{eq:AEW-stoch}. The key question becomes \emph{``Which upper-bound of $\regstoch_{\nRuns}(\flowbase)$ can be derived to guarantee the convergence rates in \eqref{eq:AEW-stoch}?''}

To answer the question above, we note that \adaweight is built on the dual extrapolation template with two phases: the test phase and the recommendation phase. This allows us to upper-bound $\regstoch_{\nRuns}(\flowbase)$ in terms of the ``distance'' between the pivot primal points in these phases (\ie $\curr[\ptest]$ and $\curr[\precom]$) and the difference between the costs measured in these phases (\ie $\curr[\testsignal]$ and $\curr[\signal]$). Particularly, we have
\begin{equation}
	\regstoch_{\nRuns}(\flowbase) \le  \sum_{\run=1}^{\nRuns} \func_{\textrm{primal}} (\next[\learn]) {\norm{\curr[\precom]-\curr[\ptest]}_1^{2}} + \sum_{\run=1}^{\nRuns} \func_{\textrm{dual}}(\next[\learn]) {(\curr[\step])^2 \norm*{\curr[\signal]-\curr[\testsignal]}_{\infty}^{2}}, \label{eq:proof_AEW_sketch}
\end{equation}
where $\func_{\textrm{primal}}(\next[\learn])$ and  $\func_{\textrm{dual}}(\next[\learn])$ are certain functions that also depend on other parameters of the game. Importantly, the terms in the right-hand-size of \eqref{eq:proof_AEW_sketch} are actually summable: they are bounded by the summation of two terms $\sum_{\run=1}^{\nRuns} \func(\curr[\learn]) (\curr[\step])^2  \norm{\curr[\signal]-\curr[\testsignal]}_{\infty}^{2}$ and  $\sqrt{ \sum_{\run=1}^{\nRuns} (\curr[\step])^2 \norm{[\curr[\signal]-\curr[\testsignal]]  -   [\grad \meanpot(\curr[\state]) - \grad \meanpot (\curr[\test]) ] }_{\infty}^{2}}$. Here, the first term arises when we bound $\norm{\curr[\precom]-\curr[\ptest]}$ using the smoothness of the BMW potential $\pot$. Moreover, by our \emph{special choice of adaptive sequence of learning rate} $\curr[\learn]$, there exists $\nRuns_0 \ll \nRuns$ such that only the first $\nRuns_0$ components in this summation are positive and hence, this summation is actually of order $\bigoh(1)$ (\ie it does not depend on $\nRuns$). The second term quantifies the error of the observed costs with respect to the actual gradients of the potential. This term is of order $\bigoh \parens{\sdev \nRuns^{3/2}}$ when we choose $\curr[\step] =\run$. Combine these results with \eqref{eq:lem:proof_adaEW}, we obtain the convergence rates indicated in \eqref{eq:AEW-stoch}.

\subsubsection{Proof with technical details (\cref{thm:AEW})}
In this section, we work with the \emph{entropy regularizer} $\hreg(\flow) \defeq \sum _{\pair \in \pairs}\sum _{ \route\in \routes^\pair}\flow^{\pair}_{\route}\log(\flow^{\pair}_{\route}), \forall \flow \in \flows$ which can be used to define the exponential weights template. Trivially, $\hreg$ is $1/\strong$-\emph{strongly convex} on $\flows$ \wrt the $\norm{\cdot}_1$ norm, where we define $\strong  \defeq \massmax \nPairs$. We will also work with the \emph{Fenchel conjugate} of $\hreg$, defined as 	$\hreg^{\ast}(\dpoint) \defeq  \max_{\flow \in \flows} \inner*{\flow}{\dpoint} - \hreg(\flow)$, and its \emph{Fenchel coupling}, defined as $\fench \parens{\flow, \dpoint} = \hreg \parens*{\flow} + {\hreg^*} \parens*{\dpoint} - \inner*{\flow}{\dpoint}$ for any $\flow \in \flows$ and any $\dpoint \in \R^{\nRoutes}$. For the sake of brevity, we will also define $\min \hreg \defeq \min_{\flowalt \in \flows} \hreg(\flow)$ and denote the $\norm{\cdot}_1$-diameter of $\flows$ by $\diam \defeq \max_{\flow, \flowalt \in \flows} \norm{\flow - \flowalt}_1$. We also denote the Kullback–Leibler divergence between two flows $\flow$ and $\flowalt$ by $\dkl{ \flow}{\flowalt }$. Note also that throughout this section, when we mention Algorithm~\ref{alg:AEW}, we understand that it is run with $\curr[\step] \defeq \run$ as chosen in \cref{thm:AEW}. 

First, we prove the following proposition (corresponding to Equation \eqref{eq:proof_AEW_sketch} in the proof~sketch):
\begin{proposition}\label{propo:energy_inequ}
	Run Algorithm~\ref{alg:AEW}, for any $\flow \in \flows$, we have:
	\begin{equation}
		\regstoch_{\nRuns}(\flowbase) \le \hreg(\flowbase) -  \min \hreg +   \constdual \sum_{\run=1}^{\nRuns}  (\curr[\step])^2 \learn^{\run+1} \norm*{\curr[\signal]-\curr[\testsignal]}_{\infty}^{2} - \sum_{\run=1}^{\nRuns}\frac{1}{2\strong\learn^{\run+1}}  \norm{\curr[\precom]-\curr[\ptest]}_1^{2}  , \label{eq:Propo_energy_inequ}
	\end{equation}
where $\constdual \defeq {\hreg(\flowbase) -  \min \hreg +\frac{\strong^2 + 2\diam^2}{2\strong} }$.
\end{proposition} 

\paragraph{Proof of \cref{propo:energy_inequ}.} Consider an ``intermediate'' point \mbox{$\mirrorvec^{\run} \defeq \logit\parens*{\curr[\learn] \dstate^{\run+1}}$}, we focus on the terms $\curr[\step] \inner{\curr[\signal]}{\curr[\mirrorvec]  - \flowbase } $ and $\curr[\step] \inner{\curr[\signal]}{ \curr[\precom] - \curr[\mirrorvec]  } $. These terms sum up to $\curr[\step] \inner{\curr[\signal]}{\curr[\precom]  - \flowbase }$ which defines $\regstoch_{\nRuns}(\flowbase)$. 
%
%
First, from the update rule of $\curr[\dstate]$ in Algorithm \ref{alg:AEW}, we have 
	\begin{align}
		 {\curr[\step]}\inner*{\curr[\signal]}{\mirrorvec^{\run} \!-\! \flowbase}  
		 \!=\! \frac{1}{\curr[\learn]}\braket{ \curr[\learn]\dstate^{\run} \! - \! \curr[\learn]\dstate^{\run+1}}{\mirrorvec^{\run}-\flowbase}  &= \frac{1}{\curr[\learn]}\bracks*{- \hreg(\mirrorvec^{\run}) \!-\! \hreg^{\ast} \parens*{ \curr[\learn] \dstate^{\run+1} } + \inner*{\flowbase}{\curr[\learn] \dstate^{\run+1} - \curr[\learn] \dstate^{\run}} + \inner*{\mirrorvec^{\run}}{\curr[\learn] \dstate^{\run}}}  \nonumber \\
		  %
		   %
		   & \le \frac{1}{\curr[\learn]}\bracks*{\fench(\flowbase,\curr[\learn]\dstate^{\run})  \!-\!  \fench(\flowbase,\curr[\learn]\dstate^{\run+1}) \!-\!   \dkl{\mirrorvec^{\run}}{\curr[\ptest]}}. \nonumber
	\end{align}
	Here, the last inequality comes from the definition of Fenchel coupling and the fact that \mbox{$\curr[\ptest] = \logit(\curr[\learn] \curr[\drecom] )$}. Now, apply the three-point inequality with Bregman divergence (Lemma 3.1 of \cite{CT93}) to the KL-divergence,
	
	%
	\begin{equation*}
				\frac{1}{\learn^{\run}}\dkl{\mirrorvec^{\run}}{\curr[\ptest]} -  \frac{1}{\learn^{\run}}\dkl{\curr[\precom]}{\curr[\ptest]}  - \frac{1}{\learn^{\run}}\dkl{\mirrorvec^{\run}}{\curr[\precom]} = \inner{\grad \hreg(\ptest)  - \grad \hreg(\precom)}{\curr[\precom] - \curr[\mirrorvec]}   
				\ge \curr[\step] \inner*{\curr[\testsignal] }{\curr[\precom] - \mirrorvec^{\run} }.	
	\end{equation*}
	%
	%
	%
 Combining the two inequalities derived above, we have:
	\begin{align}
	& \regstoch_{\nRuns}(\flowbase) \nonumber
	=	 \sum_{\run=1}^{\nRuns} \curr[\step] \inner{\curr[\signal]}{\curr[\mirrorvec]  - \flowbase } + \run \inner{\curr[\signal]}{ \curr[\precom] - \curr[\mirrorvec]  } \nonumber \\ 
	 \le	&	\underbrace{\sum_{\run=1}^{\nRuns}\frac{1}{\curr[\learn]} \bracks*{\fench(\flowbase,\curr[\learn]\dstate^{\run}) \!-\!  \fench(\flowbase, \curr[\learn]\dstate^{\run+1})} }_{\textstyle \defeq \const_1}    + 
	 \underbrace{\bracks*{ \sum_{\run=1}^{\nRuns} {\curr[\step]}\inner*{\curr[\signal]- \curr[\testsignal]}{ \curr[\precom] \!-\! \mirrorvec^{\run}} 
	 		 \!-\! 
	 		 \sum_{\run=1}^{\nRuns}  \frac{1}{\curr[\learn]}\dkl{\mirrorvec^{\run}}{\curr[\precom]}}}_{\textstyle \defeq \const_2} 
		-  \underbrace{\sum_{\run=1}^{\nRuns}\frac{1}{\curr[\learn]}\dkl{\curr[\precom]}{\curr[\ptest]}}_{\textstyle \defeq \const_3}. \label{eq:lem:energy1}  
	\end{align}

Now, we look for upper-bounds of the three terms in the right-hand-side of \eqref{eq:lem:energy1}. First, we trivially have
	%
	\begin{align}
		\const_1
		%
		%
		%
		\le \frac{1}{\learn^{\start}}\fench(\flowbase,\learn^{\start} \dstate^{\start}) + \parens*{\frac{1}{\learn^{\nRuns+1}}-\frac{1}{\learn^{\start}}} \bracks*{\hreg(\flowbase) -  \min \hreg}
		%
		%
		= \frac{1}{\learn^{\nRuns+1}} \bracks*{\hreg(\flowbase) -  \min \hreg}. \label{eq:AEW-first_term}
	\end{align}
	%
	%

	Second, for any $\nRuns$, from the Cauchy-Schwarz inequality and the fact that $	\norm*{\state-\statealt}_1\norm*{\dstate-\alt\dstate}_{\infty}=\min_ {\real>0}\braces*{\frac{\real}{2}\norm*{\state-\statealt}_1^{2}+\frac{1}{2\real}\norm*{\dstate-\alt\dstate}_{\infty}^{2}}$ for any $\state, \statealt, \dstate,\alt\dstate \in \mathbb{R}^{\dims}$,\footnote{Consider the function \mbox{$\psi(\real)=\frac{\real}{2}\norm*{\state-\statealt}_1^{2}+\frac{1}{2\real}\norm*{\dstate - \alt\dstate}_{\infty}^{2}$}; then $\real^* =\norm*{\state-\statealt}_1\norm*{\dstate-\alt\dstate}_{\infty}$ is the minimizer of $\psi$.} we have

	\begin{align}
		&  {\curr[\step]}\braket{\curr[\signal] \! -\! \curr[\testsignal]}{\mirrorvec^{\run} \!-\! \curr[\precom]}  \le     \norm*{ \curr[\step]( \curr[\signal] \!-\!   \curr[\testsignal])}_{\infty}   \norm*{ \mirrorvec^{\run} \!- \! \curr[\precom] }_{1} \le      { \frac{(\curr[\step])^{2} \strong\learn^{\run+1}}{2}  \norm*{\curr[\signal]-\curr[\testsignal]}_{\infty}^{2} + \frac{1}{2\strong\learn^{\run+1}}  \norm*{\mirrorvec^{\run}-\curr[\precom]}_{1}^{2}}. \label{eq:AdaEW_det_term2_part1} 
	\end{align} 
	%
	%
Combine this with the strong convexity of $\hreg$, we obtain:
	\begin{align}
		 \const_2  \le  \frac{\strong}{2}\sum_{\run=1}^{\nRuns}{(\curr[\step])}^{2} \learn^{\run+1} \norm*{\curr[\signal] \!-\! \curr[\testsignal]}_{\infty}^{2} + \frac{1}{2\strong}  \sum_{\run=1}^{\nRuns} \parens*{\frac{1}{\learn^{\run+1}}  -  \frac{1}{\learn^{\run}} } \norm{\curr[\mirrorvec]  - \curr[\ptest]}_1^2 
		\le \frac{\strong}{2}\sum_{\run=1}^{\nRuns}{(\curr[\step])}^{2} \learn^{\run+1} \norm*{\curr[\signal]-\curr[\testsignal]}_{\infty}^{2} + \frac{\diam^2}{2 \strong}  \parens*{\frac{1}{\learn^{\nRuns+1}}  -  1 } . 
		\label{eq:AdaEW_det_term2_conclude}
	\end{align}

	Third, recall the notion $\diam$ denoting the $\norm{}_1$-diameter of $\flows$, we have
	\begin{equation}
		\const_3 \ge  \sum_{\run=1}^{\nRuns}  \frac{1}{2 \strong \learn^{\run}}  \norm{ \curr[\precom] - \curr[\ptest]}_{1}^2  =  \sum_{\run=1}^{\nRuns}\frac{1}{2\strong\learn^{\run+1}}  \norm{\curr[\precom]-\curr[\ptest]}_1^{2}  - \frac{\diam^2}{2\strong}  \parens*{\frac{1}{\learn^{\nRuns+1}}-1} . \label{eq:AEW_term_3}
	\end{equation}
	
	Moreover, from the update rule of $\next[\learn]$ and Lemma 2 of \cite{KLBC19} (also presented in \cite{LYC18,MS10}), we have:
	
	\begin{align*}
		\frac{1 }{ \learn^{\nRuns+1}} =    \sqrt{ \stepada +  \sum_{\run=1}^{\nRuns} {(\curr[\step])}^2 \norm{    \curr[\signal] -    \curr[\testsignal]}_{\infty}^2}	\le   \sum_{\run=1}^{\nRuns}  \frac{ {(\curr[\step])}^2\norm{\curr[\signal] -   \curr[\testsignal] }_{\infty}^2}{\sqrt{\stepada + \sum_{\runalt=1}^{\run} {(\step^{\runalt})}^2 \norm{   \iter[\signal] -   \iter[\testsignal]   }_{\infty}^2} } + \stepada =   \sum_{\run=1}^{\nRuns} {(\curr[\step])}^2 \learn^{\run+1} \norm{   \curr[\signal] -  \curr[\testsignal]}_{\infty}^2      + \stepada. 
	\end{align*}
	
	Apply this inequality to \eqref{eq:AEW_term_3} then apply it with \eqref{eq:AEW-first_term}, \eqref{eq:AdaEW_det_term2_conclude} into \eqref{eq:lem:energy1}, we obtain precisely \eqref{eq:Propo_energy_inequ}. \qed\\

 Now, let us define the noises of the observed costs $\curr[\testsignal]$ and $\curr[\signal]$, induced by the flow profiles $\curr[\test]$ and $\curr[\state]$, in comparison with the corresponding gradients of the BMW potential $\meanpot$ as follows:
 
 \begin{equation*}
 	\curr[\noisetest] \defeq \curr[\testsignal] - \nabla \meanpot \parens*{\curr[\test]} \textrm{ and } \curr[\noise] \defeq \curr[\signal] - \nabla \meanpot \parens*{\curr[\state]}. 
 \end{equation*} 
 By the definition of $\sdev$ (cf. \cref{eq:noisevar}), we have $\sup \braces*{\exof*{\dnorm{\curr[\noisetest]}^2 }, \exof*{\dnorm{\curr[\noise]}^2 } } \le \sdev^{2}$. From these definitions, we also deduce that $\ex \bracks*{\noisetest^{\run} \middle| \prev[\mathcal{H}] }  = \ex \bracks*{\noise^{\run} \middle|\prev[\mathcal{H}]} = \zer$ where $\prev[\mathcal{H}] = \braces*{ \state^{\run-1}, \test^{\run-1}, \sample^{\run-1}, \ldots, \state^{\start},\test^{\start}, \sample^{\start} } $ is the filtration up to time epoch $\run-1$. 
 
\cref{eq:Propo_energy_inequ} involves the difference of the actual gradients of the BMW potential $\nabla \meanpot( \curr[\state])-\nabla \meanpot(\curr[\test])$ and the difference of costs at $\state$ and $\test$, \ie $\obs^{\run}-\obsave^{\run}$ (in the static case, these differences are equal). For the sake of brevity, we define the following terms in order to analyze the gap between these differences:
\begin{equation*}
	 \curr[\mindiff]=\min\braces*{\norm{\nabla \meanpot( \curr[\state])-\nabla \meanpot(\curr[\test])}_{\infty}^{2},\norm{\obs^{\run}-\obsave^{\run}}_{\infty}^{2}} \textrm{ and }\curr[\diff]=\left[\obs^{\run}-\obsave^{\run}\right]-\left[ \nabla \meanpot( \curr[\state])-\nabla \meanpot(\curr[\test])\right].
\end{equation*}
We aim to construct the upper-bounds the last two terms in the right-hand-side of \eqref{eq:Propo_energy_inequ} in terms of $\curr[\mindiff]$ and $\curr[\diff]$. Particularly, from \eqref{eq:Propo_energy_inequ}, we can prove the following proposition:
\begin{proposition}\label{lem:AEW_diff}
	Run Algorithm~\ref{alg:AEW} and define $\tilde{\learn}^{\run} \defeq {1} \big/{\sqrt{\stepada+ 2\sum_{\runalt=1}^{\run-1} (\step^\runalt)^{2} \iter[\mindiff]}}$, we have
	\begin{equation}
		\regstoch_{\nRuns}(\flowbase) \le \hreg(\flowbase) -  \min \hreg +  \conststoch 2\sqrt{2} \parens*{\sum_{\run=1}^{\nRuns}(\curr[\step])^{2} \dnorm{\xi_{\run}}^{2}}^{1/2} + \sum_{\run=1}^{\nRuns} \func(\tilde{\learn}^{\run+1}) (\curr[\step])^{2} \mindiff^{\run},\label{eq:AEW_stoch_conclude}
	\end{equation}
	where $\conststoch \defeq \hreg(\flowbase) -  \min \hreg +\frac{\strong^2 + 3\diam^2}{2\strong} $ and $\func(\tilde{\learn}^{\run+1}) \defeq 4 \conststoch \tilde{\learn}^{\run+1}  -  \frac{1}{8 \smooth^2 \strong \tilde{\learn}^{\run+1}}$. 
\end{proposition}
%

%
%
\para{Proof of \cref{lem:AEW_diff}} First, by definitions of $\curr[\mindiff]$ and $\curr[\diff]$, we have $	\dnorm{\obs^{\run}-\obsave^{\run}}^{2} \le  2\mindiff^{\run} + 2\dnorm{\curr[\diff]}^{2}$. Apply this and Lemma 2 of \cite{KLBC19}, we have 
%
%
%
%
%
%
%
%
%
\begin{align}
	\sum_{\run=1}^{\nRuns}(\curr[\step])^{2}\learn^{\run+1}\dnorm{\obs^{\run}-\obsave^{\run}}^{2}   
	&=   \sum_{\run=1}^{\nRuns} \frac{(\curr[\step])^{2} \dnorm{\obs^{\run}-\obsave^{\run}}^{2}}{ \sqrt{\stepada + \sum_{\runalt=1}^{\run} \step^\runalt \dnorm{\obs^{\runalt}-\obsave^{\runalt}}^{2} }}
	\notag\\
	& \le  2\sqrt{\stepada +\sum_{\run=1}^{\nRuns}(\curr[\step])^{2}\dnorm{\obs^{\run} - \obsave^{\run}} ^{2}} \nonumber\\
	%
	%
	&\le	\sqrt{\stepada + 2\sum_{\run=1}^{\nRuns}(\curr[\step])^{2}\curr[\mindiff]} + \sqrt{2\sum_{\run=1}^{\nRuns}(\curr[\step])^{2} \dnorm{\curr[\diff]}^{2}}\nonumber\\
	&\le 4\sum_{\run=1}^{\nRuns}\frac{(\curr[\step])^{2}\curr[\mindiff]}{\sqrt{\stepada + 2\sum_{s=1}^{\run}(\step^{\runalt})^{2}\curr[\mindiff]}}+ \sqrt{2\sum_{\run=1}^{\nRuns}(\curr[\step])^{2} \dnorm{\curr[\diff]}^{2}}.\label{eq:AEW_stoch_infty2}
\end{align}
%
%
%

On the other hand, from the definition of $\tilde{\learn}^{\run}$, we have $\frac{1}{\tilde\learn^{\run}}\leq \frac{1}{\learn^{\run}}, \forall \run$ and hence,

\begin{align}
	\sum_{\run=1}^{\nRuns}\frac{1}{\tilde{\learn}^{\run+1}}\norm*{\curr[\ptest]-\curr[\precom]}_1^{2}  \leq   \sum_{\run=1}^{\nRuns}  \parens*{\frac{1}{\learn^{\run+1}} - \frac{1}{\learn^{\run}} } \norm*{\curr[\precom]-\curr[\ptest]}_1^{2}  + \frac{1}{\learn^{\run}}\norm*{\curr[\precom] - \curr[\ptest]}_1^{2}  \le  \parens*{\frac{1}{\learn^{\nRuns+1}} - 1} \diam^2 + \sum_{\run=1}^{\nRuns}\frac{1}{\learn^{\run}}\norm*{\curr[\precom] - \curr[\ptest]}_1^{2}. \nonumber
\end{align}
Combine this with the fact that $\curr[\precom]-\curr[\ptest] = \frac{\step^{\run+1}}{2} \parens*{\curr[\state]-\curr[\test]}$ (from Lines \ref{line:AEW-flowave} and \ref{line:AEW-recom-flow} of Algorithm~\ref{alg:AEW}) and choose an increasing sequence of $\curr[\step]$ (such as $\curr[\step] = \run$), we have
\begin{align}
	- \sum_{\run=1}^{\nRuns}\frac{1}{2\strong \learn^{\run+1}} \norm{\curr[\precom]-\curr[\ptest]}_1^{2} 
	%
	\le &  - \frac{1}{2 \strong} \sum_{\run=1}^{\nRuns}\frac{1}{\tilde{\learn}^{\run+1}}\norm*{\curr[\precom] - \curr[\ptest]}_1^{2}  + \frac{\diam^2}{2\strong} \parens*{\frac{1}{\learn^{\nRuns+1}} - 1} \nonumber\\
	\le &  - \frac{1}{2 \strong} \sum_{\run=1}^{\nRuns}\frac{\step^{\run+1}}{2\tilde{\learn}^{\run+1}}\norm*{\curr[\state] - \curr[\test]}_1^{2}  + \frac{\diam^2}{2\strong} \parens*{\frac{1}{\learn^{\nRuns+1}} - 1} \nonumber\\
	\le  & - \frac{1}{2 \strong} \sum_{\run=1}^{\nRuns}\frac{1}{\tilde{\learn}^{\run+1}} \frac{(\step^{\run+1})^2 \dnorm*{\nabla \meanpot \parens{ \curr[\state]} - \nabla \meanpot \parens{\curr[\test]}}^{2}}{4 \smooth^2} + \frac{\diam^2}{2\strong} \parens*{\frac{1}{\learn^{\nRuns+1}} - 1}  \nonumber\\
	\le   & - \frac{1}{8 \smooth^2 \strong} \sum_{\run=1}^{\nRuns}\frac{(\curr[\step])^{2}\mindiff^{\run}}{\tilde{\learn}^{\run+1}}   + \frac{\diam^2}{2\strong} \parens*{\frac{1}{\learn^{\nRuns+1}} - 1}. \label{eq:ADaEW_Stoch_3term0}
\end{align}

Apply \eqref{eq:AEW_stoch_infty2} and \eqref{eq:ADaEW_Stoch_3term0} into \eqref{eq:Propo_energy_inequ}, we obtain \eqref{eq:AEW_stoch_conclude} and finish the proof of \cref{lem:AEW_diff}. \qed \\

We now focus in the last term of \eqref{eq:AEW_stoch_conclude}. We denote \mbox{$\nRuns_0 \defeq \max\braces*{1\leq \run\leq \nRuns: \tilde{\learn}^{\run+1}\geq \bracks*{32 \smooth^2 \strong \conststoch}^{-1/2} }$}. Then for any $\run \ge \nRuns_0$, we have \mbox{$\func(\tilde{\learn}^{\run+1})  < 0$}. In other words, in $\sum_{\run=1}^{\nRuns} \func(\tilde{\learn}^{\run+1}) (\curr[\step])^2 \mindiff^{\run}$, only the first $\nRuns_0$ components are positive. Therefore, 
\begin{align}
	\sum_{\run=1}^{\nRuns} \func(\tilde{\learn}^{\run+1}) (\curr[\step])^2 \mindiff^{\run} 
		\le 	\sum_{\run=1}^{\nRuns_0} \func(\tilde{\learn}^{\run+1}) (\curr[\step])^2 \mindiff^{\run}  
		\le &  \sum_{\run=1}^{\nRuns_0}  4 \conststoch \tilde{\learn}^{\run+1} (\curr[\step])^2 \mindiff^{\run} \nonumber\\
		= &  4 \conststoch \cdot \sum_{\run=1}^{\nRuns_0}  \frac{(\curr[\step])^2 \mindiff^{\run}}{\sqrt{1+ 2 \sum_{\runalt=1}^{\run} (\step^{\runalt})^2 \mindiff^{\runalt} }  } \nonumber\\
		\le & 4 \conststoch \cdot 2 \sqrt{1+ 2 \sum_{\run=1}^{\nRuns_0} (\step^{\run})^2 \mindiff^{\run} }  \tag*{(from Lemma 2 of \citep{KLBC19}) }\nonumber\\
		= & \frac{8}{\tilde{\learn}^{\nRuns_0 }} \conststoch 	\nonumber \\
		\le & 32 \smooth \sqrt{2 \strong} (\conststoch)^{3/2}. \label{eq:AEW_stoch_second_term} 
\end{align}

Combine \eqref{eq:AEW_stoch_second_term}, \eqref{eq:AEW_stoch_conclude} and the fact that $\ex \bracks*{  \dnorm{\xi_{\run}}^{2}} \le    2\ex \bracks*{  \dnorm{ \noise^{\run}  }^{2}} +  2\ex \bracks*{  \dnorm{  \noiseave^{\run}   }^{2}} \le  4 \sdev^2$, we~have
\begin{align}
	\exof*{\regstoch_{\nRuns}(\flowbase) } \le \hreg(\flowbase) -  \min \hreg  + \sdev  \conststoch 4\sqrt{2 \sum \nolimits_{\run=1}^{\nRuns} (\curr[\step])^2 }   +   32 \smooth \sqrt{2 \strong} (\conststoch)^{3/2}.\label{eq:AEW_conclusion}
\end{align}

Finally, in order to apply \eqref{eq:lem:proof_adaEW}, we need to make the connection between $\exof*{\reg_{\nRuns}(\flowbase)}$ and $\exof*{\regstoch_{\nRuns}(\flowbase) }$. Particularly, we have:
\begin{align}
	\exof*{\reg_{\nRuns}(\flowbase)} = \exof*{\regstoch_{\nRuns}(\flowbase) } -  \exof*{\sum_{\run=1}^{\nRuns} {\curr[\step]}\inner*{ \signal^{\run} - \grad \pot(\curr[\state]) }{\curr[\precom]-\flowbase}}  & =  \exof*{\regstoch_{\nRuns}(\flowbase) } - \exof*{\sum_{\run=1}^{\nRuns} {\curr[\step]}\inner*{\curr[\noise]}{\curr[\precom]-\flowbase}} \nonumber\\
	& = \exof*{\regstoch_{\nRuns}(\flowbase) }. \label{eq:AEW_end1}
\end{align}
Here, the last equality comes from the fact that $\exof*{{\curr[\step]}\inner*{\curr[\noise]}{\curr[\precom]-\flowbase}} = \exof*{ {\curr[\step]}\inner*{\exof*{\curr[\noise]|\prev[\mathcal{H}] }}{\curr[\precom]-\flowbase}}= 0$ (by law of total expectation). Combine \eqref{eq:AEW_end1} and \eqref{eq:AEW_conclusion}, then apply \eqref{eq:lem:proof_adaEW} with the choice of step-size $\curr[\step] = \run, \forall \run$, we have:
\begin{align}
	\ex \bracks*{\meanpot( \state^{\nRuns})-\meanpot(\flowbase)} \le \frac{1}{\nRuns^2} \parens*{\hreg(\flowbase) -  \min \hreg  + \sdev \conststoch 4\sqrt{2} \nRuns^{3/2}  +   32 \smooth \sqrt{2 \strong} (\conststoch)^{3/2}} . \label{eq:AEW_proof}
\end{align}
Finally, recall that $\hreg(\flowbase) \le \masssum \log(\massmax)$, $ -\min \hreg \le  \masssum \log(\nRoutes/\masssum)$, $\strong \defeq \nPairs \massmax$, $\diam \le 2 \masssum$, $\sdev \le 2 \ub$ and the definition of $\conststoch$ (cf. \cref{lem:AEW_diff}), we have $\conststoch \le \frac{1}{2}\nPairs \massmax \bracks*{2\log \parens*{\frac{\nRoutes \massmax}{\masssum}} + 13 } \defeq \frac{1}{2}\const$ for the constant $\const$ defined in \cref{thm:AEW}. Plug these results into \eqref{eq:AEW_proof}, we obtain \eqref{eq:AEW-stoch} and conclude the proof.

\section{\adapush: Adaptive learning with efficient per-iteration complexity}
\label{sec:adapush}
In this section, our main focus is to design an algorithm that not only has the adaptive optimality of \adaweight but also has a per-iteration complexity that is scalable in the network's size. To do this, in \cref{sec:distr_setup}, we first introduce an alternative routing paradigm, called \aclp{dis}. In \cref{sec:adapush-algo} and \cref{sec:adapush_theo}, we then propose the \adapush algorithm and analyze its convergence properties and its per-iteration complexities.

In the sequel, we will additionally use the following set of notation. For any (directed) edge $\edge \in \edges$, we denote by $\tail{\edge}$ and $\head{\edge}$ the tail vertex and head vertex of $\edge$ respectively, \ie $\edge$ goes from $\tail{\edge}$ to $\head{\edge}$. For any vertex $\vertex$ and any sub-graph $\graph^\pair$, let $\In{\vertex}{\pair}$ and $\Out{\vertex}{\pair}$ respectively denote the set of incoming edges to $\vertex$ and the set of outgoing edges from $\vertex$, and let $\child{\vertex}{\pair}$ and $\parent{\vertex}{\pair}$ denote the set of its direct successors and the set of its direct predecessors.
 For the sake of brevity, we also write $\dim^{\pair}_{\vertex} \defeq   \abs*{ \Out{\vertex}{\pair}}$ which is the out degree of vertex $\vertex$ in $\graph^\pair$ and $\dim_{\vertex} \defeq \sum_{\pair \in \pairs} \dim^{\pair}_{\vertex}$. 

\subsection{\Aclp{dis} routing paradigm}
\label{sec:distr_setup}

On the defined network structure (cf. \cref{sec:setup}), let us consider a novel route-recommendation paradigm that unfolds as follows: 

\begin{enumerate}
	\item The control interface determines at each vertex $\vertex$ a profile $\disflow^{\pair}_{\vertex} \in \Delta( \Out{\vertex}{\pair})$ for each \ac{OD} pair $\pair$. 
	\item It then makes a route recommendation for each $\pair$-type traffic by choosing the vertices (in $\graph^{\pair}$) one by one such that if the vertices $\source{\pair}, \vertex_{1}, \ldots, \vertex_{k}$ is chosen, then a vertex $\vertex_{k+1} \in \Out{\vertex_{k}}{\pair}$ will be chosen with probability~$\disflow^{\pair}_{\vertex_{k},\edge}$ where $\edge$ is the edge going from $\vertex_{k}$ to $\vertex_{k+1}$.
\end{enumerate}
%
%
%
%
%
%
%
%
In this paradigm, the recommendations are made based on ``local'' decisions at each vertex. In average, among the mass of $\pair$-type traffic arriving at a vertex $\vertex$, the proportion that is routed on edge $\edge \in \Out{\vertex}{\pair} $ is precisely $\disflow^{\pair}_{\vertex, \edge}$. We formalize this idea via the notion of \acl{dis} defined as~follows.

\para{\Acl{dis}}  
For each pair $\pair$ and vertex $\vertex$, we define the set $\disflows_{\vertex}^{\pair}
\defeq \setdef[\big]{\disflow_{\vertex}^{\pair} \in[0,1]^{\dim^{\pair}_\vertex}}{\sum \nolimits_{\edge \in \Out{\vertex}{\pair}}  \disflow^{\pair}_{\vertex,\edge} = 1}$. We call the elements of $\disflows_{\vertex}^{\pair}$ the \textbf{\aclp{dis}} corresponding to $\pair$ and $\vertex$. We also use the notations $\disflows_{\vertex} \defeq \prod_{\pair \in \pairs} \disflows^{\pair}_{\vertex}$ for the set of \acl{dis} profiles chosen at $\vertex$ across all \ac{OD} pairs, $\disflows^{\pair} \defeq \prod_{\vertex \in \vertices} \disflows^{\pair}_{\vertex}$ for the set of \acl{dis} profiles of $\pair$-type traffic across the vertices and $\disflows \defeq \prod_{\vertex \in \vertices} \disflows_{\vertex} = \prod_{\pair \in \pairs} \disflows^{\pair} $ for the set of all \acl{dis}~profiles.

\para{Traffic loads induced by \aclp{dis}}
 When the network's traffic is routed according to a \acl{dis} profile $\disflow \in \disflows$, we denote the \emph{mass} of $\pair$-type traffic arriving at $\vertex$ by $\ar^{\pair}_{\vertex}(\disflow)$ and denote the \emph{load} of $\pair$-type traffic routed on $\edge$ by $\disload^{\pair}_{\edge} (\disflow)$. Formally, they can be computed (recursively) as follows:
\begin{equation}
\disload^{\pair}_{\edge} (\disflow) = \disflow^{\pair}_{\vertex, \edge} \cdot \ar^{\pair}_{\vertex} (\disflow)
\textrm{,~and~}
\ar^{\pair}_{\vertex}(\disflow) \defeq \left\{
\begin{array}{ll}
	\mass{\pair} & \textrm{, if } \vertex = \source{\pair}  \\
	 \sum_{ \edgealt \in \In{\vertex}{\pair}} \disload^{\pair}_{\edgealt} (\disflow) & \textrm{, if } \vertex \neq \source{\pair} \\
\end{array} 
\right.	\label{eq:def_load_mass}
\end{equation}
%
%
We define the \textbf{total load} induced on $\edge$ as \mbox{$\disload_{\edge} \parens*{ \disflow } = \sum_{\pair \in \pairs} \disload^{\pair}_{\edge} (\disflow)$}.

The \textbf{congestion cost} on each edge $\edge$ is determined by the cost function $\late_\edge$ as defined in \cref{sec:setup}. Particularly, $\late_{\edge} \parens*{ \disload_{\edge}(\disflow) , \sample }$ is the cost on $\edge$ at state $\sample \in \samples$ when the \acl{dis} profile $\disflow \in \disflows$ is implemented.

Finally, corresponding to the local routing paradigm described above, we re-adjust the \textbf{learning model} at each time epoch $\run = \running$ as~follows: 
\begin{enumerate}
	\item
	The navigation chooses a \acl{dis} profile $\curr[\disflow] \in \disflows$ and routes the traffic accordingly.
	\item
	Concurrently, the state $\sample^{\run}$ of the network is drawn from $\samples$ (\acs{iid} relative to $\prob$).
	\item
	The congestion costs $\late_{\edge} \parens*{ \disload_{\edge}(\curr[\disflow]) , \curr[\sample] }$ on each edge $\edge$ is observed. 
\end{enumerate}
\subsection{\adapush algorithm}
\label{sec:adapush-algo}


In this section, we propose a new equilibrium learning method called \adapush. This method will be implemented in the \acl{dis} learning model described in the previous section.

To build \adapush, our starting point is the \adaweight method. In Algorithm~\ref{alg:AEW}, there are two major types of flows updating: the flows $\curr[\ptest]$ and $\curr[\precom]$ outputting from the logit mapping $\logit$, and the flows $\curr[\test]$ and $\curr[\state]$ obtained from averaging between two pre-computed flows. In Algorithm~\ref{alg:AEW}, it takes $\bigoh(\nRoutes)$ computations to update each of these flow profiles in each iteration. To improve the per-iteration complexity, we can leverage the weight-pushing technique \citep{TW03,GLLO07}: by assigning weights on edges and use dynamic programming principles, we can find the \acl{dis} profiles inducing the same loads (and costs) as the logit-mapped flows $\curr[\ptest]$ and $\curr[\precom]$. However, \textbf{the weight-pushing technique fails to derive \acl{dis} profiles matching the averaged flows} $\curr[\test]$ and $\curr[\state]$ of Algorithm~\ref{alg:AEW}. Particularly, weight-pushing two flows and then taking the average will not incur the same costs as implementing the average of the two flows. Recall that these averaging steps are the key elements allowing the adaptability of \adaweight, the key challenge now is to implement efficiently these~steps.

Facing up this challenge, we make the following observation: \emph{while the averaged flows of \adaweight are not ``weight-pushable'', their induced loads are}. In this perspective, there arise two new challenges. First, we need to efficiently compute the load profiles induced by the averaged flows (without explicitly computing these flows). To do this, we introduce a sub-routine called \textbf{pulling-forward} \textendash\ reflecting the fact that we start from the origin and work forward (unlike the classical weight-pushing that starts from the destination and works backward). Second, we need to derive \acl{dis} profiles that matches these averaged loads. We refer to this as \textbf{matching-load}. Note that \emph{pulling-forward and matching-load are completely novel contributions}.

Taking an overall view, \adapush is a combination of the weight-pushing technique (to handle the logit-mapped flows) and the pulling-forward and matching-load procedures (to handle the averaged flows) into the \adaweight template. For the sake of conciseness, we first present a sub-algorithm in \cref{sec:aux_subroutine} and present a pseudo-code form of \adapush in \cref{sec:pseudo_adapush}.\footnote{As a side note on \acceleweight (cf. \cref{sec:XLEW}), a direct application of weight-pushing also fails to achieve an efficient implementation for the same reasons of that of \adaweight. Note that we can use techniques that we introduce in \adapush to improve \acceleweight and achieve an efficient version. As \acceleweight is not the focus of our work, we omit the details.}

%
%
\subsubsection{\ppm sub-algorithm}
\label{sec:aux_subroutine}
In this section, we present a sub-algorithm, called \ppm, that combines the three auxiliary routines: pushing-backward, pulling-forward and matching-load. In Algorithm~\ref{alg:scorepush}, we give a pseudo-code form of $\ppm$. It takes four inputs: an ``anchor local-load'' profile $\Anc = (\Anc^{\pair}_{\edge})_{\pair \in \pairs, \edge \in \edges^{\pair}}$, a ``local-weight'' profile $\Weight = (\Weight^{\pair}_{\edge})_{\pair \in \pairs, \edge \in \edges^{\pair}}$; and two scalar numbers $\curr[\step]$ and $\sum_{\runalt=0}^{\run} \iter[\step]$ (we use these notations to facilitate the presentation of \adapush in \cref{sec:adapush-algo}). When it finishes, Algorithm~\ref{alg:scorepush} outputs a \acl{dis} profile and updates the anchor local-load profile.

\begin{algorithm}[htbp]
	\DontPrintSemicolon
	\small
	
	\SetKw{Return}{Return} 
	\SetKwBlock{OneofPhases}{}{end}
	\SetKwBlock{Phase}{}{end}
	\SetKwProg{Parallel}{In parallel for all}{ do}{end}
	\SetKwProg{Wait}{wait until}{ then do}{end}

	\KwIn{$\Anc, \Weight \in \R^{\sum_{\pair}\sum_{\vertex} \dim^{\pair}_{\vertex}}$ and $\curr[\step],  \sum_{\runalt = 0}^{\run} \step^{\runalt} > 0$} 
	\KwOut{\rm $\disflow \in \disflows$ and update $\Anc$}
	\Parallel{$\pair \in \pairs$}{
		Fix an topological order $\vertex_0, \vertex_{1}, \ldots, \vertex_{\abs{\vertices^{\pair}}}$ of the graph $\graph^{\pair}$ (such that $\vertex_0 \defeq \source{\pair}$ and $\vertex_{\abs{\vertices^{\pair}}} \defeq \sink{\pair}$)\;
		\tcp*[l]{\small {Pushing-backward phase}} 
		\For{$\vertex = \vertex_{\abs{\vertices^{\pair}}} , \ldots, \vertex_{0}$}{
			\lIf(\tcp*[f]{initialize $\scorebw{}{}{}$ at destination}){$\vertex = \sink{\pair}$}{set $\scorebw{\vertex}{}{} \gets 1$ \label{line:scorepush_cond}}
			\lElse(\tcp*[f]{compute $\scorebw{}{}{}$ based on $\scorebw{}{}{}$ of children}){set $\scorebw{\vertex}{}{} \gets \sum_{\edge \in \Out{\vertex}{\pair}}{\scorebw{\head{\edge}}{ }{\Weight} \cdot  \exp \parens*{\Weight^{\pair}_{\edge} } }$ \label{line:SP-score}} 
			%
			%
			set {$\precom^{\pair}_{\vertex, \edge} \gets \exp \parens*{\Weight^{\pair}_{\edge}} {\scorebw{\head{\edge}}{ }{\Weight}}  \big/ {\scorebw{\vertex}{ }{\Weight}}$ for $\edge \in \Out{\vertex}{\pair}$} \tcp*[r]{compute pivot outgoing flow} \label{line:SP-dflow}
			%
			%
		} 
		\tcp*[l]{\small {Pulling-forward phase}} 
		\For{$\vertex = \vertex_{0}, \ldots, \vertex_{\abs{\vertices^{\pair}}}$}{
			\lIf(\tcp*[f]{initialize $\scorefw{}{}{}$ at origin}){$\vertex = \source{\pair}$}{set $\scorefw{\vertex}{ }{\precom^{\pair}} \gets \mass{\pair}$ \label{line:flowpull_cond}}
			\lElse(\tcp*[f]{compute $\scorefw{}{}{}$ from loads of incoming edges}){
				set $\scorefw{\vertex}{ }{\precom^{\pair}} \gets \sum_{\alt\edge \in \In{\vertex}{\pair}}{ \disload^{\pair}_{\alt\edge}(\precom) }$ \label{line:FP-score}} 
			\For{$\edge \in \Out{\vertex}{\pair}$}{
				set $\disload^{\pair}_{\edge} (\precom) \gets \scorefw{\vertex}{ }{\precom^{\pair}} \cdot \precom^{\pair}_{\vertex, \edge}$
				\tcp*[r]{pull the loads forward} \label{Line:Adapush_disload_precom}
				set $ \disload^{\pair}_{\edge}( \disflow)  \gets  \bracks*{\step^{\run} \disload^{\pair}_{\edge}\parens*{\precom} +   \Anc^{\pair}_{\edge}   } \big/ \sum_{\runalt=0}^{\run} \step^{\runalt} $ \tcp*[r]{average $\&$ pull forward}\label{Line:Adapush_disload_disflow}
				set $\Anc^{\pair}_{\edge} \gets \Anc^{\pair}_{\edge} + \curr[\alpha] \cdot \disload^{\pair}_{\edge}(\precom)$ \tcp*[r]{update the anchor}
			}
			%
		} 

		\tcp*[l]{\small {Matching-load phase}} 
		\For{$\vertex = \vertex_{0}, \ldots, \vertex_{\abs{\vertices^{\pair}}}$}{
			\lIf(\tcp*[f]{initialize mass arriving at origin}){$\vertex = \source{\pair}$}{set $\ar^{\pair}_{\vertex} ({\disflow}) \gets \mass{\pair}$}
			\lElse(\tcp*[f]{compute mass that would arrive at $\vertex$ if $\disflow$ is implemented}){
				set $\ar^{\pair}_{\vertex} ({\disflow}) \gets \sum_{\edgealt \in \In{\vertex}{\pair}}{ \disload^{\pair}_{\edgealt}(\disflow) }$} 
			%
			set $\disflow^{\pair}_{\vertex, \edge} \gets \disload^{\pair}_{\edge}(\disflow) \big/ \ar^{\pair}_{\vertex} ({\disflow})$ for $\edge \in \Out{\vertex}{\pair}$\tcp*[r]{compute outgoing flow matching $\disload^{\pair}_{\edge}(\disflow)$} \label{Line:ppm_disflow}
			%
			%
			%
			%
		} 

	}  

	\caption{$\ppm(\Anc, \Weight, \curr[\step], \sum_{\runalt=0}^{\run} \step^{\runalt} )$
	}\label{alg:scorepush}
\end{algorithm}

Particularly, when we run \ppm, it executes the following three phases for each \ac{OD} pair $\pair$:

\begin{enumerate}
	\item \textit{The pushing-backward phase}: in this phase, we consider vertices in $\graph^{\pair}$ one by one in a \emph{reversed topological order},  \ie we work backwardly from the destination to the origin.\footnote{Note that since $\graph^{\pair}$ are DAGs, such an topological order exists and can be found in $\bigoh{\abs{\edges^{\pair}}}$ time.} For each vertex $\vertex$, we assign a score \textendash\ called the \emph{backward score} (denoted by $\scorebw{}{}{\Weight}$ in Algorithm~\ref{alg:scorepush}) \textendash\ depending on the weights of $\vertex$'s outgoing edges (\ie $\Weight^{\pair}_{\edge}, \forall \edge \in \Out{\vertex}{\pair}$) and its children's scores. Then, we ``pushes'' the backward score of $\vertex$ to its parents so that, in their turns, we can compute their scores. Finally, based on the backward score of $\vertex$, we compute a \acl{dis} $\precom^{\pair}_{\vertex} \in \disflows^{\pair}_{\vertex}$. 
	\item \textit{The pulling-forward phase}: in this phase,  we consider vertices in $\graph^{\pair}$ one by one in a \emph{topological order}, \ie we work forwardly from the origin to the destination. On each vertex $\vertex$, we assign a score \textendash\ called the \emph{forward score} (denoted by $\scorefw{}{}{\Weight}$ in Algorithm~\ref{alg:scorepush}) \textendash\ computed from the loads induced by $\precom$ on the incoming edges of $\vertex$ that are ``pulled" from $\vertex$'s parents. Then, for any $\edge \in \Out{\vertex}{\pair}$, we derive $\load^{\pair}_{\edge}(\precom)$ and compute another term \textendash\ called $\load^{\pair}_{\edge}\parens{\disstate}$ \textendash\ by taking the average of $\disload^{\pair}_{\edge}(\precom)$ and the anchor load $\Anc^{\pair}_{\edge}$. Finally, we update the anchor load-profile $\Anc$ (that will be used later in \adapush).  
	\item \textit{The matching-load phase}: in this phase, we once again consider vertices in $\graph^{\pair}$ in a topological order. For each $\vertex$, we compute a \acl{dis} $\disstate^{\pair}_{\vertex} \in \disflows^{\pair}_{\vertex}$ that matches precisely the load $\load^{\pair}_{\edge}\parens{\disstate}, \forall \edge \in \Out{\vertex}{\pair}$ obtained in the pulling-forward phase. To do this, we need to compute $\ar^{\pair}_{\vertex}(\disstate)$ \textendash\ the mass of $\pair$-type traffic arriving at $\vertex$ induced by $\disstate$ \textendash\; this can be done by pulling the loads corresponding to $\disstate^{\pair}$ from the parents of $\vertex$. The \acl{dis} profiles $\disstate^{\pair}, \pair \in \pairs$ constitute the output of \ppm.
\end{enumerate}

\subsubsection{Pseudo-code of \adapush.}
\label{sec:pseudo_adapush}

Having all the necessary preparations, we now present the \adapush method. The key difference that makes \adapush stand out from \adaweight is that it uses the sub-algorithm \ppm to efficiently compute the corresponding \aclp{dis} instead of working with the $\nRoutes$-dimensional flow profiles as in \adaweight. We present a pseudo-code form of \adapush in Algorithm~\ref{alg:adapush}. Similar to \adaweight, \adapush follows two phases: the test phase and the recommendation~phase:

\begin{enumerate}
	\item In the \emph{test phase}, we run \ppm with the anchor load profile $\Anc^{\run}$ and the weight profile $\curr[\learn] \curr[\drecomdis]$ as inputs (here, $\curr[\learn]$ is a learning rate). At the end of the test phase, the cost $\testsignal^{\run}$ at a  ``test'' \acl{dis} is derived and used to update the test-weight profile $\curr[\dtestdis]$.
	\item In the \emph{recommendation phase}, we re-run \ppm but this time, with the test weight $\curr[\learn] \curr[\dtestdis]$ obtained previously in the test phase as input. Moreover, we will also use the output of \ppm to update the anchor load-profile $\Anc^{\run}$. At the end of the recommendation phase, we obtain the \acl{dis} profile $\curr[\disstate]$, route the traffic accordingly, then update the weight profile $\curr[\drecomdis]$ by using the incurred costs. 	
\end{enumerate}

Finally, we observe that in each iteration of Algorithm~\ref{alg:adapush}, it is required to compute the term $\max_{\route \in \routes} \parens*{\sum_{\edge \in \route} \abs{\curr[\signal]_{\edge} - \curr[\testsignal]_{\edge}}}$ for updating the learning rate $\curr[\learn]$. This can also be done in efficiently: for any $\pair \in \pairs$, we solve a weighted-shortest-path problem on $\graph^\pair$ where each edge $\edge$ is assigned with a  weight $\real_{\edge} \defeq - \abs*{\curr[\signal]_{ \edge} - \curr[\testsignal]_{ \edge} }$. For instance, this step can be done by the classical Bellman-Ford algorithm \citep{bellman1958routing,ford1956network} that takes only $\bigoh(\abs{\vertices})$ rounds of computations.\footnote{There exist more complicated shortest-path algorithms that have better complexities (\eg \cite{awerbuch1989randomized,huang2017distributed,elkin2020distributed}); however, the complication of these algorithms are beyond the purpose of our work and they do not improve the complexity of \adapush in general. In this work, we only analyze \adapush with the simple Bellman-Ford algorithm.}

%
%
%

\begin{algorithm}[H]
	\DontPrintSemicolon
	\small
	\rmfamily 
	
	\SetKw{Do}{do} 
	\SetKw{Return}{Return} 
	\SetKwInput{KwIndis}{Input at each vertex $\vertex$} 
	\SetKwProg{RunUntil}{run}{ until it terminates, then }{end}
	\SetKwProg{VertexFor}{For}{, each vertex $\vertex$ does}{end}
	\SetKwProg{VertexWhile}{while}{, $\vertex$ does}{end}
	
	%
	Initialize $\init[\drecomdis] \gets \zer$, $\Anc^{\start}=\zer$, $\beforeinit[\step] \gets 0$ and $\init[\learn] \gets 1$\; 
	\For{$\run = \running$}{
		\tcp{Test phase}
		%
		%
		set $\distest^{\run} \gets$ output of $\ppm(\Anc^{\run}, \curr[\learn] \drecomdis^{\run}, \curr[\step], \sum_{\runalt=0}^{\run} \iter[\step])$ \tcp*[r]{compute a test \acl{dis}}
		get $\testsignal^{\run}_{\edge} \!\gets\! \late_{\edge} \parens*{ \disload_{\edge} \parens*{\disqtest^{\run}}, \sample^{\run} } $ for $\edge \in \edges$ \tcp*[r]{querry the test \acl{dis}}
		set $ \dtestdis^{\run, \pair}_{\edge}   \gets  \drecomdis^{\run, \pair}_{\edge}  -  \step^{\run} \testsignal^{\run}_{\edge}$ for any $\pair \in \pairs$ and $\edge \in \edges^{\pair}$  \tcp*[r]{update test weight}
		%
		%
		
		\tcp{Recommendation phase}
		set $\disstate^{\run}, \Anc^{\run+1} \gets$ outputs of $\ppm(\Anc^{\run}, \curr[\learn] \dtestdis^{\run}, \curr[\step], \sum_{\runalt=0}^{\run} \iter[\step])$ \label{line:adapush-recom} \tcp*[r]{compute \acl{dis} $\&$ update anchor}
		Route according to $\curr[\disstate]$ and get $\signal^{\run}_{\edge} \!\gets\! \late_{\edge} \parens*{ \disload_{\edge} \parens*{\disstate^{\run}}, \sample^{\run} } $ for $\edge \in \edges$  \tcp*[r]{route, measure costs} 
		set $\drecomdis^{\run+1, \pair}_{\edge} \gets  \drecomdis^{\run,\pair}_{\edge}   -  \step^{\run} \signal^{\run}_{\edge}$ for any $\pair \in \pairs$ and $\edge \in \edges^{\pair}$  \tcp*[r]{update weight} 
		%
		\tcp{Update the learning rate}
		set $\next[\learn] \gets 1 \big/ \sqrt{1 + \sum_{\runalt = \beforestart}^{\run} \iter[\step] \max_{\route \in \routes} \parens*{ \sum_{\edge \in \route} \abs{\iter[\signal]_{\edge} - \iter[\testsignal]_{\edge}}}^{2}}$ \label{line:adapush_endDSP}
	} 	
	\caption{\Acl{ALW} (\adapush)} \label{alg:adapush}
\end{algorithm}
\subsection{\textbf{\upshape \adapush}: convergence results and per-iteration complexities}  
\label{sec:adapush_theo}

The convergence properties of \adapush is formally presented in the following theorem:
\begin{theorem}\label{thm:adapush}
	Let $\disstate^{\start}, \disstate^{\afterstart}, \ldots$ be the sequence of \acl{dis} profiles recommended by \adapush (\ie Algorithm \ref{alg:adapush}) running with $\step^{\run}=\run, \forall \run = \running$; then the following results hold:

	\begin{enumerate}
		\item  The sequence of flow profiles $\flow^{\run}, \forall \run = \running$ where $\flow^{\run, \pair}_{\route} \defeq \mass{\pair} \prod_{\edge \in \route} \disflow^{\run, \pair}_{\tail{\edge},\edge}$ for any $\pair \in \pairs$ and $\route \in \routes$ enjoys the following equilibrium convergence rate:
		\begin{equation*}
			\exof*{\gap \parens*{\nRuns}}
			\leq  \bigoh\parens*{ \parens*{\log\nRoutes}^{{3}/{2}} \parens*{\frac{\sdev }{\sqrt{\nRuns}} + \frac{1}{\nRuns^2}  }},
		\end{equation*}
		and specifically, in the static case, it enjoys the convergence rate ${\gap \parens*{\nRuns}}	\leq  \bigoh \parens*{ \log(\nRoutes)^{3/2} / \nRuns^2}$.
		\item Moreover, each iteration of \adapush requires only an $\bigoh\parens*{\abs{\pairs} \abs{\vertices} \abs{\edges}}$ number of computations.
	\end{enumerate}
 
\end{theorem}

Result~\textit{(i)} of \cref{thm:adapush} shows that \adapush also converges toward an equilibrium with the same rate as \adaweight: an $\bigoh((\log\nRoutes)^{3/2}/\sqrt{\nRuns})$ rate in the stochastic regime and an $\bigoh((\log\nRoutes)^{3/2}/{\nRuns}^2)$ rate in the static regime. More precisely, we shall see in the proof of \cref{thm:adapush} that \adapush recommends a sequence of \aclp{dis} that induce the same BMW-potential values as the sequence of flows recommended by \adaweight which, in turn, approaches the potential value at an equilibrium (see \cref{propo:induce_load} below). Note that although the flow $\curr[\flow]$ is mentioned in \cref{thm:adapush}, it is never computed by \adapush and it is not needed for routing the traffic in practice (only the \acl{dis} $\curr[\disflow]$ is~needed). 

Moreover, Result~\textit{(ii)} of \cref{thm:adapush} shows the main difference between using \adapush and using \adaweight: the per-iteration (space and time) complexities of \adapush are polynomial in terms of the network's primitive parameters (numbers of \ac{OD} pairs, numbers of vertices and numbers of edges). Therefore, unlike \adaweight, \adapush can run efficiently even in large networks.

\subsection{Proof of \cref{thm:adapush}}

	\paragraph{First, we prove that $\last[\disflow]$ and $\last[\flow]$ \textendash\ as defined in \cref{thm:adapush} \textendash\ induce the same costs.} In fact, we can prove a stronger result as follows:
	\begin{proposition}\label{propo:induce_load}
		Any \acl{dis} profile $\disflow \in \disflows$ and the flow profile $\flow \in \flows$ such that $\flow^{\pair}_{\route} \defeq \mass{\pair} \prod_{\edge \in \route} \disflow^{\pair}_{\edge}$ induce the same load profiles, \ie $\load_{\edge}(\disflow) = \load_{\edge} (\flow)$ for any $\edge$.
	\end{proposition}
	\begin{proof}{\emph{Proof of \cref{propo:induce_load}.}}
		Fix an \ac{OD} pair $\pair \in \pairs$. Let $\routes^{\source{\pair}}_{\vertex}$ and $\routes^{\vertex}_{\sink{\pair}}$ denotes the set of paths in $\graph^\pair$ going from $\source{\pair}$ to $\vertex$ and the set of paths going from $\vertex$ to $\sink{\pair}$ respectively. Let us denote $\multi^{\pair}_{\vertex} \defeq  \sum_{\route \in \routes^{\vertex}_{\sink{\pair}}}  \prod_{\edge \in \route } \disflow^{\pair}_{\tail{\edge}, \edge}$ for any $\vertex \in \vertices^{\pair} \backslash \{\sink{\pair}\}$ and conventionally set $\multi^{\pair}_{\sink{\pair}} = 1$. By induction (following any topological order of vertices), we can prove that $\multi^{\pair}_{\vertex}= 1$ for any $\vertex \in \vertices^\pair$.
		%
		%
		%
		%
		%
		%
		  %
		  %
		%
		 From this result, the definition of the load profile in \eqref{eq:def_load_mass} and the definition of $\multi^{\pair}_{\vertex}$, the following equality holds for any $\vertex$ and $\edge \in \Out{\vertex}{\pair}$: 
		 \begin{align*}
		 	\load^{\pair}_{\edge}(\disflow) = \ar^{\pair}_{\vertex}(\disflow) \cdot \disflow^{\pair}_{\edge} \cdot \multi^{\pair}_{\head{\edge}}   	& = \mass{\pair} \sum_{\routealt \in \routes^{\source{\pair}}_{\vertex} }  \prod_{\edgealt  \in \routealt}\disflow_{\edgealt}  \cdot \disflow_{\edge}  \cdot \sum_{\route \in \routes^{\head{\edge}}_{\sink{\pair}} }  \prod_{\edgealtalt  \in \route}\disflow_{\edgealtalt} =  \mass{\pair} \sum_{\route \in \routes_\pair \atop \route \ni \edge} \prod_{\edgealt \in \route}\disflow_{\edgealt } = \sum_{\route \in \routes_\pair \atop \route \ni \edge} \flow^{\pair}_{\route} =\load^{\pair}_{\edge}(\flow). 
		 \end{align*} 
	 \qed
	\end{proof}
	Apply \cref{propo:induce_load}, since $\last[\disflow]$ and its corresponding flow $\last[\flow]$ induce the same load profiles, by definition of the cost functions, they induce the same costs, \ie $\late(\last[\flow], \sample) =  \late(\last[\disflow], \sample)$ for any state $\sample \in \samples$.

	\paragraph{Second, we prove that the recommendations of \adapush and \adaweight coincide.} Let us first focus on the recommendation phase of Algorithm~\ref{alg:adapush}. Particularly, assume that at time epoch $\run$, the anchor local-load profile $\curr[\Anc]$ in Algorithm~\ref{alg:adapush} matches the load induced by the anchor flow $\curr[\anchor]$ defined at Line~\ref{line:AEW-anchor} of Algorithm~\ref{alg:AEW} and also assume that the local-weight $\curr[\dtestdis]$ used in Algorithm~\ref{alg:adapush} and the weight $\dtest^{\run}$ defined at Line~\ref{line:AEW-dpointave} of Algorithm~\ref{alg:AEW} satisfy that $\sum_{\edge \in \route} \dtestdis^{\run, \pair}_{\edge} =\dtest^{\run}_{\route}$ for any $\route \in \routes$ and $\pair \in \pairs$. We will prove that the \acl{dis} profile $\curr[\disstate]$ output from $\ppm(\Anc^{\run}, \curr[\learn] \dtestdis^{\run}, \curr[\step], \sum_{\runalt=0}^{\run} \iter[\step])$ induces the same load profiles (and the same costs) as the flow profile $\curr[\state]$ in Algorithm~\ref{alg:AEW}.
	
	To do this, we observe that the following equality holds true for any $\pair \in \pairs$ and $\route \in \routes^{\pair}$ and with the \acl{dis} profile $\precom$ computed in the pushing-backward phase of $\ppm(\Anc^{\run}, \curr[\learn] \dtestdis^{\run}, \curr[\step], \sum_{\runalt=0}^{\run} \iter[\step])$:
	\begin{equation}\label{eq:precomflow}
		\mass{\pair} \prod_{\edge \in \route} \precom^{\pair}_{\tail{\edge}, \edge } = \mass{\pair} \prod_{\edge \in \route} \exp(\dtestdis^{\run, \pair}_{\edge}) \frac{\scorebw{\head{\edge}}{}{} }{ \scorebw{\tail{\edge}}{}{}} = \frac{\exp \parens*{ \sum_{\edge \in \route } \dtestdis^{\run, \pair}_{\edge} }} {\scorebw{\source{\pair}}{}{}} = \frac{\exp(\dtest^{\run}_{\route})}{\sum_{\routealt \in \routes^{\pair}} \exp(\dtest^{\run}_{\routealt})} = \precom^{\run}_{\route}. 
	\end{equation} 
Here, the last equality comes directly from the update rule of $\curr[\precom]$ in Line~\ref{line:AEW-flow12} of Algorithm~\ref{alg:AEW}. 

Now, combine \eqref{eq:precomflow} with \cref{propo:induce_load}, the loads induced by $\precom$ of $\ppm(\Anc^{\run}, \curr[\learn] \dtestdis^{\run}, \curr[\step], \sum_{\runalt=0}^{\run} \iter[\step])$ are precisely the loads induced by $\curr[\precom]$ in Algorithm~\ref{alg:AEW}. Moreover, by definition of local-load (cf. \cref{sec:distr_setup}), we have that the term $\disload^{\pair}_{\edge} (\precom)$ computed at Line~\ref{Line:Adapush_disload_precom} of $\ppm$ are precisely these loads.

Due to the arguments above and the assumption on the relation of $\curr[\Anc]$ and $\curr[\anchor]$, the term $\disload^{\pair}_{\edge} (\disflow)$ computed at Line~\ref{Line:Adapush_disload_disflow} of $\ppm$ is precisely the load induced by the averaged flow $\curr[\state]$ computed at Line~\ref{line:AEW-recom-flow} of Algorithm~\ref{alg:AEW}. Finally, from the updating rule of the \acl{dis} profile $\disflow$ at Line~\ref{Line:ppm_disflow} of \ppm, we deduce that these loads match precisely with the loads induced by $\disflow$ (which is the output of \ppm).

As a conclusion, since we set $\curr[\disflow]$ as the output of $\ppm(\Anc^{\run}, \curr[\learn] \dtestdis^{\run}, \curr[\step], \sum_{\runalt=0}^{\run} \iter[\step])$, the loads (and hence, the costs) induced by $\curr[\disflow]$ are precisely that of the flows $\curr[\state]$ recommended by Algorithm~\ref{alg:AEW}. Using a similar line of arguments, we can also deduce that the \acl{dis} profile $\distest^{\run}$ computed in Algorithm~\ref{alg:adapush} also matches the flow profile $\curr[\test]$ computed in \adaweight. For this reason, the assumptions we made on $\curr[\dtestdis]$ and $\curr[\Anc]$ in the above proof actually hold true.

To sum up, we have proved that the flow profiles defined in \cref{thm:adapush} have the same loads as the \acl{dis} profiles recommended by \adapush that, in turn, coincide with the loads incurred by the recommended flows of \adaweight. This leads to the fact that \adapush (Algorithm~\ref{alg:adapush}) inherits the convergence rates of \adaweight (Algorithm~\ref{alg:AEW}). This concludes the proof of Result~\textit{(i)} of \cref{thm:AEW}.
	

	\paragraph{Finally, we justify the per-iteration complexity of \adapush}. It is trivial to see that when \ppm is run, at any vertex $\vertex$ and \ac{OD} pair $\pair$, each of its phases only takes $\bigoh (\max\braces{ \texttt{In}^{\pair}_{\vertex},\texttt{Out}^{\pair}_{\vertex}})$ rounds of computations where $\texttt{In}^{\pair}_{\vertex}$ and $\texttt{Out}^{\pair}_{\vertex}$ are the in-degree and out-degree of $\vertex$ in $\graph^{\pair}$. As mentioned above, the learning rate update step can be done in $\bigoh(\abs{\vertices})$ time. We conclude that each iteration of \adapush requires only an $\bigoh\parens*{\abs{\pairs} \abs{\vertices} \abs{\edges}}$ number of computations at each vertex. 
	\qed

\section{Numerical Experiments}
\label{sec:numerics}
%

In this section, we report the results of several numerical experiments that we conducted to justify the theoretical convergence results of \ac{AEW} and \ac{ALW}. Particularly, in \cref{sec:toy-example}, we present the experiments on a toy example to highlight the advantages of \ac{AEW} / \ac{ALW} over several benchmark algorithms. Then, in \cref{sec:real-world}, we show the superiority in performance of \ac{ALW} in real-word networks. The experiments with real-world data pose several additional computational challenges; we also discuss these challenges and provide quick-fix solutions. The codes of our experiments are available at \url{dongquan-vu/Adaptive_Distributed_Routing}.

\subsection{Experiments on small-size networks}
\label{sec:toy-example}
We first consider a toy-example of on a network with 4 vertices and 4 edges. Specifically, the edges in this network are arranged to form 2 parallel paths going from an origin vertex, namely $\source{}$, to a destination vertex, namely $\sink{}$. At each time epoch, a traffic demand (\ie inflow) of size $\mass{}=10$ is sent from $\source{}$ to $\sink{}$. 

In the experiments presented below, we run \ac{AEW}, \ac{EW} and \acceleweight in \mbox{$\nRuns_{\max}=10\textrm{e}5$} epochs. We record the values of the BMW potential $\cost$ at outputs of each algorithm at each time $\nRuns=1,2, \ldots, \nRuns_{\max}$. We then identify the value $\cost(\flowbase)$ which is the minimum among all $\cost$-value of flow profiles computed by these algorithms in all iterations; thus, $\flowbase$ represents the equilibrium~flow. In order to track down the convergence properties of these algorithms, we will compute and plot out the evolution (when $\nRuns = 1, 2, \ldots, \nRuns_{\max}$) of the \emph{gaps} (cf. \cref{eq:gap}) of the flows derived from these algorithms that are analyzed in \cref{thm:EW}, \cref{thm:XLEW} and \cref{thm:AEW} (we denote these gaps by $\gap_{\textrm{\ac{AEW}}}(\nRuns)$, $\gap_{\ac{EW}}(\nRuns)$, and $\gap_{\textrm{\ac{XLEW}}}(\nRuns)$ for short).


\begin{figure}[b!]
	\centering
	\begin{subfigure}[b]{.4\linewidth}
		\begin{tikzpicture}
			\node (img){\includegraphics[height = 0.18\textheight]{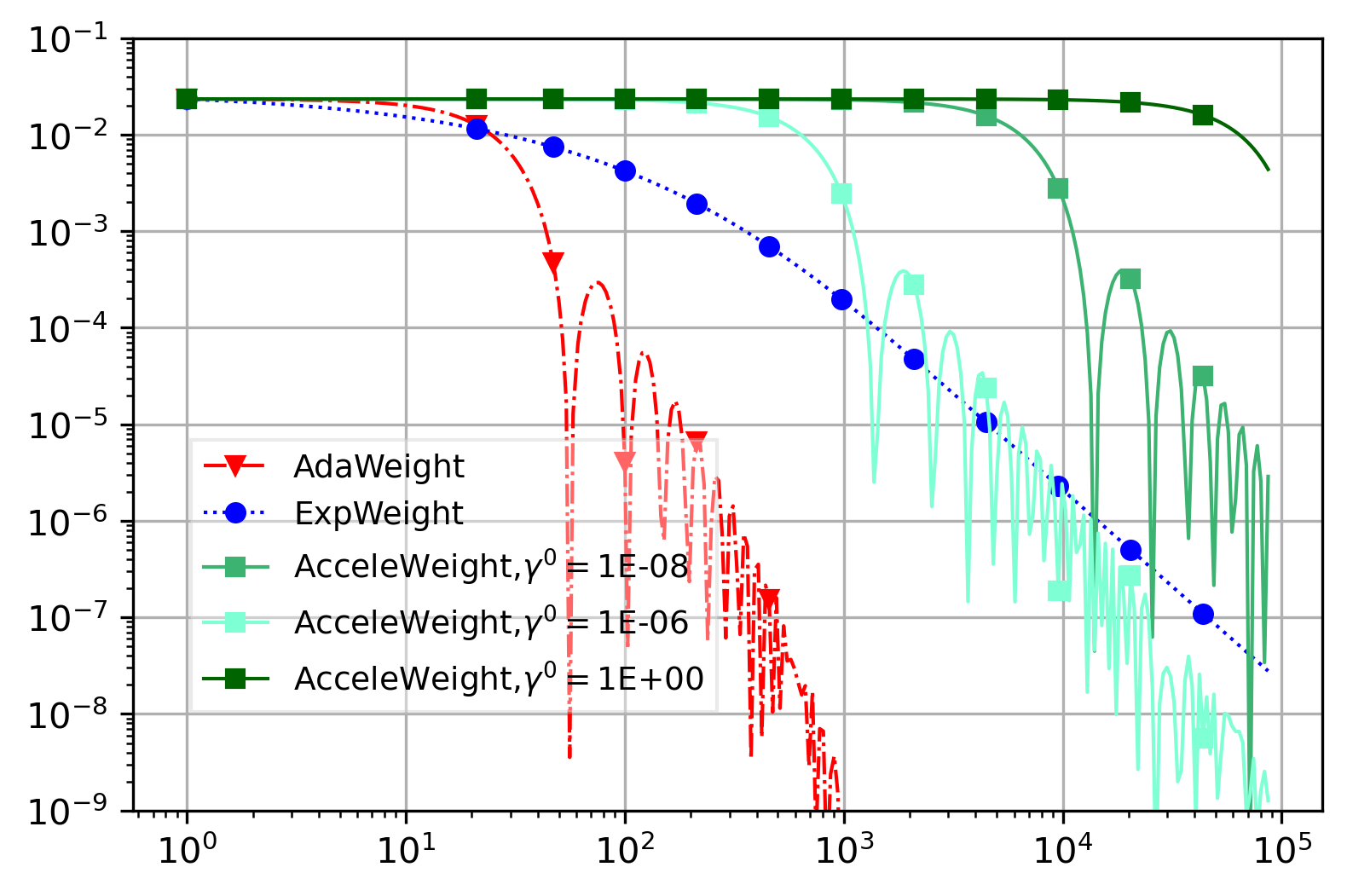}};
			\node[below=of img, node distance=0cm, yshift=1.2cm, xshift =0.2cm] {\scriptsize${\nRuns}$};
			\node[left=of img, node distance=0cm, rotate=90, anchor=center,xshift=0cm, yshift=-1.1cm] {\scriptsize${\gap(\nRuns)}$};
		\end{tikzpicture}
		\caption{\footnotesize Log-log plot of $\gap(\nRuns)$ in static setting}
		\label{fig:parallel_static}
	\end{subfigure}
	\begin{subfigure}[b]{.4\linewidth}
		\begin{tikzpicture}
			\node (img){\includegraphics[height = 0.18\textheight]{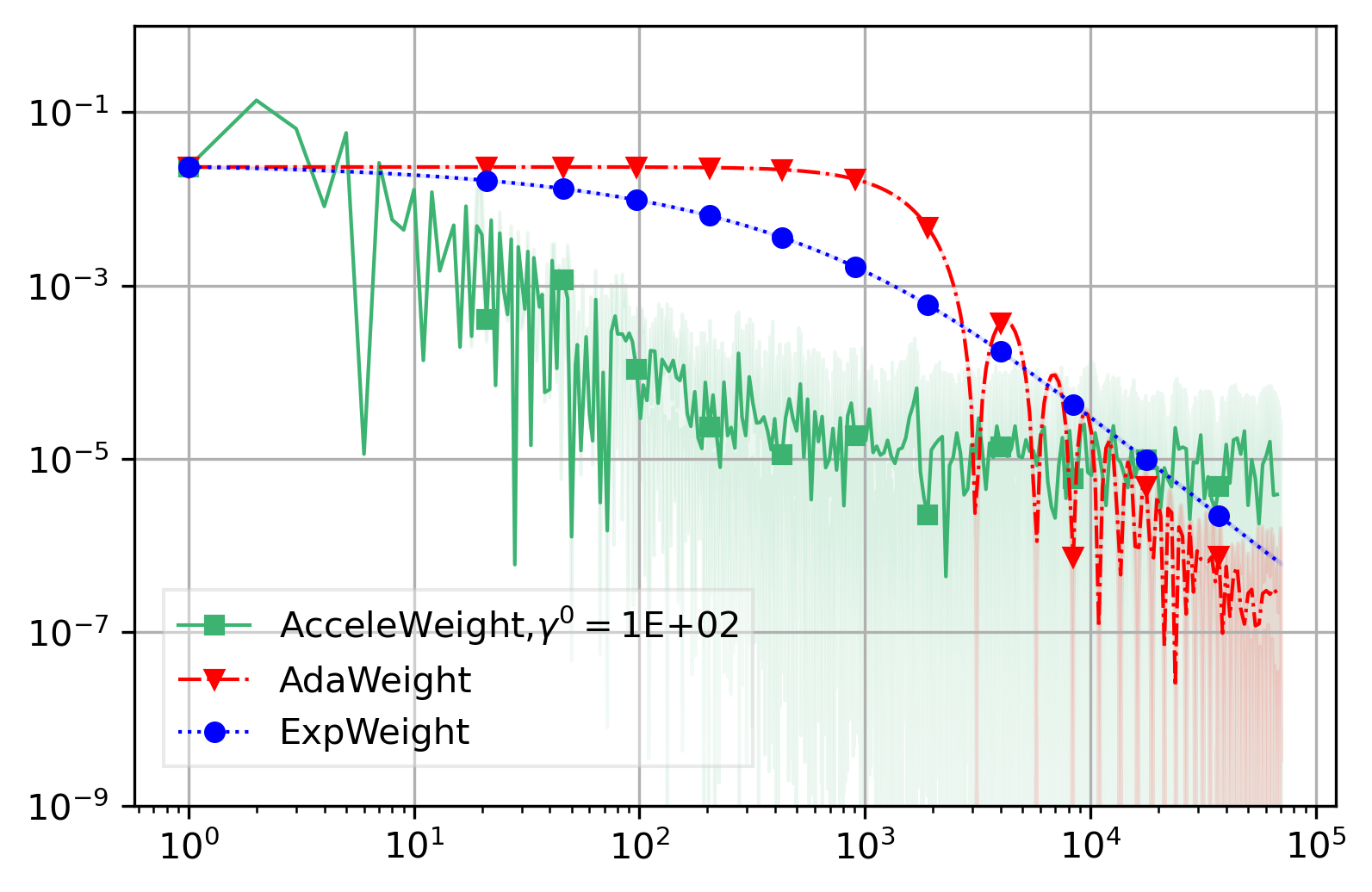}};
			\node[below=of img, node distance=0cm, yshift=1.25cm, xshift =0.2cm] {\scriptsize${\nRuns}$};
			\node[left=of img, node distance=0cm, rotate=90, anchor=center,xshift=0cm, yshift=-1.1cm] {\scriptsize${\gap(\nRuns)}$};
		\end{tikzpicture}
		\caption{\footnotesize  Log-log plot of $\gap(\nRuns)$ in stochastic setting}
		\label{fig:parallel_stochastic}
	\end{subfigure}%
	\vspace{-.5ex}
	\caption{Convergence speed of \ac{AEW}, \ac{EW}, \acceleweight in a parallel network.}
	\vspace{1ex}
	\label{fig:parallel}
\end{figure}

First, we consider a \emph{static environment}. Particularly, the costs on edges are determined via fixed linear cost functions such that when being routed with the same load, one path has a higher cost than the other. We plot out the evolution of $\gap_{\textrm{\ac{AEW}}}(\nRuns)$, $\gap_{\ac{EW}}(\nRuns)$, and $\gap_{\textrm{\ac{XLEW}}}(\nRuns)$ in \cref{fig:parallel_static}. We observe that in this static setting, all these terms converge toward 0; in other words, all three algorithms converge towards equilibria of the game. Importantly, we observe that \ac{AEW} and \acceleweight enjoy an accelerated convergence speed that is much faster than that of \ac{EW}; this coincides with theoretical results in the previous sections. Note that if \acceleweight is run with a badly-tuned initial step-size $\gamma^\beforestart$ that is too small or too large (\eg when $\gamma^\beforestart = 1$ or $\gamma^\beforestart = 1 \textrm{e}-8$), it requires a long warming-up phase and might converge slowly. We also observe that in \cref{fig:parallel}, $\gap_{\ac{EW}}(\nRuns)$ does not seem to fluctuate as much as $\gap_{\textrm{\ac{AEW}}}(\nRuns)$ and $\gap_{\textrm{\ac{XLEW}}}(\nRuns)$, this comes from the fact that $\gap_{\ac{EW}}(\nRuns)$ is computed from the time-averaged flow profiles of outputs of \ac{EW}. We recall that these time-averaged flows are not the flows recommended by \ac{EW}.

Second, we consider a \emph{stochastic} setting. Particularly, the costs on edges are altered by adding noises generated randomly from a zero-mean normal distribution. To have a better representation of these uncertainties, we run each algorithm in 5 different instances of the noises' layouts. We take the averaged results among these instances and plot out the evolution of $\gap_{\textrm{\ac{AEW}}}(\nRuns)$, $\gap_{\ac{EW}}(\nRuns)$ and $\gap_{\textrm{\ac{XLEW}}}(\nRuns)$ in \cref{fig:parallel_stochastic}. In this setting, we observe that while \ac{AEW} and \ac{EW} converge toward equilibria, \acceleweight fails to do so. Moreover, in this stochastic environment, the accelerated rate is no longer obtainable: \ac{AEW} and \ac{EW} converge with the same rate (of order $\bigoh(1/\sqrt{\nRuns})$). This confirms the theoretical results of \ac{AEW}, \ac{EW} and \acceleweight.

\subsection{Experiments on real-world datasets}
\label{sec:real-world}
In this section, we present several experiments using a real-world dataset collected and provided in \cite{dataset2021} (with free license). This dataset contains the road networks of different cities in the world. The congestion model in this dataset is assumed to follow the BPR cost functions whose coefficients are estimated a priori. The purpose of our experiments is to measure the equilibrium convergence of our proposed methods with the presence of uncertainties and no knowledge on the cost functions.

Due to their per-iteration complexity, the naive implementations of \ac{EW}, \acceleweight and \ac{AEW} in Algorithms~\ref{alg:EW}, \ref{alg:XLEW}, \ref{alg:AEW} take an extremely long running time when being run in these real-world networks. Therefore, to deal with these large-scale networks, we consider the distributed model (presented in \cref{sec:distr_setup}) and run \ac{ALW} in place of \ac{AEW}. As a benchmark, we consider a variant of \ac{EW} implemented with the classical weight-pushing technique. These implementations enjoy a per-iteration complexity that is polynomial in terms of the networks' sizes.\footnote{In this section, to improve the visibility of numerical results, we choose to only use \ac{EW} as the benchmark and do not report the performance of \acceleweight that does not guarantee an equilibrium convergence in stochastic settings.}

Before presenting the obtained experimental results (in \cref{sec:num_exp_large}), we first address a computational issue in using weight-pushing ideas in large-scale networks. As far as we know, this issue has not been reported in any previous works (prior to this work, weight-pushing is mostly analyzed theoretically and real-world implementations have not been provided). We formally address this computational issue and introduce a quick-fix solution in \cref{sec:compu_issue} (readers who are eager to see the numerical experiments can skip this section).

\subsubsection{Computational issues of weight-pushing with large weights.}
\label{sec:compu_issue}
 The \ppm sub-algorithm that we constructed in \cref{sec:adapush-algo} involves the pushing-backward phase where a backward score (denoted by $\scorebw{}{}{}$) is assigned on each vertex in the network. In theory, this score can be computed efficiently. In practice though, when the magnitude of costs values are large, the input weights of \ppm is proportional to the negative accumulation of costs (that decreases quickly to $-\infty$) and hence, the computation of $\scorebw{}{}{}$ involves a division of two infinitesimally small numbers; this is often unsolvable by computers.  

 To resolve this issue, instead of keeping track of $\scorebw{}{}{}$, we keep track of its logarithm. Particularly, fix an \ac{OD} pair $\pair$, for each vertex $\vertex$ and each outgoing edge $\edge \in \Out{\vertex}{\pair}$, let us denote $\logscore{\vertex} \defeq \log(\scorebw{\vertex}{}{})$ and $\edgescore{\edge} = \logscore{\vertex} +  \Weight^{\pair}_{\edge} $, then from Line~\ref{line:SP-score} of Algorithm~\ref{alg:scorepush}, it can be computed by
 \begin{align}
 	 	\logscore{\vertex} & = \log \bracks*{\sum_{\edge \in \Out{\vertex}{\pair}} \exp\parens*{\logscore{\vertex} + \Weight^{\pair}_{\edge} }} \nonumber\\
 	 	& = \max_{\edgealt \in \Out{\vertex}{\pair}} (\edgescore{\edgealt}) +  \log \bracks*{\sum_{\edge \in \Out{\vertex}{\pair}} \exp\parens*{\edgescore{\edge} - \max_{\edgealt \in \Out{\vertex}{\pair}} (\edgescore{\edgealt}) } }. \label{eq:compu_issue}
 \end{align}
The expression in \eqref{eq:compu_issue} allows computers to compute $\logscore{\vertex}$ without any issues even when the magnitude of costs and $\Weight^{\pair}_{\edge}$ is large. Finally, the \acl{dis} profile at Line~\ref{line:SP-dflow} of the pushing-backward phase of \ppm can be computed as $\precom^{\pair}_{\vertex, \edge} = {1} \big/{ \exp(	\logscore{\vertex} -   \edgescore{\edge})}$.

\subsubsection{Experimental results.}
\label{sec:num_exp_large}
We consider several instances in the dataset \cite{dataset2021} representing the urban traffic networks of different cities. The primitive parameters of these networks are summarized in \cref{tab:networks}. More information on these networks are given at \url{https://github.com/bstabler/TransportationNetworks}. Note that in previous work, this dataset is used mostly for social-costs optimization; in our knowledge, no work has used this dataset for equilibrium searching problems.

\begin{table}[h!]
	\centering
	\caption{\centering \footnotesize Networks' sizes of several real-world instances}\label{tab:networks}
\begin{tabular}{lccccc}
	\toprule
	 Network name
	& $\#$ of vertices 
	& $\#$ of edges
	& $\#$ of \ac{OD} pairs
	& $\min$ demands
	& $\max$ demands
	\\
	\midrule
	SiouxFalls
	& 24
	& 76
	& 528
	& 100
	& 4400
	\\
	Eastern-Massachusetts
	& 74
	& 258
	& 1113
	& 0.5
	& 957.7
	\\
	 Berlin-Friedrichshain
	& 224
	& 523
	& 506
	& 1.05
	& 107.53
	\\
	Anaheim
	& 416
	& 914
	& 1406
	& 1
	& 2106.5
	\\
	\bottomrule
\end{tabular}
\end{table}

For our experiments, in each network, we run \ac{ALW} and \ac{EW} in $\nRuns_{\max} = 10\textrm{e}4$ time epochs. 


\paragraph{Static environment.} First, we consider the static setting where in each time epoch, for each network, the costs on edges are determined by a BPR function with coefficients given by the dataset. The results are reported in \cref{fig:static}, particularly as follows:

\begin{enumerate}[leftmargin=0pt]
	\item In \cref{fig:SiouxFall_static}, \cref{fig:EM_static}, \cref{fig:Berlin_F_static} and \cref{fig:Anaheim_static}, we plot out the evolution of $\gap_{\ac{ALW}}(\nRuns)$ and $\gap_{\ac{EW}}(\nRuns)$ in different networks. Again, both algorithms converge toward equilibria. Importantly, the results in these large-scale networks show the superiority of \ac{ALW} (and \ac{AEW}) in the convergence speed in comparison with \ac{EW}.  We also make several side-comments as follows. First, the evolution of $\gap_{\textrm{\ac{ALW}}}(\nRuns)$ appears to involve less fluctuations than that in the small-size network considered in \cref{sec:toy-example}. This phenomenon is trivially explicable: in the toy-example with only two paths, any small changes in the implemented flow profile might lead to a large impact on the costs; in real-world instances, moving a mass of traffic from one path to the others does not make so much differences in the total induced costs accumulated from a large number of paths. Second, as we consider networks with larger sizes, both \ac{ALW} and \ac{EW} require a longer warming-up phase (where $\gap(\nRuns)$ does not decrease significantly) before picking up good recommendations and eventually converge. Albeit the case, this warming-up phase is still reasonably small and it does not restrict the applicability of our proposed algorithms. 
	\item To validate the rate of convergence (in terms of $\nRuns$), we plot out the evolution of the terms $\nRuns^2 \cdot \gap_{\ac{ALW}}(\nRuns)$ and $\nRuns^2 \cdot \gap_{\ac{EW}}(\nRuns)$ in \cref{fig:SiouxFall_static_T2}, \cref{fig:EM_static_T2}, \cref{fig:Berlin_F_static_T2} and \cref{fig:Anaheim_static_T2}. In these plots, $\nRuns^2 \cdot \gap_{\ac{ALW}}(\nRuns)$ approach a horizontal line; this confirms that the convergence speed of \ac{ALW} (also of \ac{AEW}) are precisely in order $\bigoh (\nRuns^2)$. On the other hand, \ac{EW} does not reach this convergence order (and hence, $\nRuns^2 \cdot \gap_{\ac{EW}}(\nRuns)$ continue increasing).
\end{enumerate}

\paragraph{Stochastic environment.} In this setting, the cost of edges are altered by adding random noises generated from zero-mean normal distributions. For validation purposes, for each network, we run \ac{ALW} and \ac{EW} in 5 instances (of noises' layout) and report the averaged results across these instances in \cref{fig:stochastic}.  

\begin{enumerate}[leftmargin=0pt]
	\item In \cref{fig:SiouxFall_stoch}, \cref{fig:EM_stoch}, \cref{fig:Berlin_F_stoch} and \cref{fig:Anaheim_stoch}, we plot out the evolution of $\gap_{\ac{EW}}(\nRuns)$ and $\gap_{\textrm{\ac{ALW}}}(\nRuns)$. These terms tend to zero as $\nRuns$ increases; this confirms that \ac{EW} and \ac{ALW} converge toward equilibria in this stochastic setting. Moreover, in this setting, the convergence rate of \ac{ALW} (also of \ac{AEW}) and \ac{EW} are of the same order. 
	\item To justify the convergence speed of the considered algorithms, in \cref{fig:SiouxFall_stoch_T2}, \cref{fig:EM_stoch_T2}, \cref{fig:Berlin_F_stoch_T2} and \cref{fig:Anaheim_stoch_T2}, we plot our the terms $\sqrt{\nRuns} \cdot \gap_{\textrm{\ac{ALW}}}(\nRuns)$ and $\sqrt{\nRuns} \cdot \gap_{\ac{EW}}(\nRuns)$. These terms approach horizontal lines as $\nRuns$ increases. This reaffirms the fact that the speed of convergence of \ac{EW} and \ac{AEW} are $\bigoh(1/\sqrt{\nRuns})$ in the stochastic regime. This result is consistent with our theoretical results.
\end{enumerate}

\paragraph{On the elapsed time of \ac{ALW}.} We end this section with an interesting remark. Even in the largest network instance in our experiments (Anaheim), it only takes \ac{ALW} around \emph{4 seconds to finish one round of learning iteration} and to output a route recommendation (the computations are conducted on a machine with the following specs: Intel Core(TM) i7-9750H CPU \@ 2.60GHz and 8GB RAM). For most of the applications in urban traffic routing, the scale of time between fluctuations of networks' states is often much larger than this elapsed time (it might take hours or even days for a significant change to happen). This highlights the implementability and practicality of our proposed methods, even in networks with much larger~sizes.



\def\sizesub{0.48\linewidth}
\def\figureheight{0.16\textheight}
\def\yshift{1.25cm}
\def\xshift{0.2cm}
\def\ylabelshift{-1.1cm}

\begin{figure}[htb!]
	\centering
	\captionsetup{justification=raggedright, singlelinecheck=false}
	\begin{subfigure}[b]{\sizesub}
		\begin{tikzpicture}
			\node (img){\includegraphics[height = \figureheight]{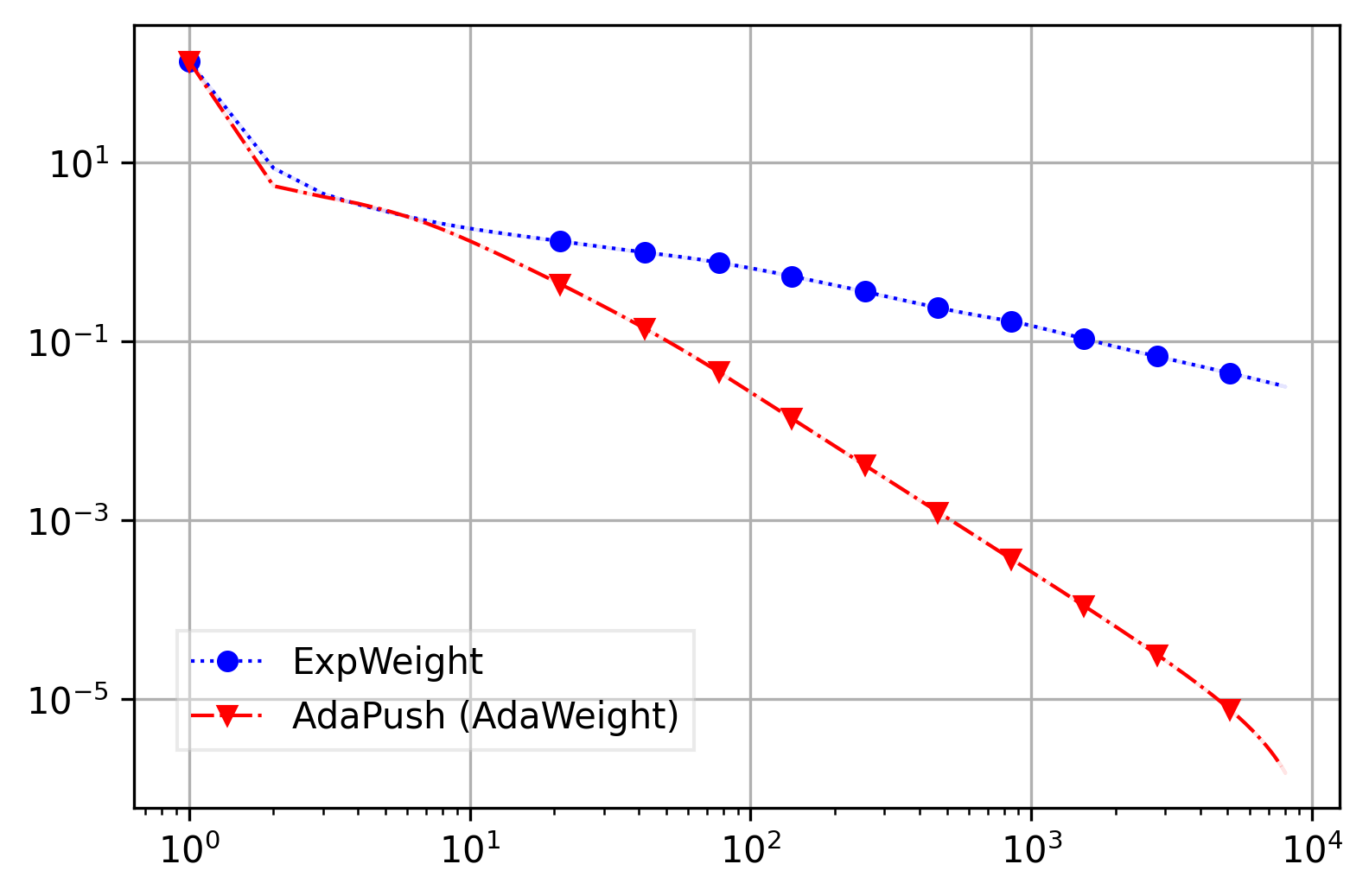}};
			\node[below=of img, node distance=0cm, yshift=\yshift, xshift =\xshift] {\scriptsize${\nRuns}$};
			\node[left=of img, node distance=0cm, rotate=90, anchor=center,xshift=0cm, yshift=\ylabelshift] {\scriptsize${\gap(\nRuns)}$};
		\end{tikzpicture}
		\caption{SiouxFalls: Log-log plot of $\gap(\nRuns)$}
		\label{fig:SiouxFall_static}
	\end{subfigure}%
	\begin{subfigure}[b]{\sizesub}
		\begin{tikzpicture}
			\node (img){\includegraphics[height = \figureheight]{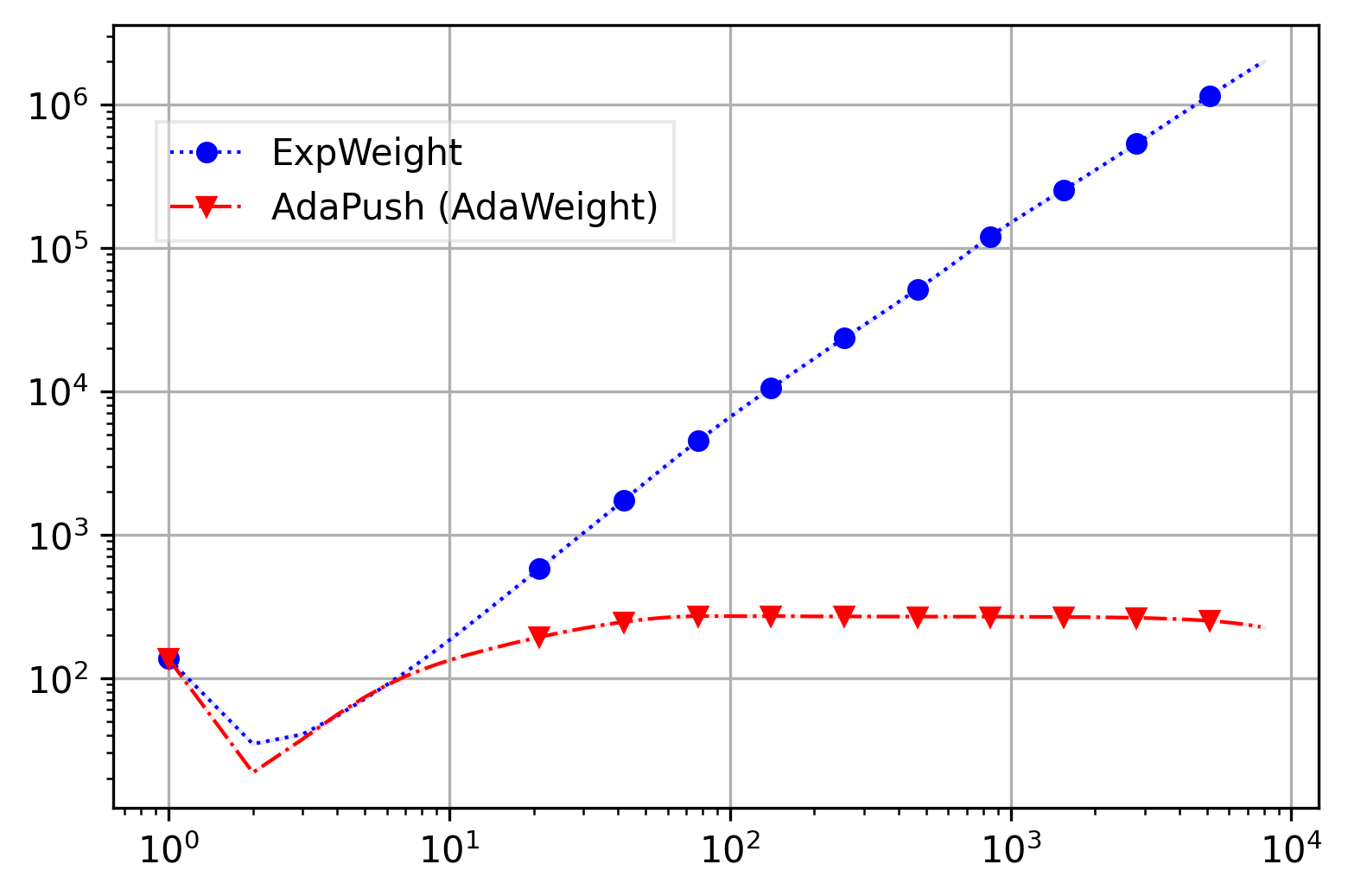}};
			\node[below=of img, node distance=0cm, yshift=\yshift, xshift =\xshift] {\scriptsize${\nRuns}$};
			\node[left=of img, node distance=0cm, rotate=90, anchor=center,xshift=0cm, yshift=\ylabelshift+2mm] {\scriptsize$\nRuns^2 \cdot {\gap(\nRuns)}$};
		\end{tikzpicture}
		\caption{SiouxFalls: Log-log plot of $\nRuns^2 \cdot \gap(\nRuns)$}\label{fig:SiouxFall_static_T2}
	\end{subfigure}


	\begin{subfigure}[b]{\sizesub}
		\begin{tikzpicture}
			\node (img){\includegraphics[height = \figureheight]{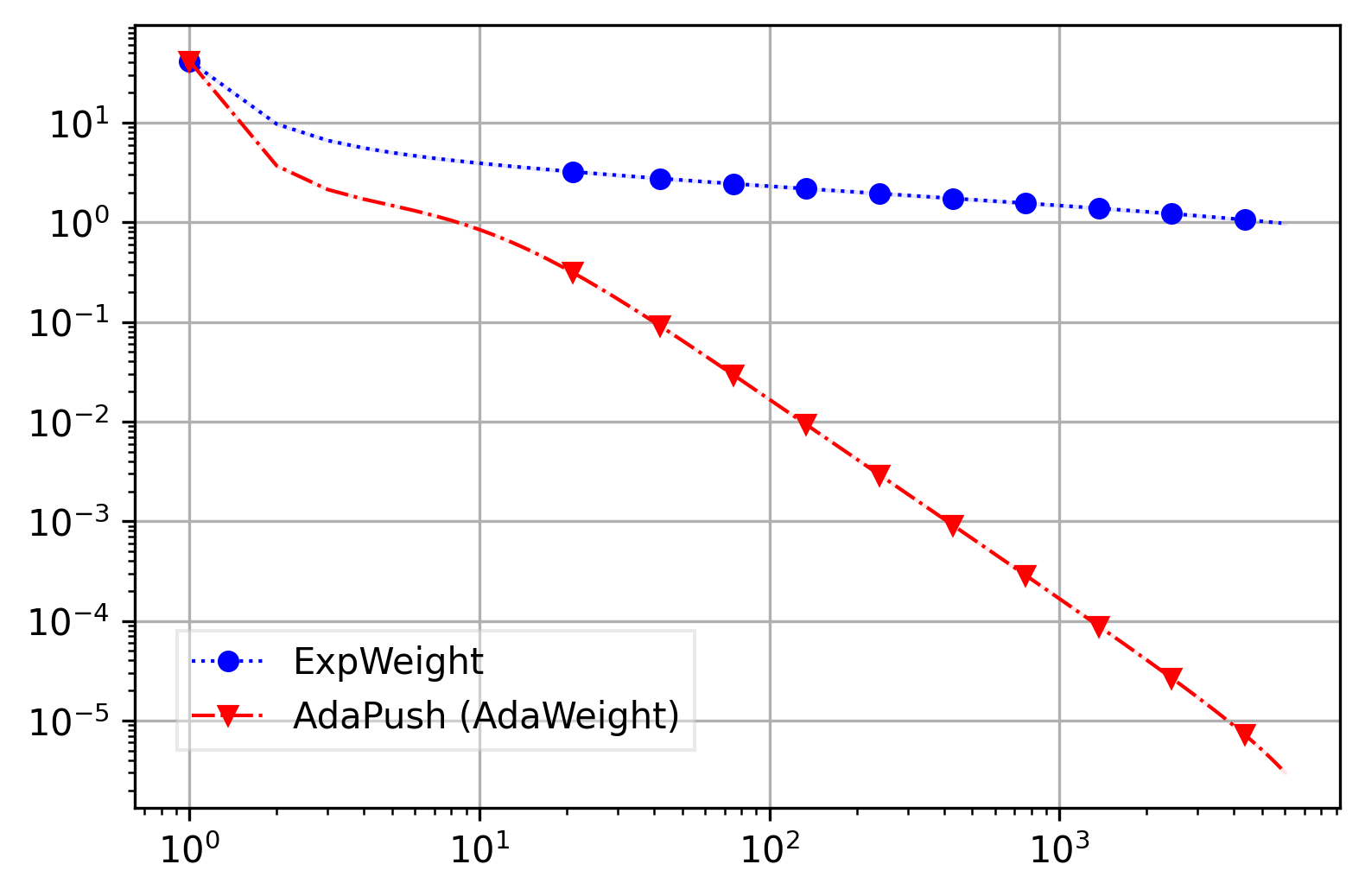}};
			\node[below=of img, node distance=0cm, yshift=\yshift, xshift =\xshift] {\scriptsize${\nRuns}$};
			\node[left=of img, node distance=0cm, rotate=90, anchor=center,xshift=0cm, yshift=\ylabelshift] {\scriptsize${\gap(\nRuns)}$};
		\end{tikzpicture}
		\caption{Eastern-Massachusetts: Log-log plot of $\gap(\nRuns)$}
		\label{fig:EM_static}
	\end{subfigure}%
	\begin{subfigure}[b]{\sizesub}
		\begin{tikzpicture}
			\node (img){\includegraphics[height = \figureheight]{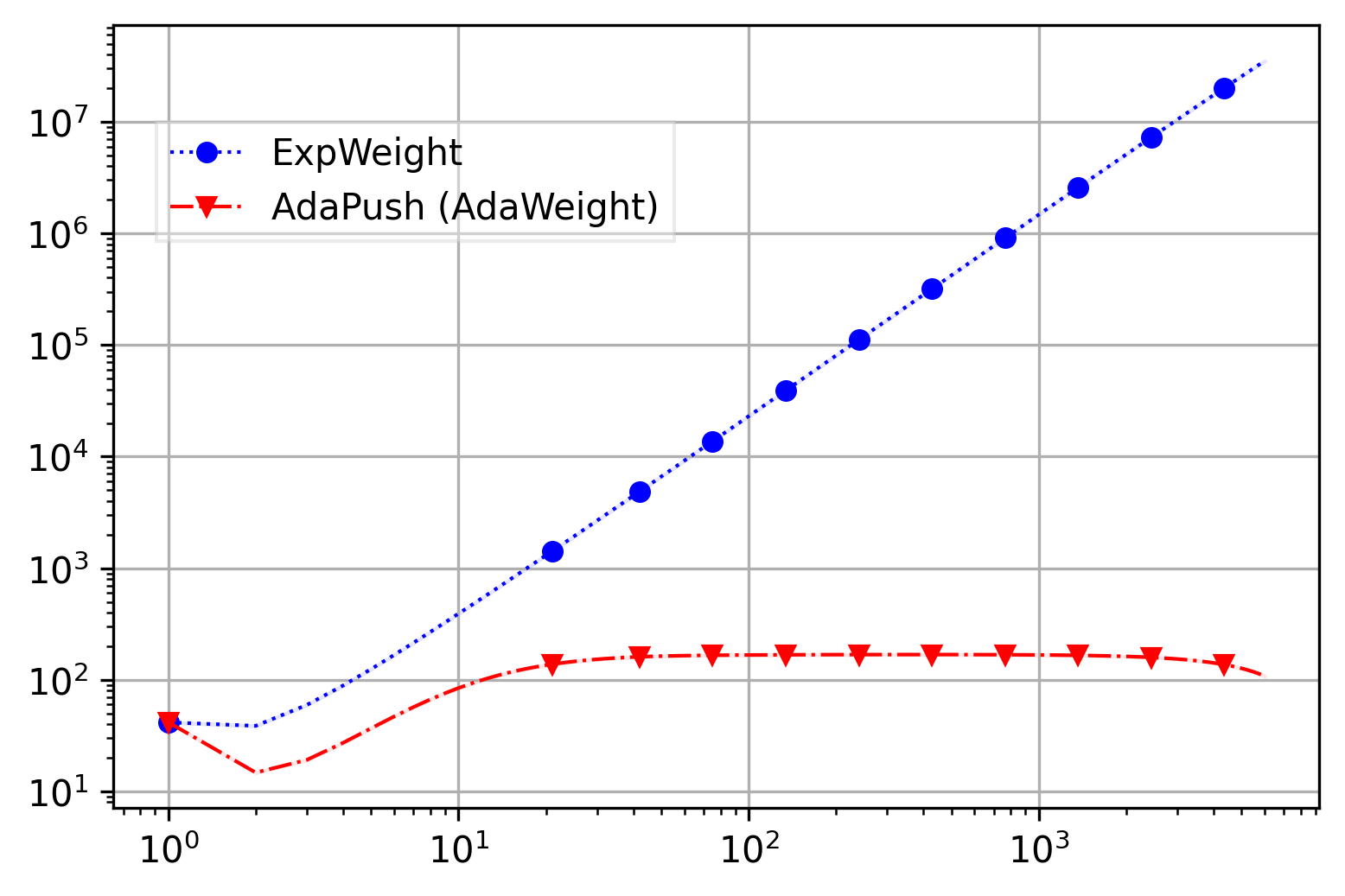}};
			\node[below=of img, node distance=0cm, yshift=\yshift, xshift =\xshift] {\scriptsize${\nRuns}$};
			\node[left=of img, node distance=0cm, rotate=90, anchor=center,xshift=0cm, yshift=\ylabelshift+2mm] {\scriptsize$\nRuns^2 \cdot {\gap(\nRuns)}$};
		\end{tikzpicture}
		\caption{Eastern-Massachusetts: Log-log plot of $\nRuns^2 \cdot \gap(\nRuns)$}\label{fig:EM_static_T2}
	\end{subfigure}


	\begin{subfigure}[b]{\sizesub}
		\begin{tikzpicture}
			\node (img){\includegraphics[height = \figureheight]{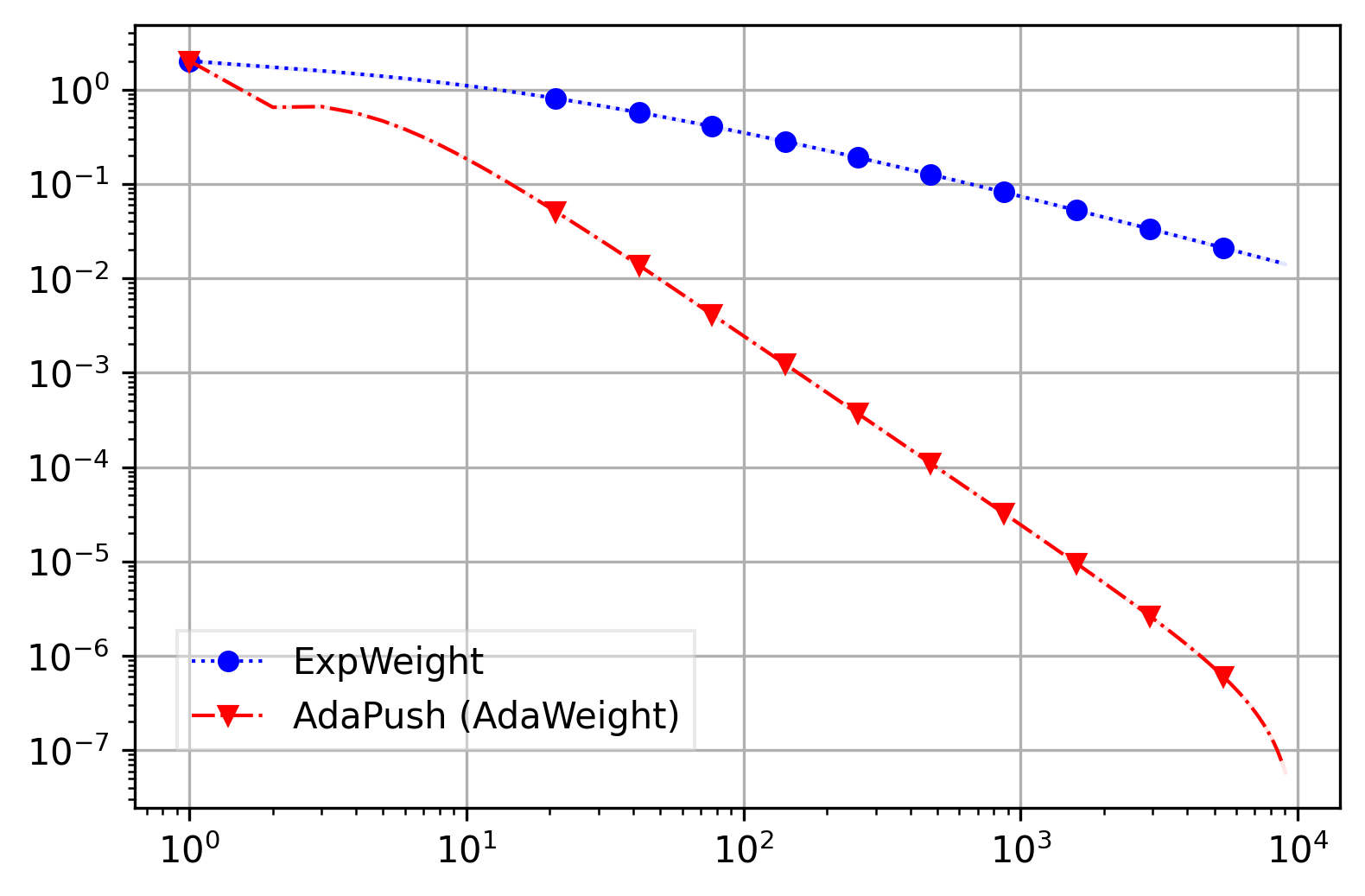}};
			\node[below=of img, node distance=0cm, yshift=\yshift, xshift =\xshift] {\scriptsize${\nRuns}$};
			\node[left=of img, node distance=0cm, rotate=90, anchor=center,xshift=0cm, yshift=\ylabelshift] {\scriptsize${\gap(\nRuns)}$};
		\end{tikzpicture}
		\caption{Berlin-Friedrichshain: Log-log plot of $\gap(\nRuns)$}
		\label{fig:Berlin_F_static}
	\end{subfigure}%
	\begin{subfigure}[b]{\sizesub}
		\begin{tikzpicture}
			\node (img){\includegraphics[height = \figureheight]{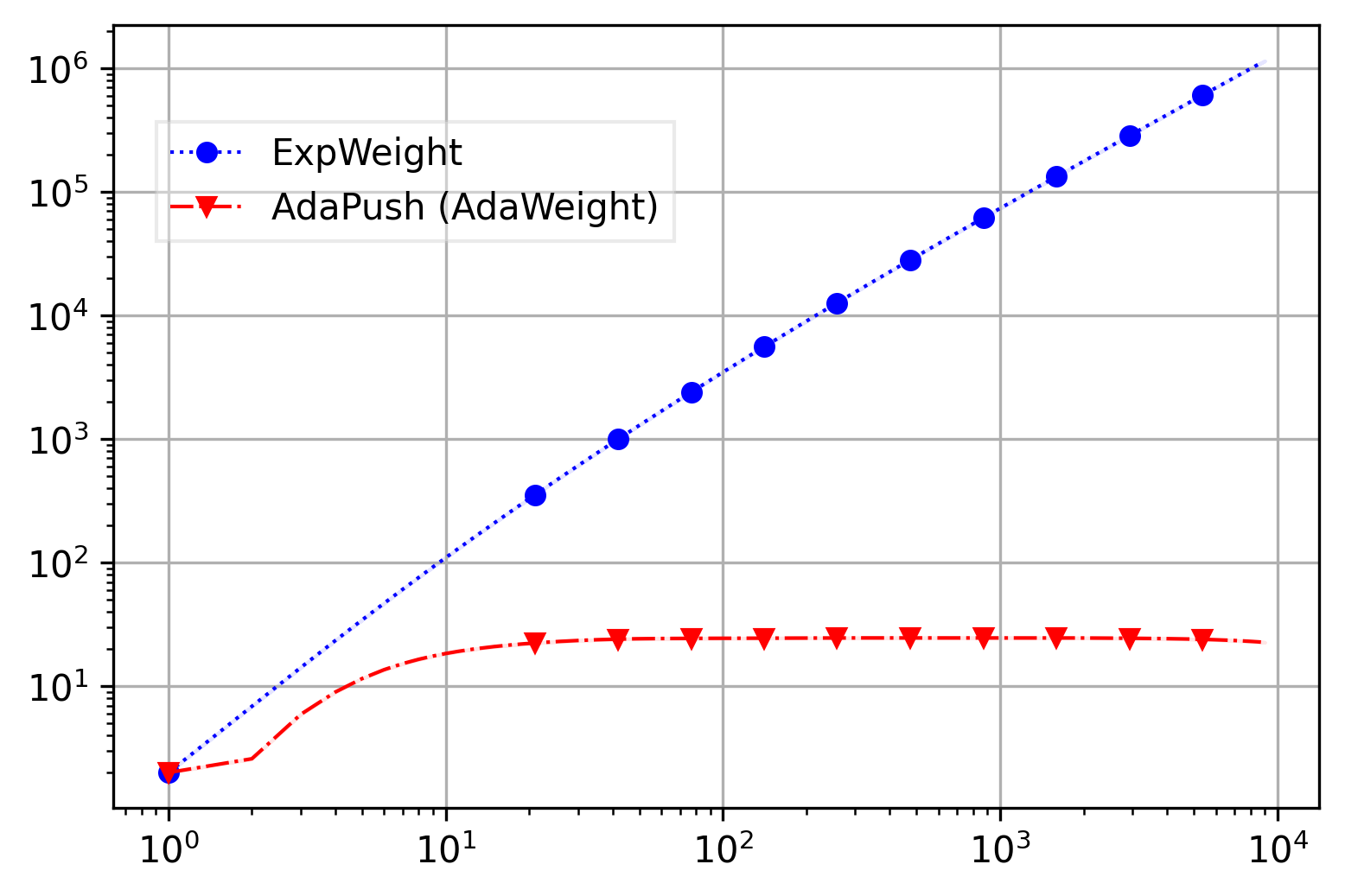}};
			\node[below=of img, node distance=0cm, yshift=\yshift, xshift =\xshift] {\scriptsize${\nRuns}$};
			\node[left=of img, node distance=0cm, rotate=90, anchor=center,xshift=0cm, yshift=\ylabelshift+2mm] {\scriptsize$\nRuns^2 \cdot {\gap(\nRuns)}$};
		\end{tikzpicture}
		\caption{Berlin-Friedrichshain: Log-log plot of $\nRuns^2 \cdot \gap(\nRuns)$}\label{fig:Berlin_F_static_T2}
	\end{subfigure}

	\begin{subfigure}[b]{\sizesub}
		\begin{tikzpicture}
			\node (img){\includegraphics[height = \figureheight]{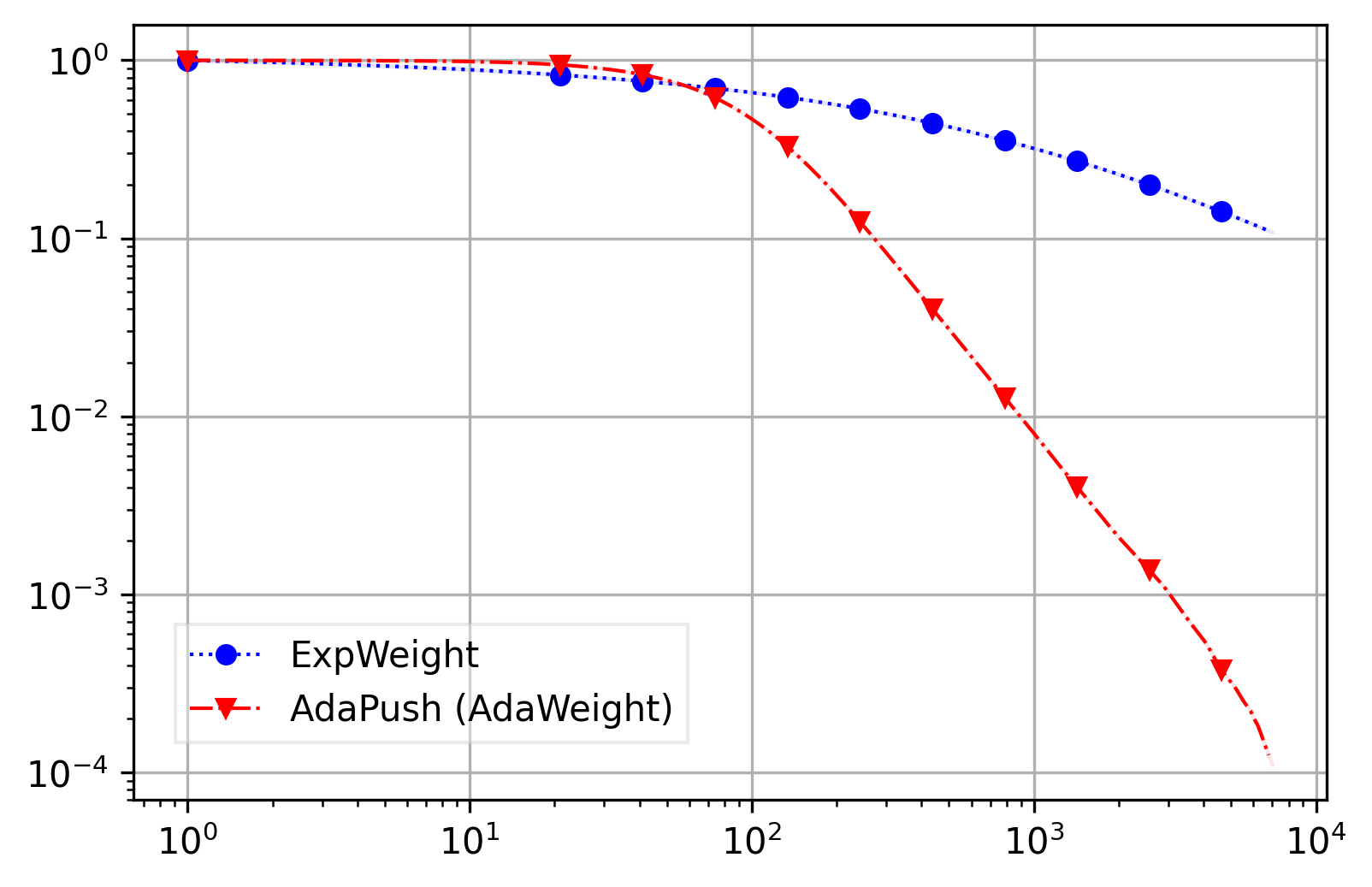}};
			\node[below=of img, node distance=0cm, yshift=\yshift, xshift =\xshift] {\scriptsize${\nRuns}$};
			\node[left=of img, node distance=0cm, rotate=90, anchor=center,xshift=0cm, yshift=\ylabelshift] {\scriptsize${\gap(\nRuns)}$};
		\end{tikzpicture}
		\caption{Anaheim: Log-log plot of $\gap(\nRuns)$}
		\label{fig:Anaheim_static}
	\end{subfigure}%
	\begin{subfigure}[b]{\sizesub}
		\begin{tikzpicture}
			\node (img){\includegraphics[height = \figureheight]{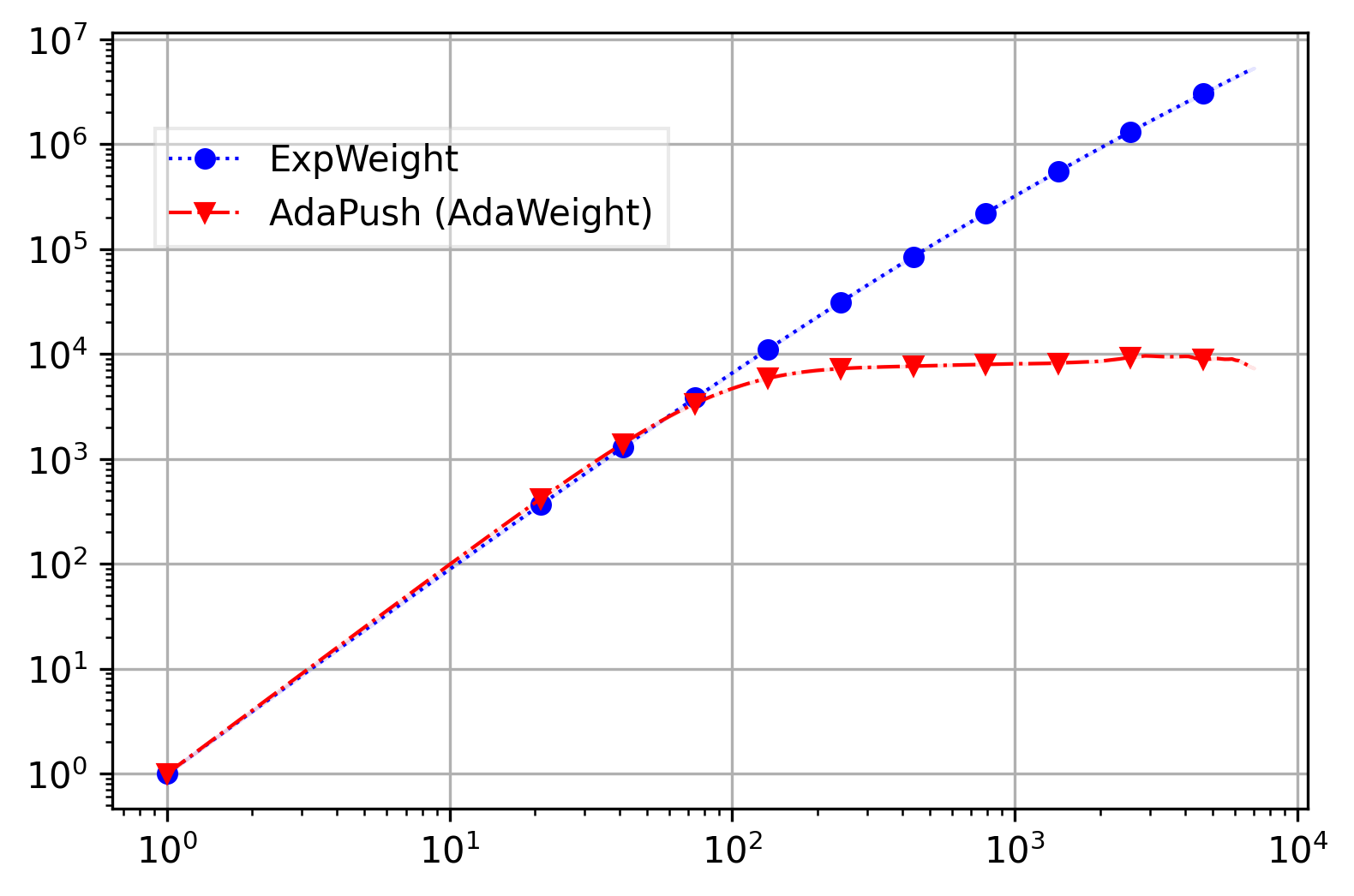}};
			\node[below=of img, node distance=0cm, yshift=\yshift, xshift =\xshift] {\scriptsize${\nRuns}$};
			\node[left=of img, node distance=0cm, rotate=90, anchor=center,xshift=0cm, yshift=\ylabelshift+2mm] {\scriptsize$\nRuns^2 \cdot {\gap(\nRuns)}$};
		\end{tikzpicture}
		\caption{Anaheim: Log-log plot of $\nRuns^2 \cdot \gap(\nRuns)$}\label{fig:Anaheim_static_T2}
	\end{subfigure}

	\vspace{-.5ex}
	\caption{Convergence speed of \adapush (\ac{AEW}) and \ac{EW} in static environments.}
	\vspace{1ex}
	\label{fig:static}
\end{figure}

\begin{figure}[htb!]
	\centering
	\captionsetup{justification=raggedright, singlelinecheck=false}
	\begin{subfigure}[b]{\sizesub}
		\begin{tikzpicture}
			\node (img){\includegraphics[height = \figureheight]{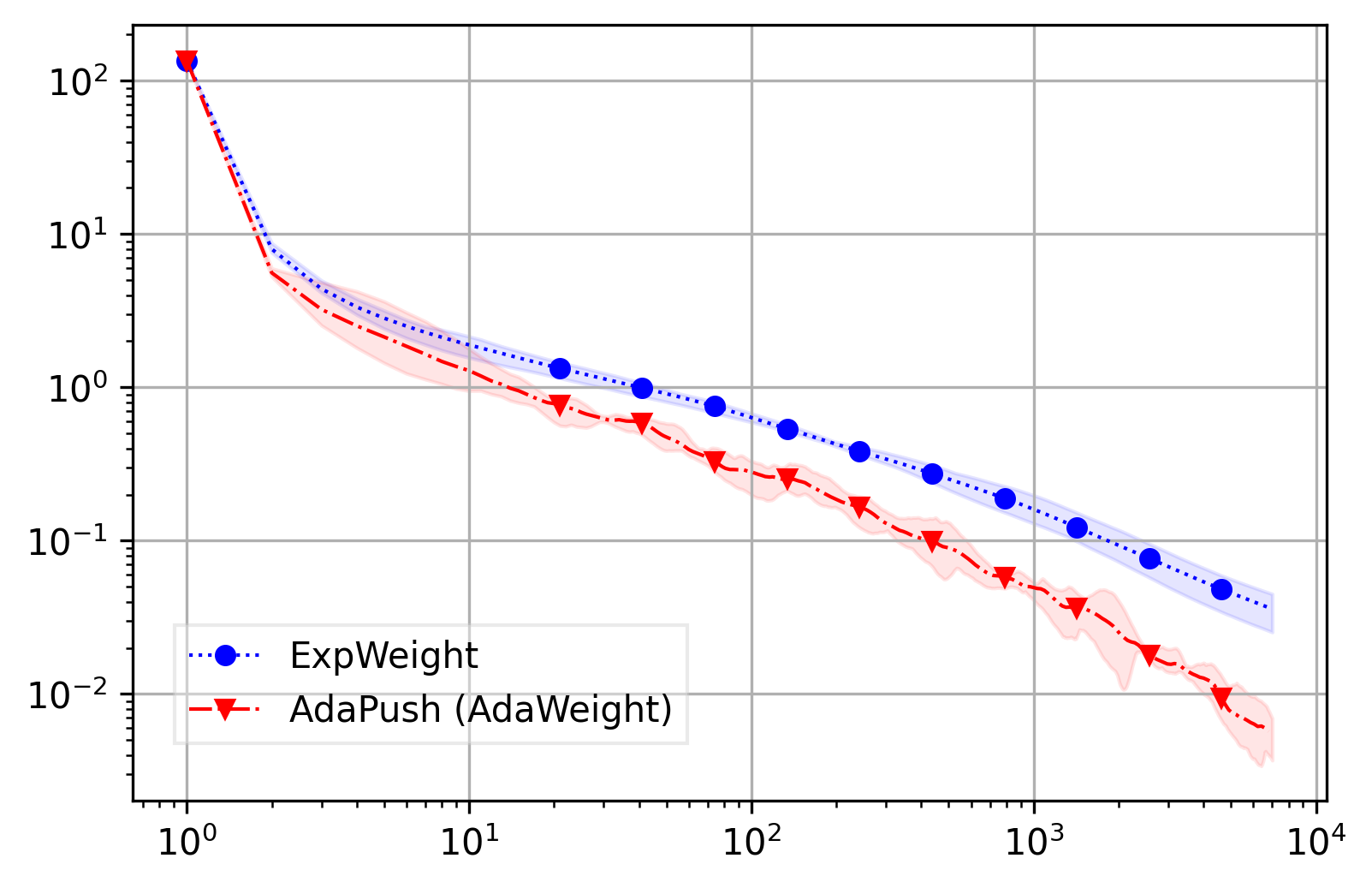}};
			\node[below=of img, node distance=0cm, yshift=\yshift, xshift =\xshift] {\scriptsize${\nRuns}$};
			\node[left=of img, node distance=0cm, rotate=90, anchor=center,xshift=0cm, yshift=\ylabelshift] {\scriptsize${\gap(\nRuns)}$};
		\end{tikzpicture}
		\caption{SiouxFalls: Log-log plot of $\gap(\nRuns)$}
		\label{fig:SiouxFall_stoch}
	\end{subfigure}%
	\begin{subfigure}[b]{\sizesub}
		\begin{tikzpicture}
			\node (img){\includegraphics[height = \figureheight]{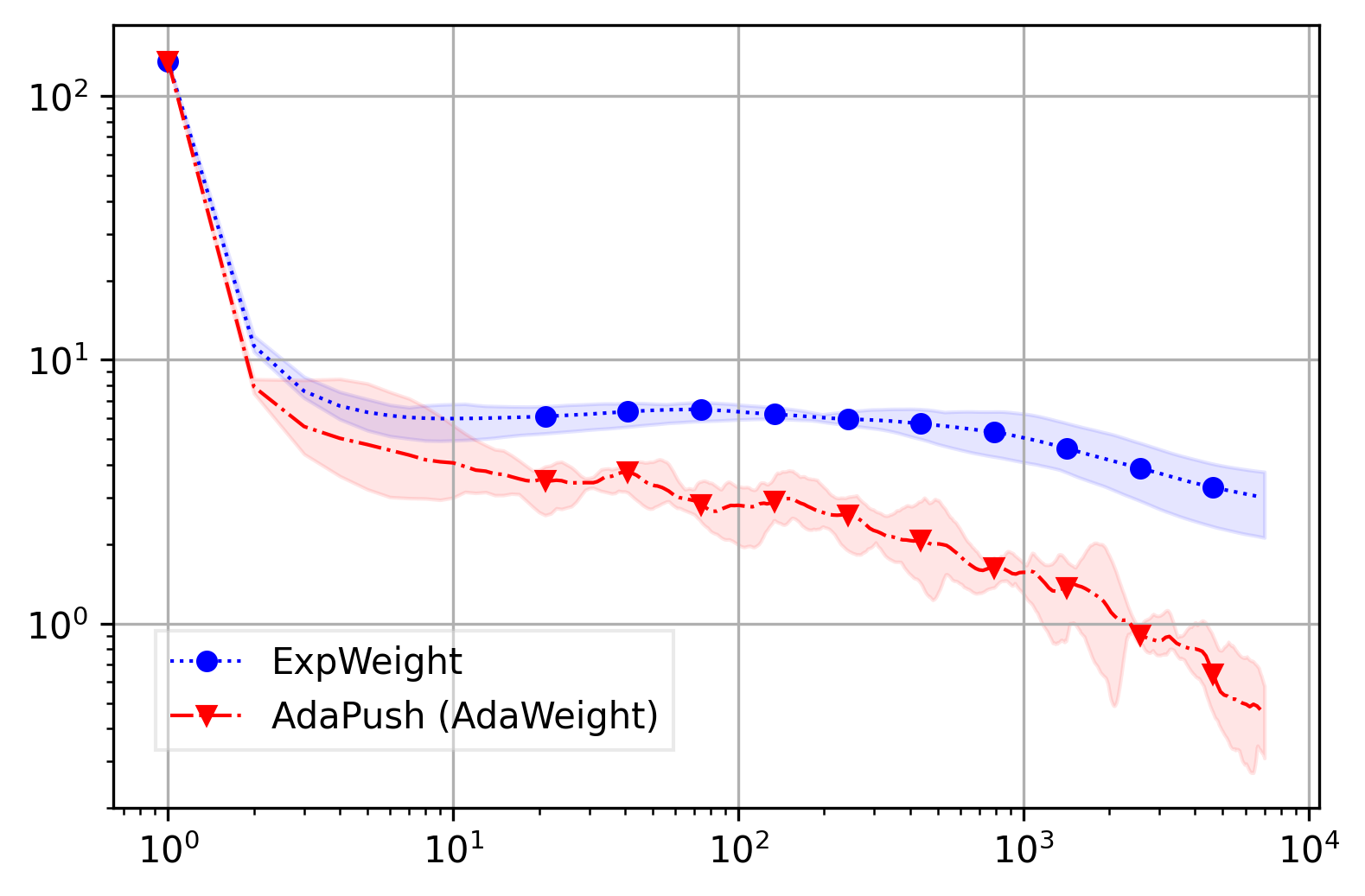}};
			\node[below=of img, node distance=0cm, yshift=\yshift, xshift =\xshift] {\scriptsize${\nRuns}$};
			\node[left=of img, node distance=0cm, rotate=90, anchor=center,xshift=0cm, yshift=\ylabelshift+2mm] {\scriptsize$\sqrt{\nRuns} \cdot {\gap(\nRuns)}$};
		\end{tikzpicture}
		\caption{SiouxFalls: Log-log plot of $\sqrt{\nRuns} \cdot \gap(\nRuns)$}\label{fig:SiouxFall_stoch_T2}
	\end{subfigure}
	
	
	\begin{subfigure}[b]{\sizesub}
		\begin{tikzpicture}
			\node (img){\includegraphics[height = \figureheight]{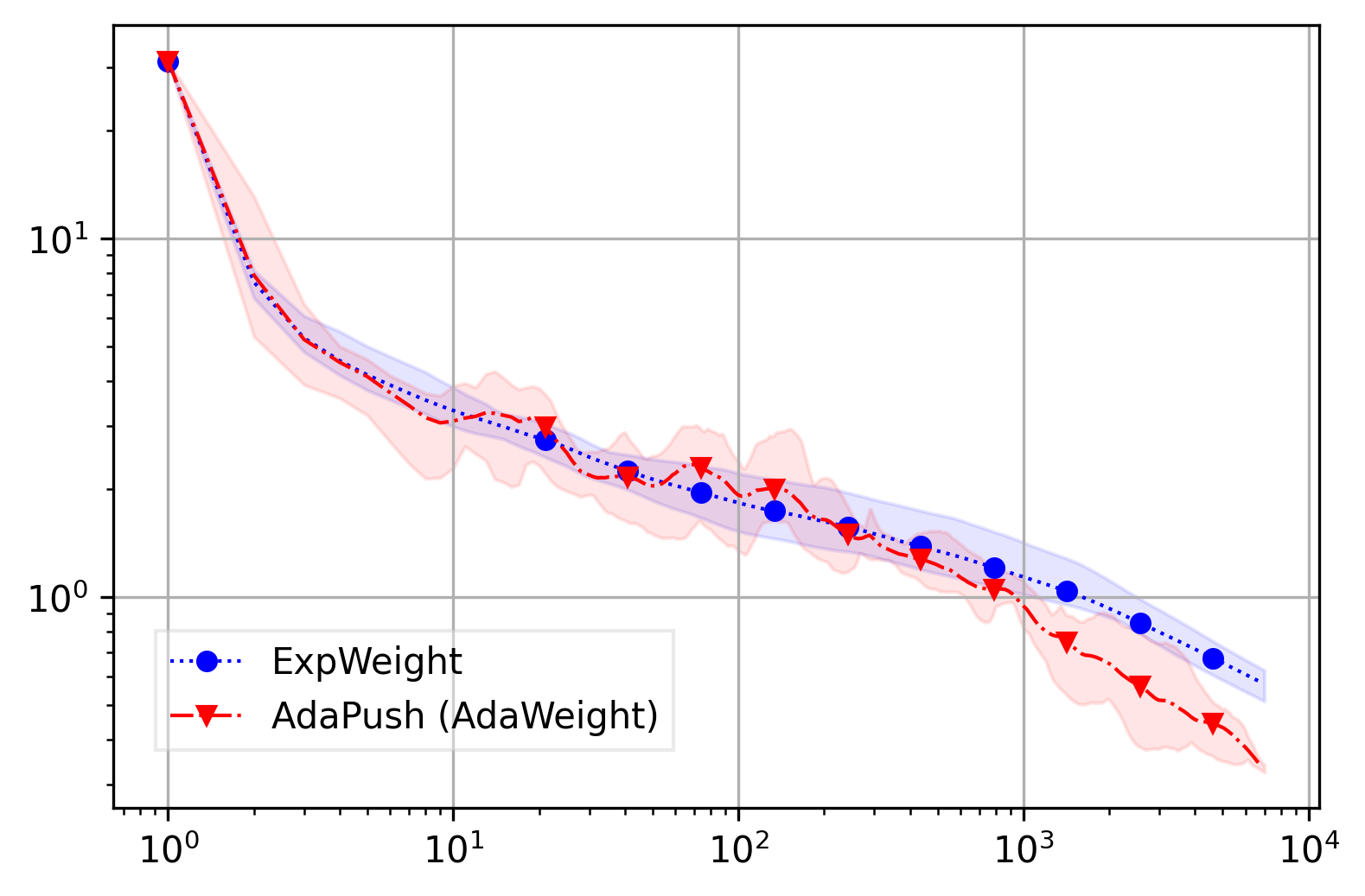}};
			\node[below=of img, node distance=0cm, yshift=\yshift, xshift =\xshift] {\scriptsize${\nRuns}$};
			\node[left=of img, node distance=0cm, rotate=90, anchor=center,xshift=0cm, yshift=\ylabelshift] {\scriptsize${\gap(\nRuns)}$};
		\end{tikzpicture}
		\caption{Eastern-Massachusetts: Log-log plot of $\gap(\nRuns)$}
		\label{fig:EM_stoch}
	\end{subfigure}%
	\begin{subfigure}[b]{\sizesub}
		\begin{tikzpicture}
			\node (img){\includegraphics[height = \figureheight]{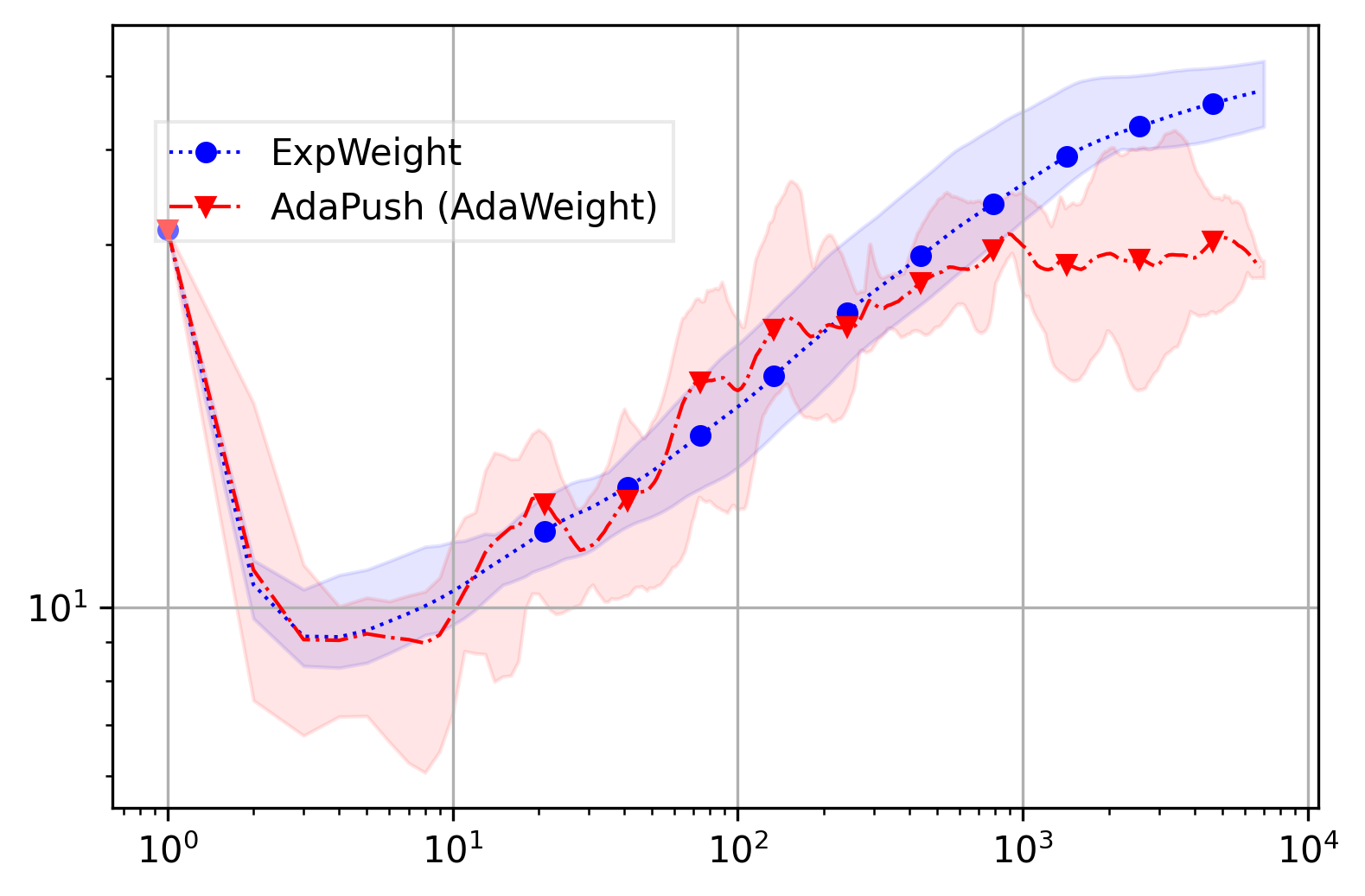}};
			\node[below=of img, node distance=0cm, yshift=\yshift, xshift =\xshift] {\scriptsize${\nRuns}$};
			\node[left=of img, node distance=0cm, rotate=90, anchor=center,xshift=0cm, yshift=\ylabelshift+2mm] {\scriptsize$\sqrt{\nRuns}\cdot {\gap(\nRuns)}$};
		\end{tikzpicture}
		\caption{Eastern-Massachusetts: Log-log plot of $\sqrt{\nRuns} \cdot \gap(\nRuns)$}\label{fig:EM_stoch_T2}
	\end{subfigure}
	

	\begin{subfigure}[b]{\sizesub}
		\begin{tikzpicture}
			\node (img){\includegraphics[height = \figureheight]{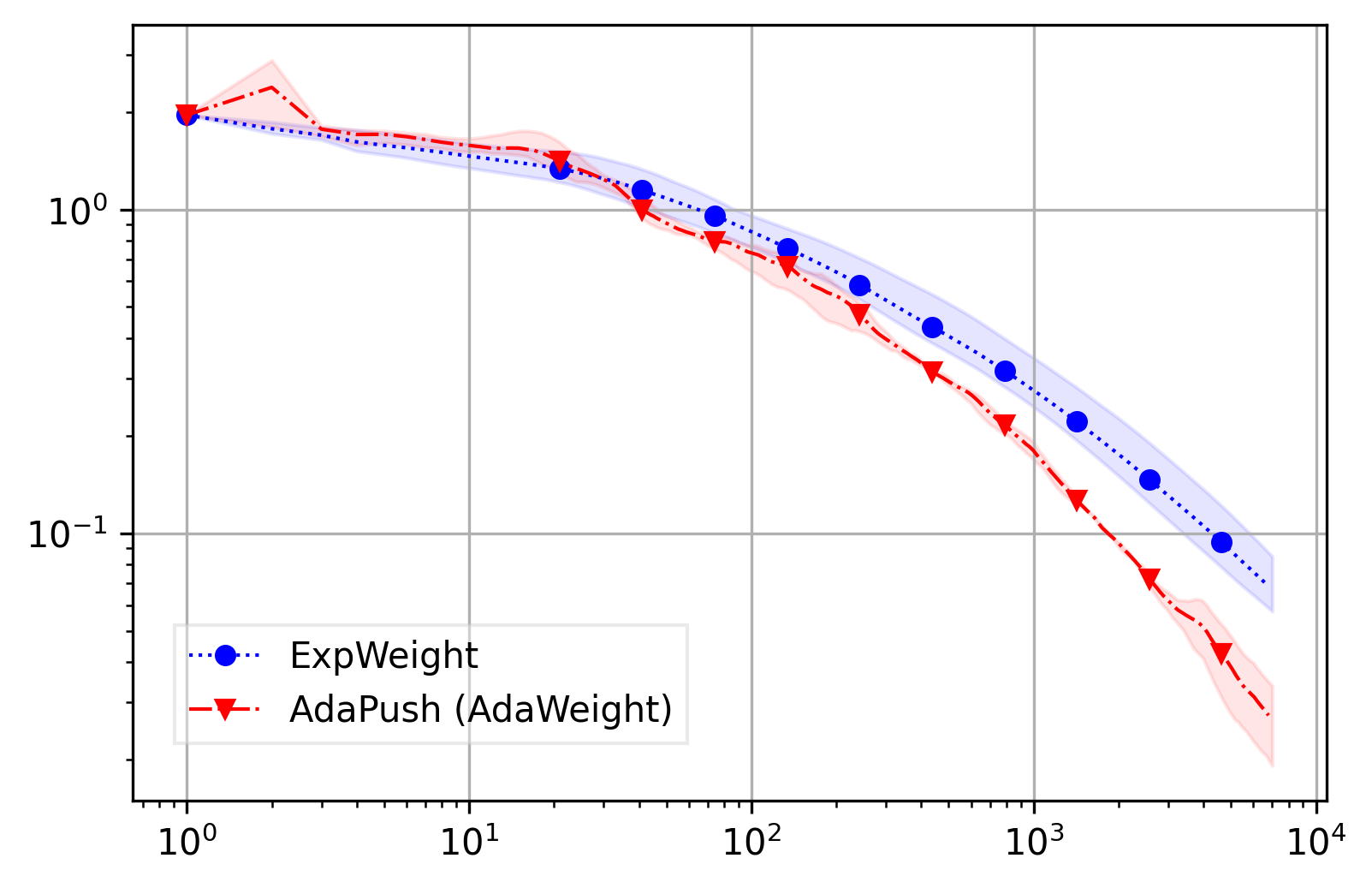}};
			\node[below=of img, node distance=0cm, yshift=\yshift, xshift =\xshift] {\scriptsize${\nRuns}$};
			\node[left=of img, node distance=0cm, rotate=90, anchor=center,xshift=0cm, yshift=\ylabelshift] {\scriptsize${\gap(\nRuns)}$};
		\end{tikzpicture}
		\caption{Berlin-Friedrichshain: Log-log plot of $\gap(\nRuns)$}
		\label{fig:Berlin_F_stoch}
	\end{subfigure}%
	\begin{subfigure}[b]{\sizesub}
		\begin{tikzpicture}
			\node (img){\includegraphics[height = \figureheight]{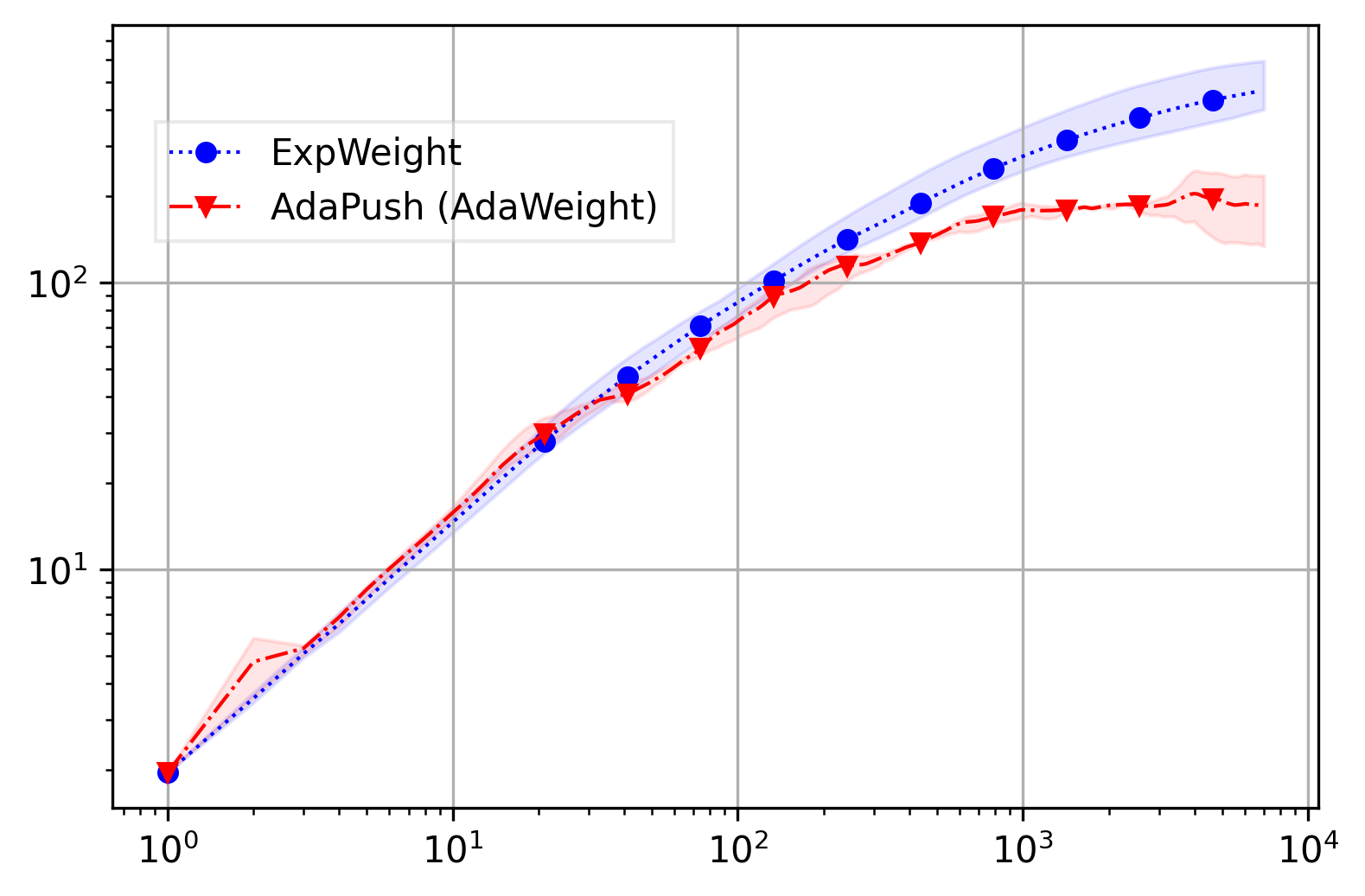}};
			\node[below=of img, node distance=0cm, yshift=\yshift, xshift =\xshift] {\scriptsize${\nRuns}$};
			\node[left=of img, node distance=0cm, rotate=90, anchor=center,xshift=0cm, yshift=\ylabelshift+2mm] {\scriptsize$\sqrt{\nRuns} \cdot {\gap(\nRuns)}$};
		\end{tikzpicture}
		\caption{Berlin-Friedrichshain: Log-log plot of $\sqrt{\nRuns} \cdot \gap(\nRuns)$}\label{fig:Berlin_F_stoch_T2}
	\end{subfigure}

	\begin{subfigure}[b]{\sizesub}
		\begin{tikzpicture}
			\node (img){\includegraphics[height = \figureheight]{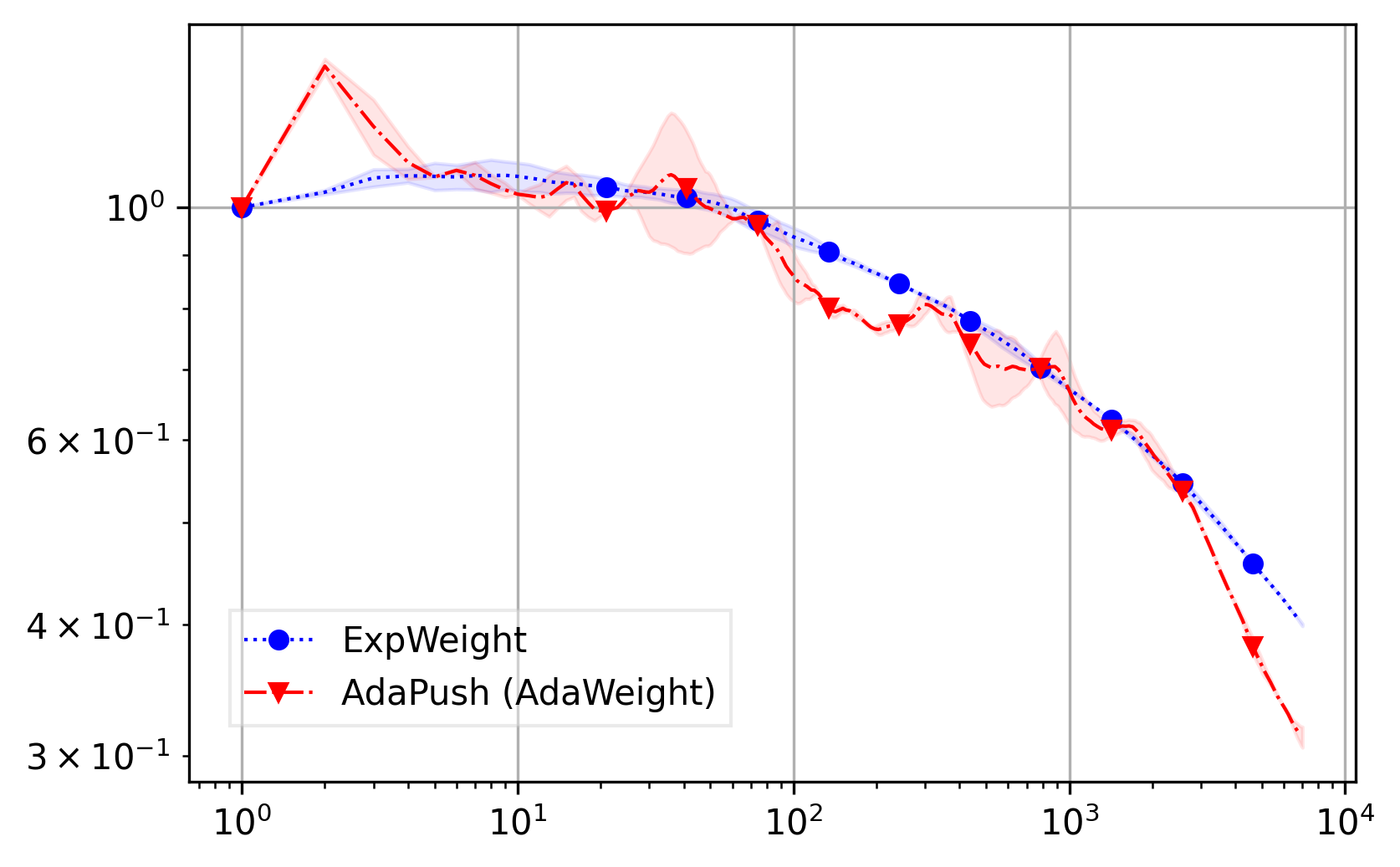}};
			\node[below=of img, node distance=0cm, yshift=\yshift, xshift =\xshift] {\scriptsize${\nRuns}$};
			\node[left=of img, node distance=0cm, rotate=90, anchor=center,xshift=0cm, yshift=\ylabelshift] {\scriptsize${\gap(\nRuns)}$};
		\end{tikzpicture}
		\caption{Anaheim: Log-log plot of $\gap(\nRuns)$}
		\label{fig:Anaheim_stoch}
	\end{subfigure}%
	\begin{subfigure}[b]{\sizesub}
		\begin{tikzpicture}
			\node (img){\includegraphics[height = \figureheight]{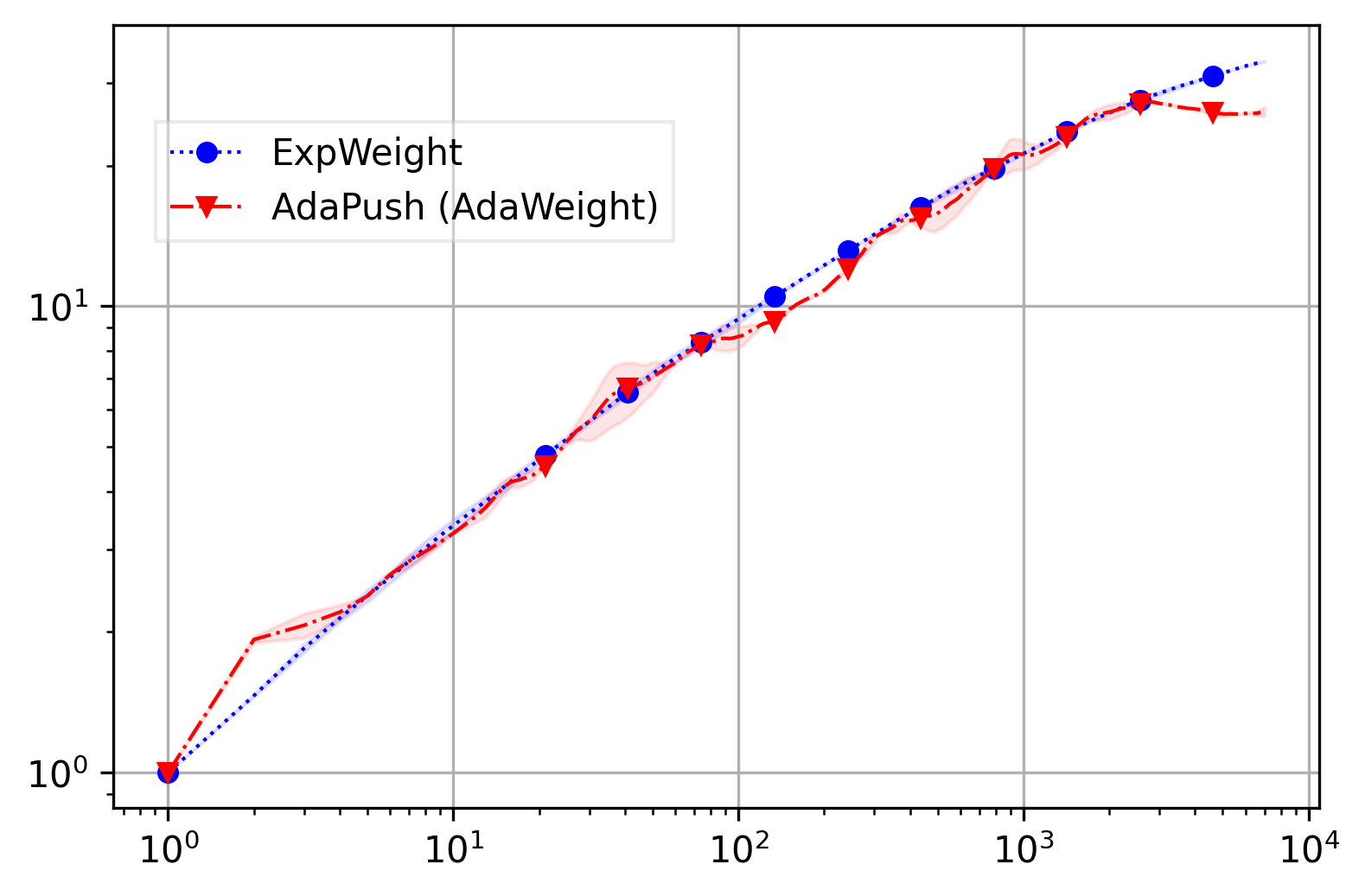}};
			\node[below=of img, node distance=0cm, yshift=\yshift, xshift =\xshift] {\scriptsize${\nRuns}$};
			\node[left=of img, node distance=0cm, rotate=90, anchor=center,xshift=0cm, yshift=\ylabelshift+2mm] {\scriptsize$\sqrt{\nRuns} \cdot {\gap(\nRuns)}$};
		\end{tikzpicture}
		\caption{Anaheim: Log-log plot of $\sqrt{\nRuns} \cdot \gap(\nRuns)$}\label{fig:Anaheim_stoch_T2}
	\end{subfigure}
	
	\vspace{-.5ex}
	\caption{Convergence speed of \adapush (\ac{AEW}) and \ac{EW} in stochastic environments.}
	\vspace{1ex}
	\label{fig:stochastic}
\end{figure}
\vspace{2ex}
\newcounter{app}
\renewcommand{\theapp}{\Alph{app}}

\numberwithin{lemma}{section}		
\numberwithin{proposition}{section}		
\numberwithin{equation}{section}		
\appendix


\section{Proof of \cref{thm:XLEW}}
\label{app:non-adaptive}

\beginPM
Our goal in this appendix is to prove the $\bigof{\log\nRoutes \big/ \nRuns^{2}}$ equilibrium convergence rate of \ac{XLEW} in static environments.
In the following, we use the entropy regularizer $\hreg$ as defined in \cref{sec:proof_adaweight} and let $\curr[\state]$, $\curr[\ppoint]$ and $\curr[\stepalt]$ be defined as per \cref{alg:XLEW}.
With all this in hand, our proof of \cref{thm:XLEW} will proceed in two basic steps:

\para{Step 1: Establish an energy function}

Building on the analysis of accelerated mirror descent algorithms \citep{AZO17,KBB15}, we will consider the energy fuction
\begin{equation}
\label{eq:energy}
\curr[\Delta]
	= \prev[\stepalt] \gap(\run)
		+ \dkl{\sol}{\curr[\ppoint]}
\end{equation}
where $\gap(\run) = \meanpot(\curr) - \min\meanpot$ and $\dkl{\sol}{\curr[\ppoint]}$ denotes the Kullback\textendash Leibler divergence between $\sol$ and $\curr[\ppoint]$.
\endedit

Our goal in the sequel will be to prove that $\curr[\Delta]$ is decreasing in $\run$.
Indeed, from \cref{prop:smooth} and the strong convexity of $\hreg$, we have
\PM{The last equality sign should be ``$\leq$'', correct? [I fixed but pls check]}\DQV{Correct. Comes from strong convexity.}
\begin{align}
\meanpot(\next)
	&\leq \meanpot(\curr[\ppointalt])
		+ \braket{\nabla\meanpot(\curr[\ppointalt])}{\next - \curr[\ppointalt]}
		+ \frac{\smooth}{2} \onenorm{\next - \curr[\ppointalt]}^{2}
	\notag\\
	&= \meanpot(\curr[\ppointalt])
		+ \braket{\nabla\meanpot(\curr[\ppointalt])}{\next - \curr[\ppointalt]}
		+ \frac{\smooth}{2} (1 - \curr[\step])^{2} \onenorm{\next[\ppoint] - \curr[\ppoint]}^{2}
	\notag\\
	&\leq \meanpot(\curr[\ppointalt])
		+ \braket{\nabla\meanpot(\curr[\ppointalt])}{\next - \curr[\ppointalt]}
		+ \smooth\strong (1 - \curr[\step])^{2} \dkl{\next[\ppoint]}{\curr[\ppoint]}. \label{eq:apply_lem_smoothdef}
\end{align}
Moreover, from \eqref{eq:apply_lem_smoothdef} and the convexity of $\meanpot$, for all $\sol\in\sols \defeq \argmin\meanpot$, we have
\PM{I rewrote several parts of the proof below for clarity, pls do an overall read.}
\begin{align}
\label{eq:proof_AMC_inequ_proof_1}
\gap(\run+1) - \curr[\step] \gap(\run)
	&= \meanpot(\next)
		- \bracks{\curr[\step] \meanpot(\curr) + \parens{1 -\curr[\step]} \meanpot(\sol)}
	\notag\\
	&\leq \meanpot(\next)
		- \meanpot(\ppointalt^{\run})
		+ \braket
			{\nabla\meanpot(\ppointalt^{\run})}
			{\curr[\ppointalt] - \curr[\step]\curr - (1- \curr[\step])\sol}
	\notag\\
	&\leq \braket{\nabla\meanpot(\curr[\ppointalt])}{\next - \curr[\ppointalt]}
		+ \smooth\strong (1 - \curr[\step])^{2} \dkl{\next[\ppoint]}{\curr[\ppoint]}
		+ \braket{\nabla\meanpot(\ppointalt^{\run})}{\curr[\ppointalt] - \curr[\step]\curr - (1- \curr[\step])\sol}
	\notag\\
	&= (1 - \curr[\step]) \braket{\nabla\meanpot(\ppointalt^{\run})}{\next[\ppoint] - \sol}
		+ \smooth \strong (1 - \curr[\step])^{2} \dkl {\next[\ppoint]}{\curr[\ppoint]}.
\end{align}
To proceed, let $\curr[\dpointalt] = \parens{1- \curr[\step]} \curr[\stepalt] \nabla\meanpot(\ppointalt^{\run})$.
Then,
\PMedit{from \cref{prop:pot-mean} and}
the update structure of \cref{alg:XLEW}, we have:
\begin{align}
& \braket{\nabla\hreg(\curr[\ppoint]) - \nabla\hreg(\next[\ppoint])}{\next[\ppoint] - \sol} \nonumber \\
	 =& \sum_{\pair \in \pairs}
		\log\parens*{
			\frac{\sum_{\routealt \in \routes^{\pair}}\curr[\ppoint]_{\routealt}\exp(-\curr[\dpointalt]_{\routealt})}{\mass{\pair}}}
			\sum_{\route\in\routes^{\pair}} \parens[\big]{\next[\ppoint]_{\route} - \sol_{\route}}
		+ \sum_{\pair \in \pairs} \sum_{\route \in \routes^{\pair}} \curr[\dpointalt]_{\route} \parens[\big]{\next[\ppoint]_{\route} - \sol_{\route}}
	\hspace{12em}
	\notag\\
	=& 0
		+ \braket{\curr[\dpointalt]}{\next[\ppoint] - \sol}
	\tag*{[since $ \sum_{\route \in \routes^{\pair}} \next[\ppoint]_{ \route} =\sum_{\route \in \routes^{\pair}} \sol_{ \route} = \mass{\pair}$]}
	\\
	= &(1- \curr[\step]) \curr[\stepalt] \braket{\nabla\meanpot(\ppointalt^{\run})}{\next[\ppoint] - \sol}. \label{eq:app_acceleweight_equality}
\end{align}
Thus, multiplying both sides of \eqref{eq:proof_AMC_inequ_proof_1} by $\curr[\stepalt]$ and combining them with \eqref{eq:app_acceleweight_equality}, we obtain
\begin{align}
\curr[\stepalt] \gap(\run+1) - \curr[\step]\curr[\stepalt] \gap(\run)
	&\leq \braket{\nabla\hreg(\curr[\ppoint]) - \nabla\hreg(\next[\ppoint])}{\next[\ppoint] - \sol}
		+ \curr[\stepalt] \strong\smooth (1 - \curr[\step])^{2} \dkl{\next[\ppoint]}{\curr[\ppoint]}
	\notag\\
	&= \dkl{\sol}{\curr[\ppoint]}
		- \dkl{\sol}{\next[\ppoint]}
		+ \parens[\big]{\curr[\stepalt] \strong\smooth (1 - \curr[\step])^{2} - 1}
			\dkl{\next[\ppoint]}{\curr[\ppoint]}.
\label{eq:proof_AMD_conclu}
\end{align}

Now, by the update rule of $\curr[\stepalt]$ in \cref{line:XLEW_stepalt} of \cref{alg:XLEW},
the choice of $\stepalt^{0}$ in \cref{line:XLEW_initial}
and
the update rule of $\curr[\step]$ in \cref{line:XLEW_step},
we get
\begin{equation}
\curr[\stepalt] \strong\smooth (1 - \curr[\step])^{2}
	= \frac{\curr[\stepalt]}{\beforeinit[\stepalt]} \parens*{1 - \frac{\prev[\stepalt]}{\curr[\stepalt]}}^{2}
	= \frac{\parens*{\prev[\stepalt] - \curr[\stepalt]}^{2}}{\curr[\stepalt] \beforeinit[\stepalt]}
	= 1.
\end{equation}
Therefore, the last term in \eqref{eq:proof_AMD_conclu} vanishes and we can rewrite \eqref{eq:proof_AMD_conclu} as
\begin{equation}
\label{eq:energy-decreasing}
\curr[\stepalt] \gap(\run+1) + \dkl{\sol}{\next[\ppoint]}
	\leq \prev[\stepalt] \gap(\run) + \dkl{\sol}{\curr[\ppoint]}.
\end{equation}
This shows that $\curr[\Delta] \leq \prev[\Delta] \leq \dotsm \leq \init[\Delta]$ (by convention, we set $\stepalt^{-1} =0$), \ie $\curr[\Delta]$ is decreasing in $\run$, as claimed.
\hfill
\qed

\para{Step 2: Upper-bounding the equilibrium gap}

By iterating \eqref{eq:energy-decreasing}, we readily obtain
\begin{equation}
\beforelast[\stepalt] \gap(\nRuns) 
	\leq \last[\Delta]
	\leq \dotsm
	\leq \init[\Delta]
	= \dkl{\sol}{\init[\ppoint]}
\label{eq:proof_AMC_last}
\end{equation}
and hence
\begin{equation}
\gap(\nRuns)
	\leq {\dkl{\sol}{\init[\ppoint]}} \big/ {\beforelast[\stepalt]}
\end{equation}
To finish the proof, we need to bound $\beforelast[\stepalt]$ and $\dkl{\sol}{\init}$ from above.

We begin by noting that
\begin{equation}
\sqrt{\strong\smooth \prev[\stepalt]}
	= \sqrt{\strong\smooth \curr[\stepalt]}
		\sqrt{1- \frac{1}{\sqrt{\curr[\stepalt] \strong\smooth}}}
	\leq \sqrt{\strong\smooth \curr[\stepalt]}
		\cdot \parens*{1 - \frac{1}{2 \sqrt{\strong\smooth \curr[\stepalt]}}}
	= \sqrt{\strong\smooth \curr[\stepalt]}
		- \frac{1}{2}
\end{equation}
so $\sqrt{\strong\smooth \curr[\stepalt]} \geq \sqrt{\strong\smooth \prev[\stepalt]} + 1/2$.
Therefore, telescoping this last bound, we get 
\begin{align*}
\sqrt{\strong\smooth \stepalt^{\nRuns-1}}
	\geq \sqrt{\strong\smooth \beforeinit[\stepalt]} + \frac{\nRuns-\start}{2}
	= \frac{\nRuns+\start}{2}
	\geq \frac{\nRuns}{2}
\end{align*}
and hence
\begin{equation}
\stepalt^{\nRuns-1}
	> {\nRuns^{2}} \big/ \parens{4 \strong\smooth} >0. \label{eq:XLEW_lb}
\end{equation}

Second, from the choice of $\beforeinit[\dpoint]$ and $\beforeinit[\step]$ in \cref{line:XLEW_initial} of \cref{alg:XLEW}, it follows that $\init_{\route} = \init[\ppoint]_{\route} = \mass{\pair} \big/ \nRoutes_\pair$ for all $\route \in \routes^{\pair}$, $\pair\in\pairs$.
\PM{There were some index inconsistencies here with $\route$ and $\pair$.
I think I patched things up \textendash\ but pls check\dots}\DQV{done}
We thus get $\braket{\nabla\hreg \parens{\init}}{\init-\sol} = 0$ so $\dkl{\sol}{\init} = \hreg(\sol) - \hreg(\init) \leq \max\hreg - \min\hreg$.
Moreover, by a straightforward calculation, we get $\max\hreg = \masssum \log(\massmax)$ and $\min\hreg = \masssum \log(\nRoutes/\masssum)$, so
\PM{Replaced the ineqs with eqs in this sentence \textendash\ am I wrong?}
\begin{align}
\dkl{\sol}{\init}
	\leq \masssum \log(\nRoutes\massmax/\masssum) .\label{eq:XLEW_ub}
\end{align}
Thus, combining \cref{eq:XLEW_ub,eq:XLEW_lb,eq:proof_AMC_last} and recalling that $\strong \defeq \massmax \nPairs$ \PM{When did we prove this?}
(\ie the strongly-convexity constant of $\hreg$) and the fact that $\masssum = \sum_{\pair \in \pairs} \mass{\pair} \le \nPairs \massmax$ we finally obtain
\begin{equation}
\gap(\nRuns)
	\leq \frac
		{4\smooth \nPairs \massmax \masssum \log\parens*{\nRoutes\massmax \big/ \masssum}}
		{\parens*{\nRuns-\start}^{2}}
	\leq \frac
		{4\smooth \nPairs^{2} \massmax^{2} \log\parens*{\nRoutes \massmax \big/ \masssum}}
		{\parens*{\nRuns-\start}^{2}}
\end{equation}
and our proof of \cref{thm:XLEW} is complete.
\hfill
\qed

\section{Proof of \cref{eq:lem:proof_adaEW}}
\label{app:adaweights}
\label{sec:app_proof_lem_proof_adaEW}
\newmacro{\proofR}{R}

	Let us denote $\proofR^{\run}:= \sum_{\runalt =1}^{\run} \step^\runalt$. By Line~\ref{line:AEW-recom-flow} of Algorithm~\ref{alg:AEW}, for any $\run$, we have: $\curr[\precom] =  \frac{\proofR^{\run}}{\curr[\step]} \state^{\run} - \frac{\proofR^{\run-1}}{\curr[\step]} \state^{\run-1}$. As a consequence,
	\begin{align}
		\sum_{\run=1}^{\nRuns}{\curr[\step]}\inner{\curr[\precom]-\flowbase}{\nabla \meanpot(\state^{\run})} = & \sum_{\run=1}^{\nRuns}{\curr[\step]}\inner*{\frac{\proofR^{\run}}{\curr[\step]} \state^{\run} - \frac{\proofR^{\run-1}}{\curr[\step]} \state^{\run-1} -\flowbase}{\nabla \meanpot(\state^{\run})} \nonumber\\
		%
		%
		=& \sum_{\run=1}^{\nRuns} \bracks*{      \proofR^{\run-1} \inner*{\state^{\run} -\state^{\run-1}}{\nabla \meanpot(\state^{\run})} + \curr[\step] \inner*{\state^{\run} -\flowbase}{\nabla \meanpot(\state^{\run})}     } \nonumber\\
		\ge &   \sum_{\run=1}^{\nRuns} \proofR^{\run-1} \bracks*{\meanpot(\state^{\run}) - \meanpot(\state^{\run-1}) } +  \sum_{\run=1}^{\nRuns}  \curr[\step]\bracks*{\meanpot(\state^{\run}) - \meanpot(\flowbase) }  \nonumber\\
		%
		%
		%
		=& \sum_{\run=1}^{\nRuns} \curr[\step] \bracks*{\meanpot(\state^{\nRuns}) - \meanpot(\flowbase)}.  \label{eq:app_adaEW_LemKavis}
	\end{align}
	Here, the last equality is achieved via telescopic sum. Now, when we choose $\curr[\step] = \run$ (as indicated in \cref{thm:AEW}), we notice that  $\proofR^{\run} > \frac{\nRuns^2}{2}$. Divide two sides of \eqref{eq:app_adaEW_LemKavis} by $\proofR^{\run}$ and taking expectations, we obtain precisely \eqref{eq:lem:proof_adaEW}.
\qed
\bibliographystyle{informs2014}
\bibliography{../bibtex/Bibliography-DQV,../bibtex/Bibliography-PM}

\end{document}